\documentclass[unsortedaddress,superscriptaddress,pra,notitlepage]{revtex4-1}

\usepackage{amsmath,amssymb}
\usepackage{braket}
\usepackage{theorem}
\usepackage{color}
\usepackage{amsfonts}
\usepackage[mathscr]{eucal}
\usepackage[hidelinks]{hyperref}
\usepackage{graphicx}
\usepackage{xcolor}
\usepackage{tikz}

\usetikzlibrary{calc,decorations.pathreplacing}
\usetikzlibrary{arrows,shapes}

\newtheorem{definition}{Definition}[section]
\newtheorem{theorem}[definition]{Theorem}
\newtheorem{prop}[definition]{Proposition}
\newtheorem{lemma}[definition]{Lemma}
\newtheorem{remark}[definition]{Remark}
\newtheorem{rem}[definition]{Remark}

\newtheorem{cor}[definition]{Corollary}

\newtheorem{example}[definition]{Example}

\newenvironment{proof}[1][Proof]{\begin{trivlist}
\item[\hskip \labelsep {\bfseries #1}]}{\hfill$\Box$\end{trivlist}}

\def\theta{\vartheta}
\def\hil{{\mathcal H}}
\def\kil{{\mathcal K}}

\def\B{{\mathcal B}}

\def\C{{\mathcal C}}
\def\D{{\mathcal D}}
\def\E{{\mathcal E}}
\def\F{{\mathcal F}}
\def\G{{\mathcal G}}

\def\L{{\mathcal B}}

\def\P{{\mathcal P}}
\def\S{{\mathcal S}}

\def\X{{\mathcal X}}
\def\Y{{\mathcal Y}}

\def\half{\frac{1}{2}}
\def\iff{\Longleftrightarrow}
\def\imp{\Longrightarrow}
\def\ep{\varepsilon}
\def\bN{\mathbb{N}}
\def\bC{\mathbb{C}}

\def\bR{\mathbb{R}}

\def\bz{\left(}
\def\jz{\right)}

\def\inv{^{-1}}

\def\egy{\mathbf 1}
\def\map{\Phi}

\def\sa{\mathrm{sa}}

\def\what{\widehat}
\def\oll{\overline}

\def\rho{\varrho}

\def\old{}
\def\nw{^{*}}
\def\x{^{(t)}}
\def\xx{(t)}

\def\bog{^{\flat}}
\def\bogg{\flat}
\def\oldd{\{\s\}}
\def\neww{*}

\def\sci{\underline{sc}}
\def\scs{\overline{sc}}
\def\sc{\mathrm{sc}}
\def\di{\underline{d}}
\def\dsup{\overline{d}}

\def\sci{\underline{\mathrm{sc}}}
\def\scs{\overline{\mathrm{sc}}}
\def\minf{\chi}
\def\minfa{\chi_{\alpha}}

\def\fdd{^{[1]}}
\def\cost{\gamma}
\def\costt{c}

\newcommand{\ki}[1]{\textit{\textit{#1}}}

\newcommand{\s}{\mbox{ }}
\newcommand{\ds}{\mbox{ }\mbox{ }}

\newcommand{\norm}[1]{\left\| #1\right\|}

\newcommand{\abs}[1]{\left| #1 \right|}

\newcommand{\vecc}[1]{\underline{#1}}
\newcommand{\dvecc}[1]{\mathbf{#1}}

\newcommand{\diad}[2]{\left|#1\right\rangle\!\left\langle #2\right|}

\newcommand{\pr}[1]{\diad{#1}{#1}}

\newcommand{\new}[1]{{\color{red} #1}}

\newcommand{\derleft}[1]{\partial^{-} #1}
\newcommand{\derright}[1]{\partial^{+} #1}

\newcommand{\wtilde}[1]{\widetilde{#1}}

\newcommand{\ext}[1]{\mathbb{#1}}

\newcommand{\tp}[1]{P_{#1}}
\newcommand{\ceil}[1]{\left\lceil #1\right\rceil}
\newcommand{\vertle}{\rotatebox{90}{$\,\ge$}}

\renewcommand{\theequation}{\thesection.\arabic{equation}}

\makeatletter
\renewcommand{\p@enumii}{}
\makeatother

\DeclareMathOperator{\id}{id}
\DeclareMathOperator{\Tr}{Tr}

\DeclareMathOperator{\Exp}{\mathbb{E}}

\DeclareMathOperator{\supp}{supp}

\DeclareMathOperator{\sgn}{sgn}
\DeclareMathOperator{\ran}{ran}

\DeclareMathOperator{\spec}{spec}

\DeclareMathOperator{\dotimes}{\dot{\otimes}}

\DeclareMathOperator{\conv}{conv}

\DeclareMathOperator{\pin}{\F}
\DeclareMathOperator{\divv}{\Delta}
\DeclareMathOperator{\argmin}{argmin}

\begin{document}

\title{Divergence radii and the strong converse exponent of classical-quantum channel coding with constant compositions}

\author{Mil\'an Mosonyi}
\email{milan.mosonyi@gmail.com}

\affiliation{
MTA-BME Lend\"ulet Quantum Information Theory Research Group
}

\affiliation{
Mathematical Institute, Budapest University of Technology and Economics, \\
Egry J\'ozsef u~1., Budapest, 1111 Hungary.
}

\author{Tomohiro Ogawa}
\email{ogawa@is.uec.ac.jp}
\affiliation{
Graduate School of Informatics and Engineering,
University of Electro-Communications,
1-5-1 Chofugaoka, Chofu-shi, Tokyo, 182-8585, Japan.
}

\begin{abstract}
\centerline{\textbf{Abstract}}
\vspace{.3cm}
\renewcommand{\theequation}{a.\arabic{equation}}

There are different inequivalent ways to define the R\'enyi 
capacity of a channel for a fixed input distribution $P$.
In [IEEE Transactions on Information Theory, 41(1):26-34, 1995],  Csisz\'ar has shown that for classical discrete memoryless channels there is a distinguished such quantity that 
has an operational 
interpretation as a generalized cutoff rate for constant composition channel coding. 
We show that the analogous notion of R\'enyi capacity,
defined in terms of the sandwiched quantum R\'enyi divergences, has the same operational interpretation in the strong converse problem of constant composition classical-quantum channel coding.
Denoting the constant composition strong converse exponent for a memoryless classical-quantum channel $W$ with composition
$P$ and rate $R$ as $\sc(W,R,P)$, our main result is that
\begin{align*}
\sc(W,R,P)=\sup_{\alpha>1}\frac{\alpha-1}{\alpha}\left[R-\chi_{\alpha}\nw(W,P)\right],
\end{align*}
where 
$\chi_{\alpha}\nw(W,P)$ is the 
$P$-weighted sandwiched R\'enyi divergence radius of the image of the channel.
\end{abstract}

\maketitle

\section{Introduction}

A classical-quantum channel
$W:\,\X\to\S(\hil)$ models a device 
with a set of possible inputs $\X$, which, on an input $x\in\X$, outputs 
a quantum system with finite-dimensional Hilbert space $\hil$
in 
state $W(x)$. The channel is called classical if all output states commute, and hence can be represented as probability distributions on some set $\Y$; the interpretation of this is that the channel outputs the symbol 
$y\in\Y$ with probability $(W(x))(y)$.
We will only deal with 
i.i.d. (independent and identically distributed) channels, meaning that
when the channel $W$ is used sequentially, it emits the state $W(\vecc{x})=W(x_1)\otimes\ldots\otimes W(x_n)$ on an input sequence
$\vecc{x}=(x_1,\ldots,x_n)$. 

A channel can be used for information transmission by suitable encoding at the sender's, and decoding at the 
receiver's side, and it is one of the central problems in information theory to devise codes that 
allow the transmission of a large amount of information with a small probability of error, and 
to quantify the ultimate efficiency achievable by any code. The latter task can be best done in an 
asymptotic setting, where the channel is allowed to be used arbitrarily many times, and the aim is to find the best achievable error asymptotics for a given rate $R$ of the amount of transmittable information per channel use. 
The fundamental results of Shannon \cite{Shannon} for classical, and of Holevo \cite{H}
and Schumacher and Westmoreland \cite{SW} for classical-quantum channels, show that there exists a critical rate
$C(W)$, called the Shannon-, resp.~Holevo capacity of the channel, below which reliable information transmission is 
possible in the sense of asymptotically vanishing error probability, and above which it is not.
Moreover, as the results of \cite{BH1998,Hayashicq,N,ON99,W} show, the error probability for an optimal sequence of codes goes to zero exponentially fast for any rate below $C(W)$, and it goes to one exponentially fast for any rate above $C(W)$. 
The optimal achievable exponents for a fixed coding rate $R$ are called the direct exponent $\mathrm{d}(W,R)$ for rates below, and the strong converse exponent $\mathrm{sc}(W,R)$ for rates above $C(W)$. The usefulness of the channel for information transmission is completely characterized once these exponents are given by some essentially computable 
expression for every possible rate $R$.

This problem is famously open for the direct exponent and low rates even for classical channels.
The strong converse exponent, on the other hand, 
has been given for every rate $R$ in \cite{DK} for classical, and in \cite{MO-cqconv} 
for classical-quantum channels (see also \cite{CsiszarKorner,Csiszar} for the classical case). 
It can be expressed as
\begin{align}\label{noconstraint sc}
\mathrm{sc}(W,R)=\sup_{\alpha>1}\frac{\alpha-1}{\alpha}\left[R-\chi_{\alpha}\nw(W)\right],
\end{align}
where $\chi_{\alpha}\nw(W)=\chi_{D_{\alpha}\nw}(W)$ is the sandwiched R\'enyi $\alpha$-capacity of $W$, defined as
\begin{align}\label{capacity equality}
\chi_{\alpha}\nw(W):=\sup_{P\in\P_f(\X)}\chi_{\alpha}\nw(W,P)
=
\sup_{P\in\P_f(\X)}I_{\alpha}\nw(W,P).
\end{align}
Here, $P_f(\X)$ is the set of finitely supported probability distributions on $\X$, and
$\chi_{\alpha}\nw(W,P)$ and $I_{\alpha}\nw(W,P)$ are the $P$-weighted sandwiched R\'enyi divergence radius
and the R\'enyi mutual information of the channel, respectively, the quantities of main interest in our paper.

These notions may be defined more generally by choosing any
quantum divergence $\divv$, i.e., some sort of generalized distance of quantum states, 
as
\begin{align}
I_{\divv}(W,P)&:=\inf_{\sigma\in\S(\hil)}\divv\bz\ext{W}(P)\|P\otimes\sigma\jz,\label{cap1}
\end{align}
($\divv$-mutual information), and
\begin{align}
\chi_{\divv}(W,P)&:=\inf_{\sigma\in\S(\hil)}\sum_{x\in\X}P(x)\divv(W(x)\|\sigma),\label{cap2}
\end{align}
($P$-weighted $\divv$-radius), where 
$P$ is a fixed, finitely supported probability distribution on the input, and 
$\ext{W}(P)=\sum_{x\in\X}P(x)\pr{x}\otimes W(x)$ is the joint state of the classical input and the quantum output.
The sandwiched R\'enyi mutual information and $P$-weighted radius are obtained by choosing $\divv$ 
to be a sandwiched R\'enyi divergence $D_{\alpha}\nw$ \cite{Renyi_new,WWY}.
In the case where $W$ is a classical channel, and $\divv$ is a R\'enyi divergence, these quantities were 
studied by Sibson \cite{Sibson} and Augustin \cite{Augustin}, respectively; see \cite{Csiszar}, 
and the recent works
\cite{NakibogluRenyi,NakibogluAugustin,Cheng-Li-Hsieh2018} for more references on their history and their applications.

The mutual information and the weighted radius provide quantifications of the information transmission capacity of a channel from different perspectives.
Indeed, the mutual information measures the $\divv$-distance of the joint input-output state of the channel from the set of uncorrelated states, with the first marginal fixed. This can be interpreted as a measure of the 
maximal amount of correlation that can be created between the input and the output of the channel with a fixed input distribution. The idea is that the more correlated the input and the output can be made, the more useful the channel is for information transmission.
The weighted radius is geometrically motivated, and the idea behind it is that the farther away some states are in $\divv$-distance (weighted by the input distribution $P$), the more distinguishable they are, and the information transmission capacity of the channel is related to the number of far away states 
among the output states of the channel.

The expression in \eqref{noconstraint sc} provides an operational interpretation for the sandwiched R\'enyi 
$\alpha$-capacities with $\alpha>1$, and singles them out as the operationally relevant quantifiers of the information transmission capacity of classical-quantum channels, from among various other quantities
(e.g., R\'enyi capacities corresponding to other notions of quantum R\'enyi divergences \cite{MO-cqconv}).
It is natural to ask whether the mutual information and the weighted divergence radius 
can be given similar operational interpretations in the context of channel coding.
Note that the standard channel coding problem does not 
yield an answer to this question, 
because after optimization over the input distribution, the 
mutual information and the weighted radius give the same capacity, 
as expressed in \eqref{capacity equality}. Therefore, to settle this problem, one needs to consider a refinement of the channel coding problem where the input distribution $P$ appears on the operational side. This can be achieved by considering constant composition coding, where the codewords are required to have the same empirical distribution for each message, and these empirical distributions are required to converge to a fixed distribution $P$ on $\X$ as the number of channel uses goes to infinity.
It was shown by Csisz\'ar in \cite{Csiszar} that in this setting
\begin{align}\label{Csiszar cc sc}
\mathrm{sc}(W,R,P)=\sup_{\alpha>1}\frac{\alpha-1}{\alpha}\left[R-\chi_{\alpha}(W,P)\right],
\end{align}
for any classical channel $W$ and input distribution $P$, 
where $\mathrm{sc}(W,R,P)$ is the strong converse exponent for coding rate $R$, and 
$\chi_{\alpha}(W,P)=\chi_{D_{\alpha}}(W,P)$. Here, $D_{\alpha}$ is the classical R\'enyi $\alpha$-divergence \cite{Renyi}.
This shows that, maybe somewhat surprisingly, it is not the perhaps more intuitive-looking concept of mutual information
\eqref{cap1} but the geometric quantity \eqref{cap2} that correctly captures the information transmission capacity
of a classical channel. 

Our main result is an exact analogue of \eqref{Csiszar cc sc} for classical-quantum channels, 
given as
\begin{align}\label{main result abstract}
\sc(W,R,P)=\sup_{\alpha>1}\frac{\alpha-1}{\alpha}\left[R-\chi_{\alpha}\nw(W,P)\right],
\end{align}
where 
$\chi_{\alpha}\nw(W,P)$ is the 
$P$-weighted sandwiched R\'enyi divergence radius of the channel.
Thus we establish that, in the strong converse domain, the operationally relevant notion of R\'enyi capacity 
for a classical-quantum channel
with fixed input distribution $P$  is the 
$P$-weighted sandwiched R\'enyi divergence radius of the channel.

The structure of the paper is as follows.
After collecting some technical preliminaries in Section \ref{sec:Preliminaries}, we study the concepts of 
divergence radius and center for general divergences in Section \ref{sec:gendivrad} and for 
the class of $\alpha$-$z$ quantum R\'enyi divergences 
in Section \ref{sec:Renyi center}. 
One of our main results is the additivity of the 
weighted  $\alpha$-$z$ R\'enyi divergence radius for classical-quantum channels and certain pairs 
$(\alpha,z)$, given in Section \ref{sec:additivity}.
We prove it using a representation of the minimizing state in \eqref{cap2} when $\divv=D_{\alpha,z}$
is a quantum $\alpha$-$z$ R\'enyi divergence \cite{AD},
as the fixed point of a certain map on 
the state space. Analogous results have been derived very recently by Nakibo\u glu in \cite{NakibogluAugustin}
for classical channels, and by Cheng, Li and Hsieh in \cite{Cheng-Li-Hsieh2018} for classical-quantum channels 
and the Petz-type R\'enyi divergences. Our results extend these with a different proof method, which in turn is 
closely related to the approach of Hayashi and Tomamichel for proving the additivity of the sandwiched 
R\'enyi mutual information \cite{HT14}.

In Section \ref{sec:sc}, we prove our main result, \eqref{main result abstract}.
The non-trivial part of this is the inequality LHS$\le$RHS, which we prove using a refinement of the arguments 
in \cite{MO-cqconv}. First, in Proposition \ref{prop:sc upper} we employ a suitable adaptation of the techniques 
of Dueck and K\"orner \cite{DK} and 
obtain the inequality in terms of the log-Euclidean R\'enyi divergence, which gives a suboptimal bound. Then in Proposition \ref{prop:upper reg} we use the 
asymptotic pinching technique from \cite{MO-cqconv} to arrive at an upper bound
in terms of the regularized sandwiched R\'enyi divergence radii, and finally we use the previously established 
additivity
property of these quantities to arrive at the desired bound. In the proof of Proposition \ref{prop:sc upper} we need a constant 
composition version of the classical-quantum channel coding theorem.
Such a result was established, for instance, by Hayashi in \cite{universalcq}, and very recently by 
Cheng, Hanson, Datta and Hsieh in \cite{ChengHansonDattaHsieh2018}, with a different exponent, by refining another random coding argument by 
Hayashi \cite{Hayashicq}. We give a slightly modified proof in Appendix \ref{sec:random coding exponent}.
Further appendices contain various technical ingredients of the proofs, and in Appendix \ref{sec:divrad further}
we give a more detailed discussion of the concepts of divergence radius and mutual information for general divergences and $\alpha$-$z$ R\'enyi divergences, which may be of independent interest.

\section{Preliminaries}
\label{sec:Preliminaries}

For a finite-dimensional Hilbert space $\hil$, let $\B(\hil)$ denote the set of all linear operators on $\hil$, 
and let $\B(\hil)_{\sa}$, $\B(\hil)_+$, and $\B(\hil)_{++}$ denote the set of 
self-adjoint, non-zero positive semi-definite (PSD), and positive definite operators, respectively. 
For an interval $J\subseteq\bR$, let 
$\B(\hil)_{\sa,J}:=\{A\in\B(\hil)_{\sa}:\spec(A)\subseteq J\}$, i.e., 
the set of self-adjoint operators on $\hil$ with all their eigenvalues in $J$.
 Let $\S(\hil):=\{\rho\in\B(\hil)_+,\,\Tr\rho=1\}$ denote the set of \ki{density operators}, or \ki{states}, on $\hil$, and $\S(\hil)_{++}$ the set of invertible density operators.

For a self-adjoint operator $A$, let $P^A_a:=\egy_{\{a\}}(A)$ denote the spectral projection of $A$
corresponding to the singleton $\{a\}$. 
The projection onto the support of $A$ is $\sum_{a\ne 0}P^A_a$; in particular, if $A$ is positive semi-definite, it is equal to $\lim_{\alpha\searrow 0}A^{\alpha}=:A^0$. In general, we follow the convention that real powers
of a positive semi-definite operator $A$ are taken only on its support, i.e., for any $x\in\bR$, 
$A^x:=\sum_{a>0}a^x P^A_a$.

For projections $P_1,\ldots,P_r$ on a Hilbert-space $\hil$, we denote by $\bigvee_{i=1}^rP_i$ the projection onto the subspace spanned by $\bigcup_{i=1}^r\ran P_i$.

Given a self-adjoint operator $A\in\B(\hil)_{\sa}$, the \ki{pinching} by $A$ is the operator
$\F_A:\,\B(\hil)\to\B(\hil)$, $\F_A(.):=\sum_a P^A_a(.)P^A_a$, i.e., the block-diagonalization 
with the eigen-projectors of $A$. By the pinching inequality \cite{H:pinching},
\begin{align*}
X\le|\spec(A)|\F_A(X),\ds\ds\ds X\in\B(\hil)_+.
\end{align*}
Since $\F_A$ can be written as a convex combination of unitary conjugations, 
\begin{align}\label{pinching mon}
f\bz\F_A(B)\jz\le \F_A (f(B)),\ds\ds\ds
\Tr g\bz\F_A(B)\jz\le \Tr g(B),
\end{align}
for any operator convex function $f$, and 
any convex function $g$ on an interval $J$, and any $B\in\B(\hil)_{\sa,J}$. The second inequality above is due to the following well-known fact:

\begin{lemma}\label{lemma:trace function properties}
Let $J\subseteq\bR$ be an interval and $f:\,J\to\bR$ be a function. 
\begin{enumerate}
\item
If $f$ is monotone increasing then $\Tr f(.)$ is monotone increasing on $\B(\hil)_{\sa,J}$. 

\item
If $f$ is convex then $\Tr f(.)$ is convex on $\B(\hil)_{\sa,J}$. 
\end{enumerate}
\end{lemma}

For a differentiable function $f$ defined on an interval $J\subseteq\bR$, let 
$f\fdd:\,J\times J\to\bR$ be its \ki{first divided difference function}, defined as
\begin{align*}
f\fdd(a,b):=\begin{cases}
\frac{f(a)-f(b)}{a-b},&a\ne b,\\
f'(a),&a=b,
\end{cases}\ds\ds\ds a,b\in J.
\end{align*}
 
The proof of the following can be found, e.g., in \cite[Theorem V.3.3]{Bhatia} or \cite[Theorem 2.3.1]{Hiai_book}:
\begin{lemma}\label{lemma:op function derivative}
If $f$ is a continuously differentiable function on an open interval $J\subseteq\bR$ then for any finite-dimensional  Hilbert space $\hil$, $A\mapsto f(A)$ is Fr\'echet differentiable on $\B(\hil)_{\sa,J}$, and its 
Fr\'echet derivative $Df(A)$ at a point $A$ is given by 
\begin{align*}
Df(A)(Y)=\sum_{a,b}f\fdd(a,b)P^A_aYP^A_b,\ds\ds\ds Y\in\B(\hil)_{\sa}.
\end{align*}
\end{lemma}

It is straightforward to verify that in the setting of Lemma \ref{lemma:op function derivative}, the function 
$A\mapsto \Tr f(A)$ is also Fr\'echet differentiable on $\B(\hil)_{\sa,J}$, and its 
Fr\'echet derivative $D(\Tr \circ f)(A)$ at a point $A$ is given by 
\begin{align}\label{Tr derivative}
D(\Tr \circ f)(A)(Y)=\Tr f'(A)Y,\ds\ds\ds Y\in\B(\hil)_{\sa},
\end{align}
where $f'$ is the derivative of $f$ as a real-valued function. 
\medskip

An operator $A\in\B(\hil^{\otimes n})$ is \ki{symmetric}, if $U_{\pi}AU_{\pi}^*=A$ for all permutations 
$\pi\in S_n$, where $U_{\pi}$ is defined by 
$U_{\pi}x_1\otimes\ldots\otimes x_n=x_{\pi\inv(1)}\otimes\ldots\otimes x_{\pi\inv(n)}$, $x_i\in\hil$, $i\in[n]$. As it was shown in \cite{universalcq}, for every finite-dimensional Hilbert space $\hil$ and every $n\in\bN$, there exists a 
\ki{universal symmetric state} $\sigma_{u,n}\in\S(\hil^{\otimes n})$ such that it is symmetric, it commutes with every symmetric state, and for every symmetric state $\omega\in\S(\hil^{\otimes n})$,
\begin{align*}
\omega\le v_{n,d}\sigma_{u,n},
\end{align*}
where $v_{n,d}$ only depends on $d=\dim\hil$ and $n$, and it is polynomial in $n$.
\medskip

By a \ki{generalized classical-quantum (gcq) channel} we mean a map $W:\,\X\to\B(\hil)_+$, where $\X$ is a non-empty set, and $\hil$ is a finite-dimensional Hilbert space. It is a \ki{classical-quantum (cq) channel} if 
$\ran W\subseteq\S(\hil)$, i.e., each output of the channel is a normalized quantum state. 
A (generalized) classical-quantum channel is \ki{classical}, if $W(x)W(y)=W(y)W(x)$ for all $x,y\in \X$.
We remark that we do not require any further structure of $\X$ or the map $W$, and in particular, $\X$ need not be finite.
Given a finite number of gcq channels $W_i:\,\X_i\to\B(\hil_i)_+$, their \ki{product} is the gcq channel
\begin{align*}
W_1\otimes\ldots\otimes W_n:\,&\X_1\times\ldots\times \X_n\to\B(\hil_1\otimes\ldots\otimes\hil_n)_+\\
&(x_1,\ldots,x_n)\mapsto W_1(x_1)\otimes\ldots\otimes W_n(x_n).
\end{align*}
In particular, if all $W_i$ are the same channel $W$ then we use the notaion $W^{\otimes n}=W\otimes\ldots\otimes W$.
\medskip

We say that a function $P:\,\X\to[0,1]$ is a \ki{probability density function} on a set $\X$ if
$1=\sum_{x\in\X}P(x):=\sup\left\{\sum_{x\in\X_0}P(x):\,\X_0\subseteq\X\text{ finite}\right\}$. The \ki{support} of $P$ is 
$\supp P:=\{x\in\X:\,P(x)>0\}$. We say that $P$ is finitely supported if $\supp$ is a finite set, and we denote by $\P_f(\X)$ the set of all finitely supported probability distributions. 
The \ki{Shannon entropy} of a $P\in\P_f(\X)$ is defined as
\begin{align*}
H(P):=-\sum_{x\in\X}P(x)\log P(x).
\end{align*}

For a sequence
$\vecc{x}\in\X^n$, the \ki{type} $\tp{\vecc{x}}\in\P_f(\X)$ of $\vecc{x}$ is 
the empirical distribution of $\vecc{x}$, defined as
\begin{align*}
\tp{\vecc{x}}:=\frac{1}{n}\sum_{i=1}^n\delta_{x_i}:\ds y\mapsto
\frac{1}{n}\abs{\{k:\,x_k=y\}},\ds\ds\ds y\in\X,
\end{align*}
where $\delta_x$ is the Dirac measure concentrated at $x$.
We say that a probability distribution $P$ on $\X$ is an \ki{$n$-type} if there exists an $\vecc{x}\in\X^n$ 
such that $P=\tp{\vecc{x}}$. We denote the set of $n$-types by $\P_n(\X)$. For an $n$-type $P$, let 
$\X^n_{P}:=\{\vecc{x}\in\X^n:\,\tp{\vecc{x}}=P\}$ be the set of sequences with the same type $P$.
A key property of types is that $\vecc{x},\vecc{y}\in\X^n$ have the same type if and only if they are permutations of each other, and for any $\vecc{x},\vecc{y}$ with $\tp{\vecc{x}}=\tp{\vecc{y}}$, we have
\begin{align}\label{type prob}
\tp{\vecc{x}}^{\otimes n}(\vecc{y})=e^{-nH(\tp{\vecc{x}})}.
\end{align}
By Lemma 2.3 in \cite{CsiszarKorner2}, for any $P\in\P_n(\X)$, 
\begin{align}\label{type card}
(n+1)^{-|\supp P|}e^{nH(P)}\le |\X^n_{P}|\le e^{nH(P)}.
\end{align}
For any $P\in\P_f(\X)$, any $m\in\bN$, and any $a\in\X$, 
\begin{align}
\sum_{\vecc{x}\in\X^m}P^{\otimes m}(\vecc{x})P_{\vecc{x}}(a)
&=
\sum_{\vecc{x}\in\X^m}P^{\otimes m}(\vecc{x})\frac{1}{m}\sum_{k=1}^m\egy_{\{a\}}(x_k)\nonumber\\
&=
\frac{1}{m}\sum_{k=1}^m\sum_{\vecc{x}\in\X^m}P^{\otimes m}(\vecc{x})\egy_{\{a\}}(x_k)
=
\frac{1}{m}\sum_{k=1}^mP(a)
=P(a).
\label{lemma:type lemma}
\end{align}
\smallskip

The following lemma is an extension 
of the minimax theorems due to Kneser \cite{Kneser} and Fan \cite{Fan}
to the case where $f$ can take the value $+\infty$. For a proof, see 
\cite[Theorem 5.2]{FarkasRevesz2006}.

\begin{lemma}\label{lemma:KF+ minimax}
Let $X$ be a compact convex set in a topological vector space $V$ and $Y$ be a convex
subset of a vector space $W$. Let $f:\,X\times Y\to\bR\cup\{+\infty\}$ be such that
\smallskip

\s(i) $f(x,.)$ is concave on $Y$ for each $x\in X$, and
\smallskip

(ii) $f(.,y)$ is convex and lower semi-continuous  on $X$ for each $y\in Y$.
\smallskip

\noindent Then 
\begin{align}\label{minimax statement}
\inf_{x\in X}\sup_{y\in Y}f(x,y)=
\sup_{y\in Y}\inf_{x\in X}f(x,y),
\end{align}
and the infima in \eqref{minimax statement} can be replaced by minima.
\end{lemma}

\section{Divergence radii}

\subsection{General divergences}
\label{sec:gendivrad}

By a \ki{quantum divergence} $\divv$ we mean a function on pairs of 
non-zero positive semi-definite matrices,
\begin{align*}
\divv:\,\cup_{d\in\bN}\bz\B(\bC^d)_{+}\times \B(\bC^d)_{+}\jz\to\bR\cup\{\pm\infty\},
\end{align*}
that is \ki{invariant under isometries}, i.e., if $V:\,\bC^{d_1}\to\bC^{d_2}$ is an isometry then 
\begin{align*}
\divv(V\rho V^*\|V\sigma V^*)=\divv(\rho\|\sigma),\ds\ds\ds
\rho,\sigma\in\B(\bC^{d_1})_{+}.
\end{align*}
Due to the isometric invariance, $\divv$ may be extended to pairs of non-zero PSD operators
on any finite-dimensional Hilbert space $\hil$, by choosing any isometry $V:\,\hil\to\bC^d$ with large enough $d$, and defining
\begin{align*}
\divv(\rho\|\sigma):=\divv(V\rho V^*\|V\sigma V^*),\ds\ds\ds\rho,\sigma\in\B(\hil)_{+}.
\end{align*}
The isometric invariance property guarantees that this extension is well-defined, in the sense
that the value of $\divv(V\rho V^*\|V\sigma V^*)$
is independent of the choice of $d$ and $V$. 
Clearly, this extension is again invariant under isometries, i.e.,
for any $\rho,\sigma\in\B(\hil)_+$ and $V:\,\hil\to\kil$ isometry, 
$\divv\bz V\rho V^*\|V\sigma V^*\jz=\divv(\rho\|\sigma)$. Note that this implies that $\divv$ is invariant under extensions with pure states, i.e., $\divv(\rho\otimes\pr{\psi}\|\sigma\otimes\pr{\psi})=\divv(\rho\|\sigma)$, where $\psi$ is an arbitrary unit vector in some Hilbert space.
Further properties will often be important. In particular, we say that 
a divergence $\divv$ is 
\begin{itemize}
\item
\ki{positive} if $\divv(\rho\|\sigma)\ge 0$ for all density operators $\rho,\sigma$, and it is 
\ki{strictly positive} if $\divv(\rho\|\sigma)=0\iff\rho=\sigma$, again for density operators;
\item
\ki{monotone under CPTP maps} if for any $\rho,\sigma\in\B(\hil)_+$ and 
any CPTP (completely positive and trace-preserving) map $\map:\,\B(\hil)\to\B(\kil)$,
\begin{align*}
\divv\bz\map(\rho)\|\map(\sigma)\jz
\le
\divv\bz\rho\|\sigma\jz;
\end{align*}
\item
\ki{jointly convex} if for all $\rho_i,\sigma_i\in\B(\hil)$, $i\in[r]:=\{1,\ldots,r\}$, and probability distribution $(p_i)_{i=1}^r$,
\begin{align*}
\divv\bz\sum_{i=1}^r p_i\rho_i\Bigg\|\sum_{i=1}^r p_i\sigma_i\jz
\le
\sum_{i=1}^r p_i\divv\bz\rho_i\|\sigma_i\jz;
\end{align*}
\item
\ki{additive on tensor products} if for any $\rho_i,\sigma_i\in\B(\hil_i)_+$, $i=1,2$,
\begin{align*}
\divv(\rho_1\otimes\rho_2\|\sigma_1\otimes\sigma_2)
=
\divv(\rho_1\|\sigma_1)
+\divv(\rho_2\|\sigma_2).
\end{align*}
\item
\ki{block additive} if for any $\rho_1,\rho_2$, $\sigma_1,\sigma_2$ such that $\rho_1^0\vee\sigma_1^0\perp\rho_2^0\vee\sigma_2^0$, we have
\begin{align*}
\divv(\rho_1+\rho_2\|\sigma_1+\sigma_2)=
\divv(\rho_1\|\sigma_1)+
\divv(\rho_2\|\sigma_2);
\end{align*}
\item
\ki{homogeneous} if 
\begin{align*}
\divv(\lambda\rho\|\lambda\sigma)=\lambda\divv(\rho\|\sigma),\ds\ds\ds
\rho,\sigma\in\B(\hil)_+,\ds\lambda\in(0,+\infty).
\end{align*}
\end{itemize}
Typical examples for divergences with some or all of the above properties are the relative entropy and some
R\'enyi divergences and related quantities; 
see Section \ref{sec:Renyi divrad}.

\begin{rem}
It is well-known \cite{P86,Uhlmann1973} that a block additive and homogenous divergence is monotone under CPTP maps if and only if it is jointly convex. The ``only if'' direction follows by applying monotonicity to 
$\what\rho:=\sum_ip_i\pr{i}_E\otimes\rho_i$ and
$\what\sigma:=\sum_ip_i\pr{i}_E\otimes\sigma_i$ under the partial trace over the $E$ system, where 
$(\ket{i})_{i=1}^r$ is an ONS in $\hil_E$. The ``if'' direction follows by using a Stinespring dilation 
$\map(.)=\Tr_E V(.)V^*$ with an isometry $V:\,\hil\to\kil\otimes\hil_E$, 
and writing the partial trace as a convex combination of unitary conjugations (e.g., by the discrete Weyl unitaries).
\end{rem}
\medskip

Given a non-empty set of positive semi-definite operators $S\subseteq\B(\hil)_+$, its \ki{$\divv$-radius} $R_{\divv}(S)$ is defined as
\begin{align}\label{div radius def}
R_{\divv}(S):=\inf_{\sigma\in\S(\hil)}\sup_{\rho\in S}\divv(\rho\|\sigma).
\end{align}
If the above infimum is attained at some $\sigma\in\S(\hil)$ then $\sigma$ is called a \ki{$\divv$-center} of $S$.
A variant of this notion is when, instead of minimizing the maximal $\divv$-distance,
we minimize an averaged distance according to some
finitely supported probability distribution $P\in\P_f(S)$. This yields the notion of the
\ki{$P$-weighted $\divv$-radius:}
\begin{align}\label{P div rad}
R_{\divv,P}(S):=\inf_{\sigma\in\S(\hil)}\sum_{\rho\in S}P(\rho)\divv(\rho\|\sigma).
\end{align}
If the above infimum is attained at some $\sigma\in\S(\hil)$ then $\sigma$ is called a \ki{$P$-weighted $\divv$-center} for $S$.

\begin{rem}
For applications in channel coding, $S$ will be the image of a classical-quantum channel, and hence a subset of the state space. 
In this case minimizing over density operators $\sigma$ in \eqref{div radius def} and 
\eqref{P div rad} seems natural, while it is less obviously so when the elements of $S$ are general positive semi-definite operators. We discuss this further in Appendix \ref{sec:divrad further}.
\end{rem}

\begin{rem}
Note that for any finitely supported probability distribution $P$ on $\B(\hil)_+$, we have, by definition, $P(\rho)=0$ for all $\rho\in\B(\hil)_+\setminus \supp P$, and hence
\begin{align*}
R_{\divv,P}(\B(\hil)_+)=R_{\divv,P}(\supp P)=R_{\divv,P}(S).
\end{align*}
for any $S\subseteq\B(\hil)_+$ with $\supp P\subseteq S$.
That is, $R_{\divv,P}(S)$ does not in fact depend on $S$, it is a function only of $P$. 
Hence, if no confusion arises, we may simply denote it as $R_{\divv,P}$.
\end{rem}

\begin{rem}
The concepts of the divergence radius and $P$-weighted divergence radius can be unified (to some extent) by the notion of the \ki{$(P,\beta)$-weighted 
$\divv$-radius}, which we explain in Section \ref{sec:general divrad}.
\end{rem}

We will mainly be interested in the above concepts when $S$ is the image of a gcq 
channel $W:\,\X\to\B(\hil)_+$, in which case we will use the notation
\begin{align}
\chi_{\divv}(W,P)&:=R_{\divv,P\circ W\inv}(\ran W)=\inf_{\sigma\in\S(\hil)}\sum_{x\in\X}P(x)\divv(W(x)\|\sigma),\label{Pdivrad}
\end{align}
where $(P\circ W\inv)(\rho):=\sum_{x\in\X:\,W(x)=\rho}P(x)$. Note that, 
as far as these quantities are concerned, the channel simply gives a parametrization of its image set, 
and the previously considered case can be recovered by parametrizing the set by itself, i.e., by taking the gcq channel $\X:=S$ and $W:=\id_{S}$.
We will call \eqref{Pdivrad} the \ki{$P$-weighted $\divv$-radius} of the channel $W$, and any state achieving the infimum in its definition a \ki{$P$-weighted $\divv$-center} for $W$. 
We define the \ki{$\divv$-capacity} of the channel $W$ as
\begin{align}\label{Delta capacity def}
\chi_{\divv}(W)&:=\sup_{P\in\P_f(\X)}\chi_{\divv}(W,P).
\end{align}
In the relevant cases for information theory, the $\divv$-capacity coincides with the $\divv$-radius of the image of the channel, i.e.,
\begin{align*}
\chi_{\divv}(W)=R_{\divv}(\ran W)=\inf_{\sigma\in\S(\hil)}\sup_{x\in\X}\divv(W(x)\|\sigma);
\end{align*}
see Section \ref{sec:general divrad}.

We will mainly be interested in the above quantities when $\divv$ is a quantum R\'enyi divergence. For some further properties of these quantities for general divergences, see Appendix \ref{sec:general divrad}.

\subsection{Quantum R\'enyi divergences}
\label{sec:Renyi divrad}

In this section we specialize to 
various notions of quantum R\'enyi divergences.
For every pair of positive definite operators $\rho,\sigma\in\L(\hil)_{++}$
and every $\alpha\in(0,+\infty)\setminus\{1\}$, $z\in(0,+\infty)$ let
\begin{align*}
Q_{\alpha,z}(\rho\|\sigma):=\Tr\bz\rho^{\frac{\alpha}{2z}}\sigma^{\frac{1-\alpha}{z}}\rho^{\frac{\alpha}{2z}}\jz^z.
\end{align*}
These quantities were first introduced in \cite{JOPP} and further studied in \cite{AD}.
The cases
\begin{align}
Q_{\alpha}(\rho\|\sigma)&:=Q_{\alpha,1}(\rho\|\sigma)=
\Tr \rho^{\alpha}\sigma^{1-\alpha},\label{quasi}\\
Q_{\alpha}\nw(\rho\|\sigma)&:=Q_{\alpha,\alpha}(\rho\|\sigma)=
\Tr \bz \rho^{\half}\sigma^{\frac{1-\alpha}{\alpha}}\rho^{\half}\jz^{\alpha},\label{sand}
\end{align}
and
\begin{align}
Q_{\alpha}\bog(\rho\|\sigma)&:=Q_{\alpha,+\infty}(\rho\|\sigma):=\lim_{z\to+\infty}Q_{\alpha,z}(\rho\|\sigma)=
\Tr e^{\alpha\log\rho+(1-\alpha)\log\sigma}\label{exp}
\end{align}
are of special significance. (The last identity in \eqref{exp} is due to the Lie-Trotter formula.)
Here and henceforth $\xx$ stands for one of the three possible values
$\xx=\oldd,\,\xx=\neww$ or $\xx=\bogg$, where $\oldd$ denotes the empty string, i.e.,
$Q_{\alpha}\x$ with $\xx=\oldd$ is simply $Q_{\alpha}$.

These quantities are extended to general, not necessarily invertible positive semi-definite operators $\rho,\sigma\in\L(\hil)_+$ as
\begin{align}
Q_{\alpha,z}(\rho\|\sigma)
&:=
\lim_{\ep\searrow 0}Q_{\alpha,z}(\rho+\ep I\|\sigma+\ep I)\label{ext0}\\
&=
\lim_{\ep\searrow 0}Q_{\alpha,z}(\rho\|\sigma+\ep I)
=
\lim_{\ep\searrow 0}Q_{\alpha,z}(\rho\|(1-\ep)\sigma+\ep I/d)
=
s(\alpha)\sup_{\ep>0}\oll Q_{\alpha,z}(\rho\|\sigma+\ep I),\label{ext1}
\end{align}
for every $z\in(0,+\infty)$,
where $d:=\dim\hil$, 
\begin{align*}
s(\alpha):=\sgn(\alpha-1)=
\begin{cases}
-1,&\alpha<1,\\
1,&\alpha>1
\end{cases},
\ds\ds\ds\ds\ds\ds\ds\ds\ds
\oll Q_{\alpha,z}:=s(\alpha)Q_{\alpha,z},
\end{align*} 
and the identities are easy to verify.
For $z=+\infty$, the extension is defined by \eqref{ext0}; see \cite{HP_GT,MO-cqconv} for details.

Various further divergences can be defined from the above quantities.
The \ki{quantum $\alpha$-$z$ R\'enyi divergences} \cite{AD} are defined as
\begin{align}\label{Renyi div def}
D_{\alpha,z}(\rho\|\sigma):=\frac{1}{\alpha-1}\log \frac{Q_{\alpha,z}(\rho\|\sigma)}{\Tr\rho}
\end{align}
for any $\alpha\in(0,+\infty)\setminus\{1\}$ and $z\in(0,+\infty]$.
It is easy to see that 
\begin{align*}
\alpha>1,\ds\ds\rho^0\nleq\sigma^0\ds\imp\ds Q_{\alpha,z}=D_{\alpha,z}=+\infty
\end{align*}
for any $z$.
Moreover, if $\alpha\mapsto z(\alpha)$ is continuously differentiable in a neighbourhood of $1$, 
on which $z(\alpha)\ne 0$, or $z(\alpha)=+\infty$ for all $\alpha$, then, according to 
\cite[Theorem 1]{LinTomamichel15} and \cite[Lemma 3.5]{MO-cqconv},
\begin{align*}
D_1(\rho\|\sigma):=
\lim_{\alpha\to1}D_{\alpha,z(\alpha)}(\rho\|\sigma)=\frac{1}{\Tr\rho}D(\rho\|\sigma)
=:D_{1,z}(\rho\|\sigma),\ds\ds\ds z\in(0,+\infty],
\end{align*}
where $D(\rho\|\sigma)$ is Umegaki's relative entropy \cite{Umegaki}, defined as
\begin{align*}
D(\rho\|\sigma):=\Tr\rho(\log\rho-\log\sigma)
\end{align*}
for positive definite operators, and extended
as above for non-zero positive semidefinite operators.
Of the R\'enyi divergences corresponding to the special $Q_{\alpha}$ quantities discussed above,
$D_{\alpha}$ is usually called the \ki{Petz-type} R\'enyi divergence,
$D_{\alpha}\nw$ the \ki{sandwiched} R\'enyi divergence \cite{Renyi_new,WWY},
and $D_{\alpha}\bog$ the \ki{log-Euclidean} R\'enyi divergence. For more on the above definitions and a more detailed reference to their literature, see, e.g., \cite{MO-cqconv}.
We will also use the \ki{max-relative entropy} \cite{Datta,Renyi_new,RennerPhD}:
\begin{align*}
D_{\infty}\nw(\rho\|\sigma):=\lim_{\alpha\to+\infty}D_{\alpha}\nw(\rho\|\sigma)
=
\log\inf\{\lambda>0:\,\rho\le \lambda\sigma\}.
\end{align*}

To discuss some important properties of the above quantities, let us introduce the following regions of the $\alpha$-$z$ plane:
\begin{align*}
& K_0:\,0<\alpha<1,\,z<\min\{\alpha,1-\alpha\}; & & K_1:\,0<\alpha<1,\,\alpha\le z\le 1-\alpha;\\
&K_2:\,0<\alpha<1,\,\max\{\alpha,1-\alpha\}\le z\le 1; & & K_3:\,0<\alpha<1,\,1-\alpha\le z\le \alpha;\\
& K_4:\,0<\alpha<1,\,1\le z; & & K_5:\,1<\alpha,\,\alpha/2\le z\le 1; \\
&K_6:\,1<\alpha,\,\max\{\alpha-1,1\}\le z\le \alpha; & &
K_7:\,1<\alpha\le z;
\end{align*}

The $(\alpha,z)$ values for which $D_{\alpha,z}$ is monotone under CPTP maps have been completely characterized in
\cite{Hi3,An,Beigi,FL,CFL,Zhang2018} (cf.~also \cite[Theorem 1]{AD}). This can be summarized as follows.
\begin{lemma}\label{lemma:az monotonicity}
$D_{\alpha,z}$ is monotone under CPTP maps
$\iff$ $\oll Q_{\alpha,z}$ is monotone under CPTP maps
$\iff$ $\oll Q_{\alpha,z}$ is jointly convex
$\iff$ $(\alpha,z)\in K_2\cup K_4\cup K_5\cup K_6$.
\end{lemma}

\begin{cor}\label{cor:az joint convexity}
$D_{\alpha,z}$ is jointly convex if $(\alpha,z)\in K_2\cup K_4$.
\end{cor}
\begin{proof}
Immediate from Lemma \ref{lemma:az monotonicity}, as the joint convexity of $\oll Q_{\alpha,z}$ implies the joint convexity of $D_{\alpha,z}=\frac{1}{\alpha-1}\log s(\alpha)\oll Q_{\alpha,z}$ whenever $\alpha\in(0,1)$.
\end{proof}

Recall that a function $f:\,C\to\bR\cup\{+\infty\}$ on a convex set $C$ is \ki{quasi-convex} if 
$f((1-t)x+ty)\le\max\{f(x),f(y)\}$ for all $x,y\in C$ and $t\in[0,1]$.

\begin{lemma}\label{lemma:2nd convexity}
On top of the cases discussed in Lemma \ref{lemma:az monotonicity} and Corollary \ref{cor:az joint convexity}, 
$D_{\alpha,z}$ is convex in its second argument if $(\alpha,z)\in K_3\cup K_6\cup K_7$, and 
$\oll Q_{\alpha,z}$ is convex in its second argument if
$(\alpha,z)\in  K_3\cup K_7$. Moreover, $D_{\alpha,z}$ is jointly quasi-convex if 
$(\alpha,z)\in K_5$.
\end{lemma}
\begin{proof}
The assertion about the quasi-convexity of $D_{\alpha,z}$ is immediate from the joint convexity of 
$\oll Q_{\alpha,z}$ when $(\alpha,z)\in K_5$.

Note that it is enough to prove convexity in the second argument for positive definite operators, 
due to \eqref{ext1}.

Assume that $(\alpha,z)\in K_2\cup K_3$, i.e., $0<\alpha<1$, $1-\alpha\le z\le 1$. Then 
$0<\frac{1-\alpha}{z}\le 1$, and hence $\sigma\mapsto \sigma^{\frac{1-\alpha}{z}}$ is concave. Since
$\B(\hil)_+\ni A\mapsto\Tr A^z$ is both monotone and concave (see Lemma \ref{lemma:trace function properties}), we get that 
$\sigma\mapsto Q_{\alpha,z}(\rho\|\sigma)$ is concave, from which the 
convexity of both $\oll Q_{\alpha,z}$ and $D_{\alpha,z}$ in their second argument follows for 
$(\alpha,z)\in K_3$ (and also for $K_2$, although that is already covered by joint convexity).

Assume next that $(\alpha,z)\in K_6\cup K_7$, i.e., $1<\alpha$, and $\max\{1,\alpha-1\}\le z$. Then 
$-1\le\frac{1-\alpha}{z}<0$, and hence $f:\,t\mapsto t^{\frac{1-\alpha}{z}}$ is a non-negative operator monotone decreasing function on $(0,+\infty)$. Applying the duality of the Schatten $p$-norms to $p=z$, we have 
\begin{align*}
D_{\alpha,z}(\rho\|\sigma)=\sup_{\tau\in\S(\hil)}\frac{z}{\alpha-1}\log
\Tr\rho^{\frac{\alpha}{2z}}\sigma^{\frac{1-\alpha}{z}}\rho^\frac{\alpha}{2z}\tau^{1-\frac{1}{z}}
=
\sup_{\tau\in\S(\hil)}\frac{z}{\alpha-1}\log\omega_{\tau}\bz f(\sigma)\jz,
\end{align*}
where $\omega_{\tau}(.):=\Tr\rho^{\frac{\alpha}{2z}}(.)\rho^\frac{\alpha}{2z}\tau^{1-\frac{1}{z}}$
is a positive functional.
By \cite[Proposition 1.1]{AH}, $D_{\alpha,z}(\rho\|.)$ is the supremum of convex functions 
on $\B(\hil)_{++}$, and hence is itself convex. 
This immediately implies that $\oll Q_{\alpha,z}$ is convex in its second argument when 
$(\alpha,z)\in K_6\cup K_7$ (of which the case $K_6$ also follows from joint convexity).
\end{proof}

\begin{lemma}\label{lemma:lsc}
For any fixed $\rho\in\B(\hil)_+$, the maps
\begin{align*}
\sigma\mapsto \oll Q_{\alpha,z}(\rho\|\sigma)\ds\ds\text{and}\ds\ds
\sigma\mapsto D_{\alpha,z}(\rho\|\sigma)
\end{align*}
are lower semi-continuous on $\B(\hil)_+$ for any $\alpha\in(0,+\infty)\setminus\{1\}$ and $z\in(0,+\infty)$, and for $z=+\infty$ and $\alpha>1$.
\end{lemma}
\begin{proof}
The cases $\alpha\in(0,+\infty)\setminus\{1\}$ and $z\in(0,+\infty)$ are obvious from the last expression in \eqref{ext1}, and the case $z=+\infty$ was discussed in \cite[Lemma 3.27]{MO-cqconv}.
\end{proof}

It is known that $D_{\alpha}$, $D_{\alpha}^*$ and $D_{\alpha}\bog$ are non-negative on pairs of states
\cite{Renyi_new,P86,MO-cqconv}, but it seems that the non-negativity of general $\alpha$-$z$ R\'enyi divergences 
has not been analyzed in the literature so far. We show in Appendix \ref{sec:positive Renyi} that they are indeed non-negative for any pair of parameters $(\alpha,z)$.

\subsection{The R\'enyi divergence center}
\label{sec:Renyi center}

Let $W:\,\X\to\B(\hil)_+$ be a gcq channel. Specializing to $\divv=D_{\alpha,z}$ in \eqref{Pdivrad} yields the 
\ki{$P$-weighted R\'enyi $(\alpha,z)$ radii} of the channel for a finitely supported input probability distribution $P\in\P_f(\X)$,
\begin{align}\label{weighted alpha divrad def}
\chi_{\alpha,z}(W,P)
:=
\chi_{D_{\alpha,z}}(W,P)
=
\inf_{\sigma\in\S(\hil)}\sum_{x\in\X}P(x)D_{\alpha,z}(W(x)\|\sigma)
=
\min_{\sigma\in\S(\hil)}\sum_{x\in\X}P(x)D_{\alpha,z}(W(x)\|\sigma).
\end{align}
The existence of the minimum is guaranteed by the lower semi-continuity stated in Lemma \ref{lemma:lsc}.
We will call any state $\sigma$ achieving the minimum in \eqref{weighted alpha divrad def} a
\ki{$P$-weighted $D_{\alpha,z}$ center} for $W$.
In direct connection with the notations introduced in \eqref{quasi}--\eqref{exp}, we will use the notations
\begin{align*}
\chi_{\alpha}(W,P):=\chi_{\alpha,1}(W,P),\ds\ds\ds
\chi_{\alpha}\nw(W,P):=\chi_{\alpha,\alpha}(W,P),\ds\ds\ds
\chi_{\alpha}\bog(W,P):=\chi_{\alpha,+\infty}(W,P),
\end{align*}
and for $\alpha=1$,
\begin{align*}
\chi(W,P):=\chi_{1}(W,P):=\chi_{1,z}(W,P),\ds\ds\ds z\in(0,+\infty),
\end{align*}
which is just a generalization of the \ki{Holevo quantity} for an ensemble $\{W(x),P(x)\}_{x\in\supp P}$ ,where the $W(x)$ need not be normalized. Moreover, we will use 
\begin{align*}
\chi_{\infty}\nw(W,P):=\inf_{\sigma\in\S(\hil)}\sum_{x\in\X}P(x)D_{\infty}\nw(W(x)\|\sigma),
\end{align*}
and call it the $P$-weighted max-R\'enyi radius. It is known (see, e.g., \cite[Section IV]{MO-cqconv}) that
\begin{align}\label{radius limits}
\lim_{\alpha\searrow 1}\chi_{\alpha}\nw(W,P)=\chi_1\nw(W,P)=\chi(W,P),\ds\ds\ds
\lim_{\alpha\nearrow +\infty}\chi_{\alpha}\nw(W,P)=\chi_{\infty}\nw(W,P).
\end{align}

It is sometimes convenient that it is enough to consider the infimum above over invertible states, i.e., we have 
\begin{align}\label{positive minimization}
\chi_{\alpha,z}(W,P)=\inf_{\sigma\in\S(\hil)_{++}}\sum_{x\in\X}P(x)D_{\alpha,z}(W(x)\|\sigma),
\end{align}
which is obvious from the second expression in \eqref{ext1}.
Moreover, any minimizer of \eqref{weighted alpha divrad def} has the same support as the joint support of the 
channel states $\{W_x\}_{x\in\supp P}$, with projection 
\begin{align*}
\bigvee_{x\in\supp P}W(x)^0=W(P)^0,
\end{align*}
at least for a certain range of $(\alpha,z)$ values, as we show below.

\begin{lemma}\label{lemma:minimizer support}
Let $\sigma$ be a $P$-weighted $D_{\alpha,z}$ center for $W$.
If $(\alpha,z)$ is such that $D_{\alpha,z}$ is quasi-convex in its second argument
then $\sigma^0\le W(P)^0$. 
\end{lemma}
\begin{proof}
Define $\F(X):=W(P)^0XW(P)^0+(I-W(P)^0)X(I-W(P)^0)$, $X\in\B(\hil)$, and let 
$\tilde\sigma:=W(P)^0\sigma W(P)^0/\Tr W(P)^0\sigma$. 
We will show that $D_{\alpha,z}(W(x)\|\tilde\sigma)\le D_{\alpha,z}(W(x)\|\sigma)$ for all 
$x\in\supp P$, which will yield the assertion. Note that we can assume without loss of generality that 
$W(P)^0\sigma\ne 0$, since otherwise $D_{\alpha,z}(W(x)\|\sigma)=+\infty$ for all $x\in\supp P$, and hence
$\sigma$ clearly cannot be a minimizer for \eqref{weighted alpha divrad def}.

According to the decomposition $\hil=\ran W(P)^0\oplus\ran(I-W(P)^0)$, define the block-diagonal unitary 
$U:=\begin{bmatrix}I & 0\\ 0& -I\end{bmatrix}$, so that 
$\F(.)=\half\bz(.)+U(.)U^*\jz$. For every $x\in\X$,
\begin{align*}
D_{\alpha,z}(W(x)\|\F(\sigma))
&\le
\max\left\{D_{\alpha,z}(W(x)\|\sigma),D_{\alpha,z}(W(x)\|U\sigma U^*)\right\}\\
&=
\max\left\{D_{\alpha,z}(W(x)\|\sigma),D_{\alpha,z}(UW(x)U^*\|U\sigma U^*)\right\}
=
D_{\alpha,z}(W(x)\|\sigma),
\end{align*}
where the first inequality is due to quasi-convexity, and the first equality is due to the fact that 
$UW(x)U^*=W(x)$. On the other hand,
\begin{align*}
D_{\alpha,z}(W(x)\|\F(\sigma))
&=
D_{\alpha,z}(W(x)\|(\Tr W(P)^0\sigma)\tilde\sigma)\\
&=
D_{\alpha,z}\bz W(x)\|\tilde\sigma\jz-\log\Tr W(P)^0\sigma
\ge
D_{\alpha,z}\bz W(x)\|\tilde\sigma\jz,
\end{align*}
where the inequality is strict unless $\sigma^0\le W(P)^0$.
\end{proof}
\medskip

For fixed $W$ and $P$, we define
\begin{align*}
F(\sigma):=\sum_{x\in\X}P(x)D_{\alpha,z}(W(x)\|\sigma),\ds\ds\ds \sigma\in\B(\hil)_+.
\end{align*}
In the following, we may naturally interpret $W(x)$ as an operator acting on 
$\ran W(x)$ or on $\ran W(P)$.

\begin{lemma}\label{lemma:F derivative}
$F$ is Fr\'echet-differentiable at every $\sigma\in\B(\hil)_{++}$, with 
Fr\'echet-derivative $DF(\sigma)$ given by 
\begin{align}
DF(\sigma):\,Y\mapsto&
\frac{z}{\alpha-1}\sum_{x\in\X}P(x)\frac{1}{Q_{\alpha,z}(W(x)\|\sigma)}\nonumber\\
&\cdot
\Tr\sum_{a,b}h_{\alpha,z}^{[1]}(a,b)P^{\sigma}_a W(x)^{\frac{\alpha}{2z}}\bz W(x)^{\frac{\alpha}{2z}}\sigma^{\frac{1-\alpha}{z}} W(x)^{\frac{\alpha}{2z}}\jz^{z-1}W(x)^{\frac{\alpha}{2z}}P^{\sigma}_b\,Y,
\label{F derivative}
\end{align}
where $h_{\alpha,z}^{[1]}$ is the first divided difference function of $h_{\alpha,z}(t):=t^{\frac{1-\alpha}{z}}$.
\end{lemma}
\begin{proof}
We have $F=\sum_{x\in\X}P(x) (g_x\circ\iota_x\circ H_{\alpha,z})$, where 
$H_{\alpha,z}:\,\B(\hil)_+\to\B(\hil)$, $H_{\alpha,z}(\sigma):=\sigma^{\frac{1-\alpha}{z}}$ is 
Fr\'echet differentiable at every $\sigma\in\B(\hil)_{++}$ with 
$DH_{\alpha,z}(\sigma):\,Y\mapsto \sum_{a,b}h_{\alpha,z}^{[1]}(a,b)P^{\sigma}_a YP^{\sigma}_b$, according to Lemma \ref{lemma:op function derivative}.
For a fixed $x$, $\iota_x:\,\B(\hil)\to\B(\ran (W(x))$ is defined as
$A\mapsto W(x)^{\frac{\alpha}{2z}}AW(x)^{\frac{\alpha}{2z}}$, and, as a linear map, it is Fr\'echet differentiable at every $A\in\B(\hil)$, with its derivative being equal to itself. Finally, 
$g_x:\,\B(\ran W(x))\to\bR$ is defined as $g_x(T):=\Tr T^z$, and it is Fr\'echet differentiable at every 
$T\in\B(\ran W(x))_{++}$, with Fr\'echet derivative 
$Dg_x(T):\,Y\mapsto z\Tr T^{z-1}Y$, according to \eqref{Tr derivative}. 
If $\sigma\in\B(\ran W(P))_{++}$ then 
$H_{\alpha,z}(\sigma)\in\B(\ran W(P))_{++}$, and 
$\iota_x(H_{\alpha,z}(\sigma))\in\B(\ran W(x))_{++}$. Hence, we can apply the chain rule for derivatives, and obtain \eqref{F derivative}.
\end{proof}

\begin{lemma}\label{lemma:support2}
Let $\sigma$ be a $P$-weighted $D_{\alpha,z}$ center for $W$.
If
$\alpha\ge 1$ or $\alpha\in(0,1)$ and 
$1-\alpha<z<+\infty$ then 
$W(P)^0\le\sigma^0$.
\end{lemma}
\begin{proof}
When $\alpha>1$ and $W(P)^0\nleq\sigma^0$, there exists an $x\in\supp P$ with $W_x^0\nleq\sigma^0$ so that 
$D_{\alpha,z}(W(x)\|\sigma)=+\infty$. Hence, $\sigma$ cannot be a minimizer for \eqref{weighted alpha divrad def}.

Assume for the rest that $\alpha\in(0,1)$, and $\sigma$ is such that $W(P)^0\nleq \sigma^0$; 
this is equivalent to the existence of an $x_0\in\supp P$ such that 
$W_{x_0}P^{\sigma}_0\ne 0$.
Let us define the state $\omega:=cP^{\sigma}_0$, with $c:=1/\Tr P^{\sigma}_0$.
For every $t\in[0,1]$, let 
\begin{align*}
\sigma_t:=(1-t)\sigma+t\omega=
\sum_{\lambda\in\spec(\sigma)\setminus\{0\}}(1-t)\lambda P^{\sigma}_{\lambda}+tcP^{\sigma}_0,
\end{align*}
so that $\sigma_t\in\B(\hil)_{++}$ for every $t\in(0,1]$.
Note that if $t<t_0:=\lambda_{\min}(\sigma)/(c+\lambda_{\min}(\sigma))$, where 
$\lambda_{\min}(\sigma)$ is the smallest non-zero eigenvalue of $\sigma$,  then 
$P^{\sigma_t}_{ct}=P^{\sigma}_0$, and 
$P^{\sigma_t}_{(1-t)\lambda}=P^{\sigma}_{\lambda}$, $\lambda\in\spec(\sigma)\setminus\{0\}$.

By Lemma \ref{lemma:F derivative}, the derivative of $f(t):=F(\sigma_t)$ at any $t\in (0,1)$ is given by 
\begin{align*}
&f'(t)=DF(\sigma_t)(\omega-\sigma)\\
&=
\frac{z}{\alpha-1}\sum_{x\in\X}P(x)\frac{1}{Q_{\alpha,z}(W(x)\|\sigma_t)}
\Bigg[
h_{\alpha,z}'\bz ct\jz c\Tr A_{x,t}P^{\sigma}_0\\
&\hspace{5.5cm}
-
\sum_{\lambda\in\spec(\sigma)\setminus\{0\}}h_{\alpha,z}'\bz(1-t)\lambda\jz\lambda
\Tr A_{x,t}P^{\sigma}_{\lambda}
\Bigg]\\
&=
\sum_{x\in\X}P(x)\frac{1}{Q_{\alpha,z}(W(x)\|\sigma_t)}
\Bigg[
(1-t)^{\frac{1-\alpha}{z}-1}\sum_{\lambda\in\spec(\sigma)\setminus\{0\}}\lambda^{\frac{1-\alpha}{z}}\Tr A_{x,t}P^{\sigma}_{\lambda}\\
&\hspace{4.5cm}
-
t^{\frac{1-\alpha}{z}-1}c^{\frac{1-\alpha}{z}}\Tr A_{x,t}P^{\sigma}_0
\Bigg],
\end{align*}
where 
$A_{x,t}:=W(x)^{\frac{\alpha}{2z}}\bz W(x)^{\frac{\alpha}{2z}}\sigma_t^{\frac{1-\alpha}{z}} W(x)^{\frac{\alpha}{2z}}\jz^{z-1}W(x)^{\frac{\alpha}{2z}}$.

Our aim will be to show that $\lim_{t\searrow 0}f'(t)=-\infty$. This implies that 
$f(t)<f(0)$ for small enough $t>0$, contradicting the assumption that $F$ has a global minimum at $\sigma$.
Note that $\lim_{t\searrow 0}Q_{\alpha,z}(W(x)\|\sigma_t)=Q_{\alpha,z}(W(x)\|\sigma)$, which is strictly positive for every $x\in\supp P$. Indeed, the contrary would mean that 
$D_{\alpha,z}(W(x)\|\sigma)=+\infty$, contradicting again the assumption that $F$ has a global minimum at $\sigma$.
Hence, the proof will be complete if we show that 
$t^{\frac{1-\alpha}{z}-1}\Tr A_{x_0,t}P^{\sigma}_0$ diverges to $+\infty$ while 
$\Tr A_{x,t}P^{\sigma}_{\lambda}$ is bounded as $t\searrow 0$ for any $x\in\supp P$ and $\lambda\in\spec(\sigma)\setminus\{0\}$.

Note that for any $t\in(0,t_0)$ and $z\ge 1$,
\begin{align}
tc I\le\sigma_t\le I &\ds\imp\ds
(tc)^{\frac{1-\alpha}{z}} I\le\sigma_t^{\frac{1-\alpha}{z}}\le I\nonumber\\
 &\ds\imp\ds
(tc)^{\frac{1-\alpha}{z}}W(x)^{\frac{\alpha}{z}} \le 
W(x)^{\frac{\alpha}{2z}}\sigma_t^{\frac{1-\alpha}{z}}W(x)^{\frac{\alpha}{2z}}
\le W(x)^{\frac{\alpha}{z}}\nonumber\\
&\ds\imp\ds
t^{\frac{1-\alpha}{z}}c_1W(x)^0\le W(x)^{\frac{\alpha}{2z}}\sigma_t^{\frac{1-\alpha}{z}}W(x)^{\frac{\alpha}{2z}}
\le c_3W(x)^0\nonumber\\
&\ds\imp\ds
t^{\frac{1-\alpha}{z}(z-1)}c_2 W(x)^0\le
\left[W(x)^{\frac{\alpha}{2z}}\sigma_t^{\frac{1-\alpha}{z}}W(x)^{\frac{\alpha}{2z}}\right]^{z-1}
\le
c_4W(x)^0\label{operator bound1}\\
&\ds\imp\ds
t^{\frac{1-\alpha}{z}(z-1)}c_2 W(x)^{\frac{\alpha}{z}}\le
A_{x,t}
\le
c_4W(x)^{\frac{\alpha}{z}},\label{operator bound2}
\end{align}
where $c_1:=c^{\frac{1-\alpha}{z}}\lambda_{\min}(W(x))^{\frac{\alpha}{z}}>0$, 
$c_2:=c_1^{z-1}>0$, 
$c_3:=\norm{W(x)}^{\frac{\alpha}{z}}>0$, 
$c_4:=c_3^{z-1}>0$,
and the inequalities in \eqref{operator bound1}--\eqref{operator bound2} hold in the opposite direction when $z\in(0,1)$.
This immediately implies that
\begin{align*}
&t^{\frac{1-\alpha}{z}-1}\Tr A_{x_0,t}P^{\sigma}_0
\ge
t^{\frac{1-\alpha}{z}-1+\frac{1-\alpha}{z}(z-1)}c_2\Tr W(x_0)^{\frac{\alpha}{z}}P^{\sigma}_0
\xrightarrow[t\searrow 0]{}+\infty,&z\ge 1,\\
&t^{\frac{1-\alpha}{z}-1}\Tr A_{x_0,t}P^{\sigma}_0
\ge
t^{\frac{1-\alpha}{z}-1}c_4\Tr W(x_0)^{\frac{\alpha}{z}}P^{\sigma}_0
\xrightarrow[t\searrow 0]{}+\infty,&z\in(0,1),
\end{align*}
since 
$\Tr W(x_0)^{\frac{\alpha}{z}}P^{\sigma}_0>0$ by assumption,
$\frac{1-\alpha}{z}-1+\frac{1-\alpha}{z}(z-1)=-\alpha<0$, and 
$\frac{1-\alpha}{z}-1<0$ iff $1-\alpha< z$ when $z\in(0,1)$.

Next, observe that 
\begin{align*}
(1-t)P^{\sigma}_{\lambda}\le\sigma_t &\ds\imp\ds
(1-t)^{\frac{1-\alpha}{z}}P^{\sigma}_{\lambda}\le\sigma_t^{\frac{1-\alpha}{z}}
\end{align*}
where the inequality follows since, by assumption, $0<\frac{1-\alpha}{z}<1$, and 
$x\mapsto x^{\gamma}$ is operator monotone on $(0,+\infty)$ for $\gamma\in(0,1)$.
Hence,
\begin{align*}
0\le\Tr A_{x,t}P^{\sigma}_{\lambda}
\le
(1-t)^{\frac{\alpha-1}{z}}
\Tr A_{x,t}\sigma_t^{\frac{1-\alpha}{z}}
=
(1-t)^{\frac{\alpha-1}{z}}Q_{\alpha,z}(W(x)\|\sigma_t)\xrightarrow[t\searrow 0]{}Q_{\alpha,z}(W(x)\|\sigma),
\end{align*}
which is finite. This finishes the proof.
\end{proof}

\begin{rem}
Note that the region of $(\alpha,z)$ values given in Lemma \ref{lemma:support2} covers $z=1$ for all 
$\alpha\in(0,+\infty]$, i.e., all the Petz-type R\'enyi divergences, 
and 
$\{(\alpha,\alpha):\,\alpha\in(1/2,+\infty]\}$, i.e., 
the sandwiched R\'enyi divergences for every parameter $\alpha$ for which they are monotone under CPTP maps, except for $\alpha=1/2$.
It is an open question whether the condition $z> 1-\alpha$
in Lemma \ref{lemma:support2} can be improved, or maybe completely removed.
\end{rem}

\begin{rem}
Note that the case $\alpha>1$ in Lemma \ref{lemma:support2} is trivial, and this is the case that we actually need for the strong converse exponent of constant composition classical-quantum channel coding in Section \ref{sec:sc}; more precisely, we need the case $z=\alpha>1$.
\end{rem}

Let us define $\Gamma_D$ to be the set of $(\alpha,z)$ values such that 
for any gcq channel $W$ and any input probability distribution $P$, 
any $P$-weighted $D_{\alpha,z}$ center $\sigma$ for $W$ satisfies
$\sigma^0=W(P)^0$. 
Then Corollary \ref{cor:az joint convexity} and Lemmas \ref{lemma:2nd convexity}, \ref{lemma:minimizer support}
and \ref{lemma:support2} yield
\begin{align*}
\Gamma_D\supseteq\left\{(\alpha,z):\,\alpha\in(0,1),\,1-\alpha< z+\infty\right\}
\cup
\left\{(\alpha,z):\,\alpha>1,\, z\ge\max\{\alpha/2,\alpha-1\}\right\}.
\end{align*}
\smallskip

The following characterization of the weighted $D_{\alpha,z}$ centers
will be crucial in proving the additivity of the weighted sandwiched R\'enyi divergence radius of a gcq channel.

\begin{theorem}\label{prop:fixed point characterization}
Assume that $(\alpha,z)\in\Gamma_D$ are such that 
$D_{\alpha,z}$ is convex in its second variable.
Then $\sigma$ 
is a $P$-weighted $D_{\alpha,z}$ center for $W$
if and only if it is a fixed point of 
the map
\begin{align}\label{D fixed point eq}
\map_{W,P,D_{\alpha,z}}(\sigma):=&
\sum_{x\in\X}P(x)\frac{1}{Q_{\alpha,z}(W(x)\|\sigma)}
\bz\sigma^{\frac{1-\alpha}{2z}}W(x)^{\frac{\alpha}{z}}\sigma^{\frac{1-\alpha}{2z}}\jz^{z}
\end{align}
defined on $\S_{W,P}(\hil)_{++}:=\{\sigma\in\S(\hil)_+:\,\sigma^0=W(P)^0\}$.
\end{theorem}
\begin{proof}
By the assumption that $(\alpha,z)\in\Gamma_D$, 
we may restrict the Hilbert space to be $\ran W(P)^0$, and assume that $\sigma$ is invertible.
Let $F(A):=\sum_{x\in\X}P(x)D_{\alpha,z}(W(x)\|A)$, $A\in\B(\hil)_{++}$.
Due to the assumption that $D_{\alpha,z}$ is convex in its second variable, $\sigma$ is a minimizer of $F$ 
if and only if $DF(\sigma)(Y)=0$ for all self-adjoint traceless $Y$.
By Lemma \ref{lemma:F derivative}, this condition is equivalent to 
\begin{align*}
\lambda I=
\frac{z}{\alpha-1}\sum_{x\in\X}P(x)\frac{1}{Q_{\alpha,z}(W(x)\|\sigma)}
\sum_{a,b}h_{\alpha,z}^{[1]}(a,b)P^{\sigma}_a W(x)^{\frac{\alpha}{2z}}\bz W(x)^{\frac{\alpha}{2z}}\sigma^{\frac{1-\alpha}{z}} W(x)^{\frac{\alpha}{2z}}\jz^{\alpha-1}W(x)^{\frac{\alpha}{2z}}P^{\sigma}_b
\end{align*}
for some $\lambda\in\bR$. 
Multiplying both sides by $\sigma^{1/2}$ from the left and the right, and taking the trace, we get 
$\lambda=-1$. Hence, the above is equivalent to (by multiplying both sides by $\sigma^{1/2}$ from the left and the right)
\begin{align}
\sigma&=
\frac{z}{1-\alpha}\sum_{x\in\X}P(x)\frac{1}{Q_{\alpha,z}(W(x)\|\sigma)}
\sum_{a,b}a^{1/2}b^{1/2}h_{\alpha,z}^{[1]}(a,b)P^{\sigma}_a W(x)^{\frac{\alpha}{2z}}\bz W(x)^{\frac{\alpha}{2z}}\sigma^{\frac{1-\alpha}{z}} W(x)^{\frac{\alpha}{2z}}\jz^{z-1}W(x)^{\frac{\alpha}{2z}}P^{\sigma}_b\nonumber\\
&=
\sum_{a,b}P^{\sigma}_a\bz\frac{z}{1-\alpha}a^{1/2}b^{1/2}h_{\alpha,z}^{[1]}(a,b)
\what\map_{W,P,\alpha,z}(\sigma)\jz P^{\sigma}_b,\label{fixed point1}
\end{align}
where 
\begin{align*}
\what\map_{W,P,D_{\alpha,z}}(\sigma)&:=\sum_{x\in\X}P(x)\frac{1}{Q_{\alpha,z}\nw(W(x)\|\sigma)}
W(x)^{\frac{\alpha}{2z}}\bz W(x)^{\frac{\alpha}{2z}}\sigma^{\frac{1-\alpha}{z}} W(x)^{\frac{\alpha}{2z}}\jz^{z-1}W(x)^{\frac{\alpha}{2z}}.
\end{align*}
Writing the operators in \eqref{fixed point1} in block form according to the spectral decomposition of $\sigma$, 
we see that \eqref{fixed point1} is equivalent to
\begin{align*}
\forall a,b:\ds \delta_{a,b}a^{\frac{\alpha-1}{z}+1}P^{\sigma}_a
=P^{\sigma}_a\what\map_{W,P,\alpha,z}(\sigma)P^{\sigma}_b
&\ds\iff\ds
\sigma^{\frac{\alpha-1}{z}+1}=\what\map_{W,P,D_{\alpha,z}}(\sigma)\\
&\ds\iff\ds\sigma=\sigma^{\frac{1-\alpha}{2z}}\what\map_{W,P,D_{\alpha,z}}(\sigma)\sigma^{\frac{1-\alpha}{2z}}.
\end{align*}
This can be rewritten as 
\begin{align*}
\sigma&=
\sum_{x\in\X}P(x)\frac{1}{Q_{\alpha,z}(W(x)\|\sigma)}
\sigma^{\frac{1-\alpha}{2z}}W(x)^{\frac{\alpha}{2z}}
\bz W(x)^{\frac{\alpha}{2z}}\sigma^{\frac{1-\alpha}{z}}W(x)^{\frac{\alpha}{2z}}\jz^{z-1}
W(x)^{\frac{\alpha}{2z}}\sigma^{\frac{1-\alpha}{2z}}\\
&=
\sum_{x\in\X}P(x)\frac{1}{Q_{\alpha,z}(W(x)\|\sigma)}
\bz\sigma^{\frac{1-\alpha}{2z}}W(x)^{\frac{\alpha}{z}}\sigma^{\frac{1-\alpha}{2z}}\jz^{z},
\end{align*}
where the last identity follows from $Xf(X^*X)X^*=(\id_{\bR}f)(XX^*)$.
\end{proof}

\begin{rem}
The special case $z=1$ yields the characterization of the $P$-weighted Petz-type R\'enyi divergence center as the fixed point of the map
\begin{align}\label{Petz-Renyi fixed point}
\map_{W,P,D_{\alpha}}(\sigma):=\sum_{x\in\X}P(x)\frac{1}{Q_{\alpha}\old(W(x)\|\sigma)}\sigma^{\frac{1-\alpha}{2}}
W(x)^{\alpha}\sigma^{\frac{1-\alpha}{2}},\ds\ds\ds
\sigma\in\S_{W,P}(\hil)_{++},
\end{align}
for any $\alpha\in(0,+\infty)\setminus\{1\}$.
Note that in the classical case $D_{\alpha,z}$ is independent of $z$, i.e., $D_{\alpha,z}=D_{\alpha}$ for all 
$z>0$, and the above characterization of the minimizer has been derived recently by 
Nakibo\u glu in \cite[Lemma 13]{NakibogluAugustin}, using very different methods. 
Following Nakibo\u glu's approach, Cheng, Gao and Hsieh has derived the 
above characterization for the Petz-type R\'enyi divergence center in 
\cite[Proposition 4]{Cheng-Li-Hsieh2018}.
The advantage of Nakibo\u glu's approach is that it also provides quantitative bounds of the deviation of
$\sum_xP(x)D_{\alpha}(W(x)\|\sigma)$ from $\chi_{\alpha,z}(W,P)$ for an arbitrary state $\sigma$; however, it is 
not clear whether this approach can be extended
to the case $z\ne 1$, in particular, for $z=\alpha$, which is the relevant case
for the strong converse exponent of constant composition classical-quantum channel coding, as we will see in Section \ref{sec:sc}.
\end{rem}

\begin{rem}
A similar approach as in the above proof of Theorem \ref{prop:fixed point characterization}
was used by Hayashi and Tomamichel in \cite[Appendix C]{HT14} to characterize the optimal state for the 
sandwiched R\'enyi mutual information as the fixed point of a non-linear map on the state space.
We comment on this in more detail in Section \ref{sec:generalized mutual informations}.
Hayashi and Tomamichel's 
approach was used later in \cite{ChengHsiehTomamichel2019} to give another derivation of \eqref{Petz-Renyi fixed point} for $\alpha\in(0,1)$.
\end{rem}

\begin{example}\label{ex:noiseless}
We say that a cq channel $W$ is \ki{noiseless} on $\supp P$ if 
$W(x)W(y)=0$ for all $x,y\in\supp P$, $x\ne y$, i.e., 
the output states corresponding to inputs in $\supp P$ are perfectly distinguishable. A straightforward computation shows that if $W$ is noiseless on $\supp P$ then 
$\sigma:=W(P)=\sum_x P(x)W(x)$ satisfies the fixed point equation
\eqref{D fixed point eq} for any pair $(\alpha,z)$. 
Hence, if $(\alpha,z)$ satisfies the conditions of Proposition \ref{prop:fixed point characterization} then
$W(P)$ is a minimizer for \eqref{weighted alpha divrad def}, and we have 
\begin{align*}
\chi_{\alpha,z}(W,P)=\sum_{x\in\X}P(x)D_{\alpha,z}(W(x)\|W(P))=H(P):=-\sum_{x\in\X}P(x)\log P(x).
\end{align*}
Thus, the R\'enyi $(\alpha,z)$ radius of $W$ is equal to the Shannon entropy of the input distribution, independently of the value of $(\alpha,z)$.
\end{example}

\begin{cor}
If $(\alpha,z)$ satisfies the conditions of Proposition \ref{prop:fixed point characterization}, and $D_{\alpha,z}$ is monotone under CPTP maps then 
\begin{align}\label{entropy bound}
\chi_{\alpha,z}(W,P)\le H(P)
\end{align}
for any cq channel $W$ and input distribution $P$.
\end{cor}
\begin{proof}
We may assume without loss of generality that $\X=\supp P$.
Let $\tilde W(x):=\pr{e_x}$ for some orthonormal basis $(e_x)_{x\in\supp P}$ in a Hilbert space $\kil$, and let 
$\map(.):=\sum_{x\in\supp P}W(x)\bra{e_x}(.)\ket{e_x}$, which is a CPTP map from 
$\B(\kil)$ to $\B(\hil)$. We have $W=\map\circ \tilde W$, and the assertion follows from Example \ref{ex:noiseless}.
\end{proof}

\begin{rem}
Our approach to prove \eqref{entropy bound} follows that of Csisz\'ar \cite{Csiszar}.
A (much) simpler approach to prove the inequality \eqref{entropy bound} was given by Nakibo\u glu 
\cite[Lemma 13]{NakibogluAugustin}
(see also \cite[Proposition 4]{Cheng-Li-Hsieh2018} for an adaptation to various quantum R\'enyi divergences).
Obviously,
\begin{align}\label{upper bound}
\chi_{\alpha,z}(W,P)\le\sum_{x\in\X}P(x)D_{\alpha,z}(W(x)\|W(P)).
\end{align} 
Assume now that $D_{\alpha,z}$ satisfies the monotonicity property 
$\B(\hil)_+\ni\sigma_1\le\sigma_2$ $\imp$ $D_{\alpha,z}(\rho\|\sigma_1)\ge D_{\alpha,z}(\rho\|\sigma_2)$ for any 
$\rho\in\B(\hil)_+$. It is easy to see that this holds for every $(\alpha,z)$ with $z\ge |\alpha-1|$. In this case, we can lower bound 
$W(P)$ by $P(x)W(x)$, and hence 
$D_{\alpha,z}(W(x)\|W(P))\le D_{\alpha,z}(W(x)\|P(x)W(x))=-\log P(x)$, whence the RHS of \eqref{upper bound} can be upper bounded by $H(P)$.
\end{rem}

\subsection{Additivity of the weighted R\'enyi radius}
\label{sec:additivity}

Let $W^{(i)}:\,\X^{(i)}\to\B(\hil^{(i)})_+$, $i=1,2$, be gcq channels, and $P^{(i)}\in\P_f(\X^{(i)})$ be input probability distributions.
For any $\alpha\in(0,+\infty)$ and $z\in(0,+\infty]$,
\begin{align*}
&\chi_{\alpha,z}\bz W^{(1)}\otimes W^{(2)},P^{(1)}\otimes P^{(2)}\jz\nonumber\\
&\ds=
\inf_{\sigma_{12}\in\S(\hil_1\otimes\hil_2)}\sum_{x_1\in\X^{(1)},\,x_2\in\X^{(2)}}
P^{(1)}(x_1)P^{(2)}(x_2)D_{\alpha,z}\bz W^{(1)}(x_1)\otimes W^{(1)}(x_1)\|\sigma_{12}\jz\nonumber\\
&\ds\le
\inf_{\sigma_i\in\S(\hil_i)}\sum_{x_1\in\X^{(1)},\,x_2\in\X^{(2)}}
P^{(1)}(x_1)P^{(2)}(x_2)D_{\alpha,z}\bz W^{(1)}(x_1)\otimes W^{(2)}(x_2)\|\sigma_1\otimes\sigma_2\jz
\\
&\ds=
\inf_{\sigma_i\in\S(\hil_i)}\sum_{x_1\in\X^{(1)},\,x_2\in\X^{(2)}}
P^{(1)}(x_1)P^{(2)}(x_2)\left[D_{\alpha,z}\bz W^{(1)}(x_1)\|\sigma_1\jz+D_{\alpha,z}\bz W^{(2)}(x_2)\|\sigma_2\jz\right]\\
&\ds=
\chi_{\alpha,z}\bz W^{(1)},P^{(1)}\jz
+
\chi_{\alpha,z}\bz W^{(2)},P^{(2)}\jz,\nonumber
\end{align*}
where the first and the last equalities follow by definition, the second equality by the additivity of the $\alpha$-$z$ divergences, and the inequality follows by restricting the infimum to product states.
Hence, $\chi_{\alpha,z}$ is subadditive. In particular, for fixed $W:\,\X\to\S(\hil)$ and $P\in\P_f(\X)$,
the sequence $m\mapsto \chi_{\alpha,z}(W^{\otimes m},P^{\otimes m})$ is subadditive, and hence
\begin{align*}
\lim_{m\to+\infty}\frac{1}{m}\chi_{\alpha,z}(W^{\otimes m},P^{\otimes m})
=
\inf_{m\in\bN}\frac{1}{m}\chi_{\alpha,z}(W^{\otimes m},P^{\otimes m})
\le
\chi_{\alpha,z}(W,P).
\end{align*}
In fact,
\begin{align}\label{weak subadditivity}
\frac{1}{m}\chi_{\alpha,z}(W^{\otimes m},P^{\otimes m})
\le
\chi_{\alpha,z}(W,P)
\end{align}
for all $m\in\bN$. 

As it turns out, we also have the stronger property of additivity, at least for $(\alpha,z)$ pairs for which the optimal $\sigma$ can be characterized by the fixed point equation \eqref{D fixed point eq}.

\begin{theorem}\label{thm:additivity}
\ki{(Additivity of the weighted R\'enyi radius)}
Let $W^{(1)}:\,\X^{(1)}\to\S(\hil^{(1)})$ and 
$W^{(2)}:\,\X^{(2)}\to\S(\hil^{(2)})$ be gcq channels, 
and $P^{(i)}\in\P_f(\X^{(i)})$, $i=1,2$, be input distributions.
Assume, moreover, that $\alpha$ and $z$ satisfy the conditions of 
Theorem \ref{prop:fixed point characterization}.
Then
\begin{align}
\chi_{\alpha,z}\bz W^{(1)}\otimes W^{(2)},P^{(1)}\otimes P^{(2)}\jz=
\chi_{\alpha,z}\bz W^{(1)},P^{(1)}\jz
+
\chi_{\alpha,z}\bz W^{(2)},P^{(2)}\jz.
\end{align}
\end{theorem}
\begin{proof}
Let $\sigma_i$ be a minimizer of \eqref{weighted alpha divrad def} for $(W^{(i)},P^{(i)})$. By Theorem \ref{prop:fixed point characterization}, this means that 
$\map_{W^{(i)},P^{(i)},D_{\alpha,z}}(\sigma_i)=\sigma_i$.
It is easy to see that 
\begin{align*}
\map_{W^{(1)}\otimes W^{(2)},P^{(1)}\otimes P^{(2)},D_{\alpha,z}}(\sigma_1\otimes\sigma_2)
=
\map_{W^{(1)},P^{(1)},D_{\alpha,z}}(\sigma_1)\otimes
\map_{W^{(2)},P^{(2)},D_{\alpha,z}}(\sigma_2)
=
\sigma_1\otimes\sigma_2.
\end{align*}
Hence, again by Proposition \ref{prop:fixed point characterization}, 
$\sigma_1\otimes\sigma_2$ is a minimizer of \eqref{weighted alpha divrad def} for 
$(W^{(1)}\otimes W^{(2)},P^{(1)}\otimes P^{(2)})$. 
This proves the assertion.
\end{proof}

\begin{cor}\label{cor:additivity0}
For any gcq channel $W:\,\X\to\B(\hil)_+$, any $P\in\P_f(\X)$, and any 
pair $(\alpha,z)$ satisfying the conditions in Theorem \ref{prop:fixed point characterization}, we have
\begin{align*}
\chi_{\alpha,z}(W^{\otimes m},P^{\otimes m})
=
m\chi_{\alpha,z}(W,P),\ds\ds\ds m\in\bN.
\end{align*}
\end{cor}

We will need the following special case for the application to classical-quantum channel coding in the next section:
\begin{cor}\label{cor:additivity}
For any gcq channel $W:\,\X\to\B(\hil)_+$, any $P\in\P_f(\X)$, and any $\alpha\in(1/2,+\infty]$,
\begin{align*}
\chi_{\alpha}\nw(W^{\otimes m},P^{\otimes m})
=
m\chi_{\alpha}\nw(W,P),\ds\ds\ds m\in\bN.
\end{align*}
\end{cor}

\begin{rem}
As far as we are aware, the idea of proving the additivity of an information quantity by characterizing some optimizer state as the fixed point of a non-linear operator on the state space appeared first in \cite{HT14}.
We comment on this in more detail in Appendix \ref{sec:generalized mutual informations}.
\end{rem}

\section{Strong converse exponent with constant composition}
\label{sec:sc}

\subsection{The main result}
\label{sec:main}

Let $W:\,\X\to\S(\hil)$ be a classical-quantum channel.
A \ki{code} $\C_n$ for $n$ uses of the channel is a pair $\C_n=(\E_n,\D_n)$, where 
$\E_n:\,[M_n]\to\X^n$, $\D_n:\,[M_n]\to\B(\hil^{\otimes n})_+$, where 
$|\C_n|:=M_n\in\bN$ is the size of the code, and 
$\D_n$ is a POVM, i.e., $\sum_{i=1}^{M_n}\D_n(i)=I$.
The average success probability of a code $\C_n$ is 
\begin{align*}
P_s(W^{\otimes n},\C_n):=\frac{1}{|\C_n|}\sum_{m=1}^{|\C_n|}\Tr W^{\otimes n}(\E_n(m))\D_n(m).
\end{align*}

A sequence of codes $\C_n=(\E_n,\D_n)$, $n\in\bN$, is called a sequence of \ki{constant composition codes with asymptotic composition $P\in\P_f(\X)$} if there exists a sequence of types
$P_n\in\P_n(\X)$, $n\in\bN$, such that $\lim_{n\to+\infty}\norm{P_n- P}_1=0$, and $\E_n(k)\in \X^n_{P_n}$ for all $k\in\{1,\ldots,|\C_n|\}$, $n\in\bN$.
(See Section \ref{sec:Preliminaries} for the notation and basic facts concerning types.)
For any rate $R\ge 0$,
the strong converse exponents of $W$ with composition constraint $P$ are defined as
\begin{align}
\sci(W,R,P):=\inf\left\{\liminf_{n\to+\infty}-\frac{1}{n}\log P_s(W^{\otimes n},\C_n):\,\liminf_{n\to+\infty}\frac{1}{n}\log|\C_n|\ge R\right\},\label{sci}\\
\scs(W,R,P):=\inf\left\{\limsup_{n\to+\infty}-\frac{1}{n}\log P_s(W^{\otimes n},\C_n):\,\liminf_{n\to+\infty}\frac{1}{n}\log|\C_n|\ge R\right\},\label{scs}
\end{align}
where the infima are taken over code sequences of constant composition $P$. 
We may define the following variants:
\begin{align}
\sci(W,R,P)^*:=\inf\Big\{\liminf_{n\to+\infty}-\frac{1}{n}\log P_s(W^{\otimes n},\C_n):&\,\liminf_{n\to+\infty}\frac{1}{n}\log|\C_n|\ge R,\nonumber\\
&\,\lim_{n\to+\infty}\max_{1\le k\le|\C_n|}\norm{P_{\E_n(k)}-P}_1=0\Big\},\label{sci2}\\
\scs(W,R,P)^*:=\inf\Big\{\limsup_{n\to+\infty}-\frac{1}{n}\log P_s(W^{\otimes n},\C_n):
&\,\liminf_{n\to+\infty}\frac{1}{n}\log|\C_n|\ge R\nonumber\\
&\,\lim_{n\to+\infty}\max_{1\le k\le|\C_n|}\norm{P_{\E_n(k)}-P}_1=0
\Big\}.\label{scs2}
\end{align}
Obviously, we have 
\begin{align*}
\begin{array}{ccc}
\sci(W,R,P)^* &\le& \sci(W,R,P)\\
\mathrm{\vertle} & & \mathrm{\vertle}\\
\scs(W,R,P)^* &\le& \scs(W,R,P).
\end{array}
\end{align*}

Our main result is the following:
\begin{theorem}\label{thm:main result}
For any classical-quantum channel $W$, and finitely supported probability distribution $P$ on the input of $W$, and any rate $R$,
\begin{align}\label{main result}
\sci(W,R,P)^*=\scs(W,R,P)=\sup_{\alpha>1}\frac{\alpha-1}{\alpha}\left[R-\chi_{\alpha}\nw(W,P)\right].
\end{align}
\end{theorem}

We will prove the equality in \eqref{main result} as two separate inequalities in 
Lemma \ref{lemma:sc lower} and Proposition \ref{prop:upper reg}. 

\begin{rem}
According to \eqref{radius limits}, we have 
\begin{align*}
&\lim_{\alpha\searrow 1}\frac{\alpha-1}{\alpha}\left[R-\chi_{\alpha}\nw(W,P)\right]=0
=\frac{\alpha-1}{\alpha}\left[R-\chi_{\alpha}\nw(W,P)\right]\Big\vert_{\alpha=1}
,\\
&\lim_{\alpha\nearrow +\infty}\frac{\alpha-1}{\alpha}\left[R-\chi_{\alpha}\nw(W,P)\right]=
R-\chi_{\infty}\nw(W,P)
=\frac{\alpha-1}{\alpha}\left[R-\chi_{\alpha}\nw(W,P)\right]\Big\vert_{\alpha=+\infty},
\end{align*}
where we use the natural convention $\frac{\alpha-1}{\alpha}\big\vert_{\alpha=+\infty}:=1$.
With these, we may rewrite \eqref{main result} as 
\begin{align*}
\sci(W,R,P)^*=\scs(W,R,P)=\sup_{\alpha\in[1,+\infty]}\frac{\alpha-1}{\alpha}\left[R-\chi_{\alpha}\nw(W,P)\right].
\end{align*}
While this rewriting is trivial, it will be useful when applying minimax theorems in Section 
\ref{sec:sc cost const}.
\end{rem}

The following lower bound follows by a standard argument, due to Nagaoka \cite{N}.
For readers' convenience, we write out the details in Appendix \ref{sec:lower bound}.
Essentially the same argument, following Nagaoka's method, was used already in 
\cite{ChengHansonDattaHsieh2018} to show the slightly weaker bound with 
$\sci(W,R,P)$ in place of $\sci(W,R,P)^*$.
\begin{lemma}\label{lemma:sc lower}
For any $R>0$, 
\begin{align*}
\sup_{\alpha>1}\frac{\alpha-1}{\alpha}\left[R-\minfa\nw(W,P)\right]\le \sci(W,R,P)^*.
\end{align*}
\end{lemma}

\begin{rem}\label{rem:optimality}
Note that in the strong converse problem, one's aim is to make the decay of the success probability as slow as 
possible. Lemma \ref{lemma:sc lower} shows that one cannot find a better (i.e., smaller) exponent than 
the RHS of \eqref{main result}, and hence it is called the optimality part of the strong converse theorem.
Our concern in the rest will be the achievability part, i.e., that the exponent in \eqref{main result} can in fact be attained.
\end{rem}

Thus, our aim in the rest is to show that the second term is upper bounded by the rightmost term in \eqref{main result}.
We will follow the approach of \cite{MO-cqconv}, which in turn was inspired by \cite{DK}. 
A key technical ingredient in this approach is the so-called \ki{dummy channel technique}, first introduced by Haroutunian \cite{Har1968}.
We start with the following:
\begin{prop}\label{prop:sc upper}
For any $R>0$, 
\begin{align}\label{sc upper}
\scs(W,R,P)\le\sup_{\alpha>1}\frac{\alpha-1}{\alpha}\left[R-\minfa\bog(W,P)\right].
\end{align}
\end{prop}
\begin{proof}
We will show that 
\begin{align}\label{sc upper2}
\scs(W,R,P)\le\min\{F_1(W,R,P),F_2(W,R,P)\},
\end{align}
where
\begin{align*}
F_1(W,R,P)&:=\inf_{V:\,\minf(V,P)>R}D(V\|W|P),\\
F_2(W,R,P)&:=\inf_{V:\,\minf(V,P)\le R}\left[D(V\|W|P)+R-\minf(V,P)\right].
\end{align*}
Here, the infima are over channels $V:\,\X\to\S(\hil)$ satisfying the indicated properties, and
\begin{align*}
D(V\|W|P):=\sum_{x\in\X}P(x)D(V(x)\|W(x)).
\end{align*}
It was shown in \cite[Theorem 5.12]{MO-cqconv} that the RHS of \eqref{sc upper2} is the same as the RHS of \eqref{sc upper}.

We first show that $\scs(W,R,P)\le F_1(W,R,P)$. To this end, let $r>F_1(W,R,P)$; then, by definition, there exists a channel $V:\,\X\to\S(\hil)$ such that 
\begin{align*}
D(V\|W|P)<r\ds\ds\ds\text{and}\ds\ds\ds
\minf(V,P)>R.
\end{align*}
Due to $\minf(V,P)>R$, Corollary \ref{cor:random coding exponent} yields the existence of a sequence 
of constant composition codes $\C_n$ with composition $P_n$, $n\in\bN$, such that 
$\supp P_n\subseteq\supp P$ for all $n$, $\lim_{n\to+\infty}\norm{P_n-P}_1=0$, 
the rate is lower bounded as $\frac{1}{n}\log|\C_n|\ge R$, $n\in\bN$,
and 
$\lim_{n\to+\infty}P_s(V^{\otimes n},\C_n)=1$.
Note that for any message $k$, 
\begin{align*}
\Tr\bz V^{\otimes n}(\E_n(k))-e^{nr}W^{\otimes n}(\E_n(k))\jz_+
\ge
\Tr\bz V^{\otimes n}(\E_n(k))-e^{nr}W^{\otimes n}(\E_n(k))\jz\D_n(k),
\end{align*}
and hence
\begin{align*}
\Tr W^{\otimes n}(\E_n(k))\D_n(k)
&\ge
e^{-nr}\left[\Tr V^{\otimes n}(\E_n(k))\D_n(k)
-
\Tr\bz V^{\otimes n}(\E_n(k))-e^{nr}W^{\otimes n}(\E_n(k))\jz_+
\right],
\end{align*}
This in turn yields, by averaging over $k$, that
\begin{align*}
P_s(W^{\otimes n},\C_n)&\ge
e^{-nr}\left[P_s(V^{\otimes n},\C_n)
-
\frac{1}{|\C_n|}\sum_{k=1}^{|\C_n|}\Tr\bz V^{\otimes n}(\E_n(k))-e^{nr}W^{\otimes n}(\E_n(k))\jz_+
\right]\\
&=
e^{-nr}\left[P_s(V^{\otimes n},\C_n)
-
\Tr\bz V^{\otimes n}(\vecc{x}^{(n)})-e^{nr}W^{\otimes n}(\vecc{x}^{(n)})\jz_+
\right],
\end{align*}
where $\vecc{x}^{(n)}$ is any sequence in $\X^n$ with type $P_n$. Since 
$D(V\|W|P)<r$, Corollary \ref{cor:infospec limit} yields that 
$\lim_{n\to+\infty}\Tr\bz V^{\otimes n}(\vecc{x}^{(n)})-e^{nr}W^{\otimes n}(\vecc{x}^{(n)})\jz_+=0$, and so finally
\begin{align*}
\liminf_{n\to+\infty}\frac{1}{n}\log P_s(W^{\otimes n},\C_n)&\ge
-r,\ds\ds\text{whence}\ds\ds
\scs(W,R,P)\le r.
\end{align*}
Since this holds for any $r>F_1(W,R,P)$, we get $\scs(W,R,P)\le F_1(W,R,P)$.

From this, one can prove that also $\scs(W,R,P)\le F_2(W,R,P)$, the same way as it was done in 
\cite[Lemma 5.11]{MO-cqconv} (which in turn followed the proof in \cite[Lemma 2]{DK}); one only has to make sure that the extension of the code can be done in a way that it remains constant composition with composition $P$,
but that is easy to verify.
\end{proof}
\medskip

From the above result, we can obtain the desired upper bound.

\begin{prop}\label{prop:upper reg}
For any $R>0$, 
\begin{align}\label{sc upper reg}
\scs(W,R,P)\le
\sup_{\alpha>1}\frac{\alpha-1}{\alpha}\left[R-\minfa\nw(W,P)\right].
\end{align}
\end{prop}
\begin{proof}
We employ the asymptotic pinching technique from \cite{MO-cqconv}. Let
$W_m:\,\X^m\to\S(\hil^{\otimes m})$ be defined as
\begin{align*}
W_m(\vecc{x}):=\pin_m W^{\otimes m}(\vecc{x}),
\end{align*}
where $\pin_m$ is the pinching by the universal symmetric state $\sigma_{u,m}$, introduced in Section \ref{sec:Preliminaries}. Employing Proposition 
\ref{prop:sc upper} with $W\mapsto W_m$, $R\mapsto Rm$ and $P\mapsto P^{\otimes m}$, we get that for any $R>0$, there exists a sequence of codes $\C^{(m)}_k=(\E^{(m)}_k,\D^{(m)}_k)$ with constant composition $P^{(m)}_k\in\P_k(\X^n)$, $k\in\bN$, such that 
$\frac{1}{k}\log|\C^{(m)}_k|\ge mR$ for all $k$, 
$\lim_{k\to+\infty}\norm{P^{(m)}_k-P^{\otimes m}}_1=0$, 
and 
\begin{align*}
\limsup_{k\to+\infty}-\frac{1}{k}\log P_s(W_m^{\otimes k},\C^{(m)}_k)
\le
\sup_{\alpha>1}\frac{\alpha-1}{\alpha}\left[mR-\minfa\bog(W_m,P^{\otimes m})\right].
\end{align*}

For every $k\in\bN$, define $\C_{km}:=(\E^{(m)}_k,\pin_m\D^{(m)}_k)$, which can be considered a code
for $W^{\otimes km}$, with the natural identifications 
$(\X^m)^k\equiv\X^{km}$ and $(\hil^{\otimes m})^{\otimes k}\equiv\hil^{\otimes km}$. For a general $n\in\bN$, choose $k\in\bN$ such that $km\le n< (k+1)m$, and for every 
$i=1,\ldots,|\C^{(m)}_k|$, define $\E_n(i)$ to be $\E_{km}(i)$ concatenated with $n-km$ copies of some fixed 
$x_0\in\supp P$, independent of $i$ and $n$, and let $\D_n(i):=\D_{km}(i)\otimes I_{\hil}^{\otimes (n-km)}$.
Then it is easy to see that 
\begin{align*}
\liminf_{n\to+\infty}\frac{1}{n}\log|\C_n|\ge R,
\end{align*}
and 
\begin{align}\label{sc upper4}
\limsup_{n\to+\infty}-\frac{1}{n}\log P_s(W^{\otimes n},\C_n)
\le
\sup_{\alpha>1}\frac{\alpha-1}{\alpha}\left[R-\frac{1}{m}\minfa\bog(W_m,P^{\otimes m})\right].
\end{align}

We need to show that the above sequence of codes is of constant composition $P$.
Let ${\bf x}=(\vecc{x}_1,\ldots,\vecc{x}_k)\in(\X^m)^k\equiv\X^{km}$ be a codeword for $\C^{(m)}_k$, and let $P^{(m)}_{\bf{x}}$
and $P_{\bf{x}}$ denote the corresponding types when $\bf{x}$ is considered as an element of
$(\X^m)^k$ and of $\X^{km}$, respectively. For any $a\in\X$,
\begin{align*}
P_{{\bf x}}(a)&=\frac{1}{km}\sum_{\vecc{x}\in\X^m}\#\{i:\,{\bf x}_i=\vecc{x}\}\cdot\#\{j:\,x_j=a\}
=
\sum_{\vecc{x}\in\X^m}P^{(m)}_{{\bf x}}(\vecc{x})P_{\vecc{x}}(a)
=
\sum_{\vecc{x}\in\X^m}P^{(m)}_k(\vecc{x})P_{\vecc{x}}(a)
\end{align*}
only depends on ${\bf x}$ through its type $P^{(m)}_{{\bf x}}=P^{(m)}_k$, which is independent of ${\bf x}$. Thus, the type of $\E_{km}(i)$ is independent of $i$, i.e., $\C_{km}$ is a constant composition code for every 
$k\in\bN$. For a general $n\in\bN$ with $km\le n<(k+1)m$, we have
\begin{align*}
P_{\E_n(i)}(a)=\frac{km}{n}P_{\E_{km}(i)}(a)+\delta_{a,x_0}\frac{n-km}{n},\ds\ds\ds
i\in\{1,\ldots,|\C_n|=|\C_{km}|\},
\end{align*}
and hence $\C_n$ is also of constant composition.

Next, we show that $\lim_{n\to+\infty}\norm{P_n-P}_1=0$, where $P_n$ is the type of $\C_n$. For $km\le n<(k+1)m$, we have 
\begin{align*}
\sum_{a\in\X}|P_n(a)-P(a)|
&\le
\sum_{a\in\X}|P_n(a)-P_{km}(a)|+\sum_{a\in\X}|P_{km}(a)-P(a)|,
\end{align*}
and
\begin{align*}
\sum_{a\in\X}|P_n(a)-P_{km}(a)|
&=
\sum_{a\in\X}\bz 1-\frac{km}{n}\jz P_{km}(a)+\bz 1-\frac{km}{n}\jz=
2\bz 1-\frac{km}{n}\jz\to 0
\end{align*}
as $k\to+\infty$. For the second term, we get
\begin{align*}
\sum_{a\in\X}|P_{km}(a)-P(a)|
&=
\sum_{a\in\X}
\abs{
\sum_{\vecc{x}\in\X^m}P^{(m)}_k(\vecc{x})P_{\vecc{x}}(a)
-
\sum_{\vecc{x}\in\X^m}P^{\otimes m}(\vecc{x})P_{\vecc{x}}(a)
}\\
&\le
\sum_{\vecc{x}\in\X^m}\abs{P^{(m)}_k(\vecc{x})-P^{\otimes m}(\vecc{x})}
\sum_{a\in\X}P_{\vecc{x}}(a)
=\norm{P^{(m)}_k-P^{\otimes m}}_1,
\end{align*}
where in the first identity we used \eqref{lemma:type lemma},
and the last expression goes to $0$ as $k\to+\infty$ by assumption.

Since we have established that the codes used in \eqref{sc upper4} are of constant composition $P$, we get that 
for any $m\in\bN$, 
\begin{align}
\scs(W,R,P)\le\sup_{\alpha>1}\frac{\alpha-1}{\alpha}\left[R-\frac{1}{m}\minfa\bog(W_m,P^{\otimes m})\right].
\end{align}
According to \cite[Lemma 4.10]{MO-cqconv}, 
$\minfa\bog(W_m,P^{\otimes m})\ge\minfa\nw(W^{\otimes m},P^{\otimes m})-3\log v_{m,d}$, and hence
\begin{align}
\scs(W,R,P)
&\le\sup_{\alpha>1}\frac{\alpha-1}{\alpha}\left[R-\frac{1}{m}\minfa\nw(W^{\otimes m},P^{\otimes m})\right]+3\frac{\log v_{m,d}}{m}\\
&\le\sup_{\alpha>1}\frac{\alpha-1}{\alpha}\left[R-\inf_{m\in\bN}\frac{1}{m}\minfa\nw(W^{\otimes m},P^{\otimes m})\right]+3\frac{\log v_{m,d}}{m}
\end{align}
for every $m\in\bN$, from which 
\begin{align}\label{sc exponent multiletter}
\scs(W,R,P)
&\le\sup_{\alpha>1}\frac{\alpha-1}{\alpha}\left[R-\inf_{m\in\bN}\frac{1}{m}\minfa\nw(W^{\otimes m},P^{\otimes m})\right].
\end{align}
Finally, Corollary \ref{cor:additivity} yields the desired bound \eqref{sc upper reg}.
\end{proof}

\begin{rem}\label{rem:new ingredients}
The proofs of Propositions \ref{prop:sc upper} and \ref{prop:upper reg} follow closely the proof of
the composition constraint-free version in \cite{MO-cqconv}, the main ingredients of which are a suitable adaptation of the 
result by Dueck and K\"orner \cite{DK} to the classical-quantum setting, the expression of the Dueck-K\"orner exponent in terms of the log-Euclidean R\'enyi capacities \cite[Theorem 5.12]{MO-cqconv}, and the transition from the 
log-Euclidean R\'enyi capacities to the sandwiched R\'enyi capacities using the asymptotic pinching technique
\cite[Theorem 5.14]{MO-cqconv}.
Making the proof work for constant composition codes requires some new technical ingredients, like 
the constant composition version of the random coding exponent (see Appendix \ref{sec:random coding exponent} for 
a discussion)
or the evaluation of the non-i.i.d.~information spectrum quantity in Appendix \ref{sec:infospectrum}.
Probably the most important difference between the two proofs, though, is that different additivity results are 
used to arrive at single-letter expressions at the end of the proofs. 
Indeed, at the end of the proof of \cite[Theorem 5.14]{MO-cqconv} we utilized that the weighted sandwiched
$\alpha$-radius and the sandwiched $\alpha$-mutual information yield the same $\alpha$-capacity after 
optimization over the input distributions (see Corollary \ref{cor:alpha-z capacities}), and the known additivity 
of the sandwiched $\alpha$-mutual information, 
$I_{\alpha,\alpha}(W^{(1)}\otimes W^{(2)},P^{(1)}\otimes P^{(2)})
=
I_{\alpha,\alpha}(W^{(1)},P^{(1)})
+
I_{\alpha,\alpha}(W^{(2)},P^{(2)})$, $\alpha>1$,
\cite[Theorem 11]{Beigi}.
For a fixed input distribution, however, this transition from the weighted divergence radius to the mutual 
information is not possible, and therefore we needed a new additivity result for the former quantities, which we established in Section \ref{sec:additivity}.
\end{rem}

\begin{rem}\label{rem:min success}
We have defined the strong converse exponent, and stated the main result, Theorem \ref{thm:main result}, using the average success probability. By the standard argument of throwing away the worse half of any code, it can be seen immediately that Theorem \ref{thm:main result} holds unchanged if the strong converse exponent is defined using the 
worst case (minimal over all messages) success probability.
\end{rem}

\begin{rem}\label{rem:direct exponent}
Similarly to the strong converse exponents, one can define the direct exponents as
\begin{align}
\di(W,R,P):=\sup\left\{\liminf_{n\to+\infty}-\frac{1}{n}\log (1-P_s(W^{\otimes n},\C_n)):\,\liminf_{n\to+\infty}\frac{1}{n}\log|\C_n|\ge R\right\},\label{direct exp def1}\\
\dsup(W,R,P):=\sup\left\{\limsup_{n\to+\infty}-\frac{1}{n}\log (1-P_s(W^{\otimes n},\C_n)):\,\liminf_{n\to+\infty}\frac{1}{n}\log|\C_n|\ge R\right\},\label{direct exp def2}
\end{align}
where the suprema are taken over code sequences of constant composition $P$.
The following, so-called sphere packing bound has been shown by Dalai and Winter in \cite{DW2015}:
\begin{align}\label{sp bound}
\dsup(W,R,P)\le\sup_{0<\alpha<1}\frac{\alpha-1}{\alpha}\left[R-\chi_{\alpha}(W,P)\right].
\end{align}
Note that the right-hand sides of \eqref{main result} and \eqref{sp bound} are very similar to each other, except that the range of optimization is $\alpha>1$ in the former and $\alpha\in (0,1)$ in the latter, 
and the weighted R\'enyi radii corresponding to the sandwiched R\'enyi divergences appear in the former, and 
to the Petz-type R\'enyi divergences in the latter. 
Also, while 
\eqref{main result} holds for any $R>0$ (and is non-trivial for $R>\chi_1(W,P)$), 
it is known that \eqref{sp bound} holds as an equality only for high enough rates (and is non-trivial for 
$R<\chi_1(W,P)$) for classical channels, and it is a long-standing open problem if the same equality is true for classical-quantum channels.
More refined bounds on the direct exponent with improved sub-exponential corrections were obtained recently 
by Cheng, Hsieh and Tomamichel in \cite{ChengHsiehTomamichel2019}.
\end{rem}

\subsection{Strong converse exponent for classical-quantum channel coding with cost constraint}
\label{sec:sc cost const}

Assume now that using an input $x$ has some cost $\cost(x)\in\bR$. The average cost 
of an input sequence $\vecc{x}\in\X^n$ per channel use is then 
\begin{align*}
\cost(\vecc{x}):=\frac{1}{n}\sum_{k=1}^n \cost(x_k),
\end{align*}
and the (worst-case) cost of a code $\C_n=(\E_n,\D_n)$ for $n$ uses of the channel is 
\begin{align*}
\cost(\C_n):=\max_{1\le m\le|\C_n|}\cost(\E_n(m)).
\end{align*}

In the problem of \ki{classical-quantum channel coding with cost constraint}, one is seeking to determine
the trade-off between the coding rate and the error asymptotics when only codes 
with a fixed upper bound on their cost are allowed. 
The direct and strong converse exponents may be defined analogously to \eqref{direct exp def1}--\eqref{direct exp def2} and \eqref{sci}--\eqref{scs2}. In particular, the strong converse exponents are 
\begin{align}
\sci_{\cost<\costt}(W,R):=\inf\left\{\liminf_{n\to+\infty}-\frac{1}{n}\log P_s(W^{\otimes n},\C_n):\,
\liminf_{n\to+\infty}\frac{1}{n}\log|\C_n|\ge R,\,\limsup_n \cost(\C_n)< \costt\right\},\\
\scs_{\cost<\costt}(W,R):=\inf\left\{\limsup_{n\to+\infty}-\frac{1}{n}\log P_s(W^{\otimes n},\C_n):\,\liminf_{n\to+\infty}\frac{1}{n}\log|\C_n|\ge R,\,\limsup_n \cost(\C_n)< \costt\right\}.
\end{align}
Lower bounds on the 
direct and the strong converse exponents were given in Lemma 4.3 and Lemma 4.4 in \cite{Hayashibook2},
in terms of the Petz-type R\'enyi mutual informations $I_{\alpha,1}(P)$ (see Section \ref{sec:generalized mutual informations} for definitions).

The exact strong converse exponent can be obtained from Theorem \ref{thm:main result} using the simple observation that for any $\vecc{x}\in\X^n$,
\begin{align*}
\cost(\vecc{x})=\sum_{x\in\X}P_{\vecc{x}}(x)\cost(x)=\Exp_{P_{\vecc{x}}}(\cost),
\end{align*}
where $P_{\vecc{x}}$ is the type of $\vecc{x}$ (see Section \ref{sec:Preliminaries}). 
Let us introduce 
\begin{align*}
P_{f,\cost<\costt}(\X):=\{P\in\P_f(\X):\,\Exp_p(\cost)<\costt\}.
\end{align*}
Similarly to the constant composition case, the following lower bound follows by a straightforward application of Nagaoka's method:
\begin{lemma}\label{lemma:sc optimality cost}
In the above setting,
\begin{align}\label{sc optimality cost}
\sci_{\cost<\costt}(W,R)\ge 
\sup_{\alpha\in[1,+\infty]}\frac{\alpha-1}{\alpha}\left[R-\sup_{P\in\P_{f,\cost<\costt}(\X)}
\chi_{\alpha}\nw(W,P)\right].
\end{align}
\end{lemma}
We give the proof in Appendix \ref{sec:lower bound}.
\medskip

Note that, with the change of variables $u:=\frac{\alpha-1}{\alpha}$, the RHS in \eqref{sc optimality cost} can be rewritten as 
\begin{align}
\sup_{\alpha\in[1,+\infty]}\frac{\alpha-1}{\alpha}\left[R-\sup_{P\in\P_{f,\cost<\costt}(\X)}
\chi_{\alpha}\nw(W,P)\right]
=
\sup_{0\le u\le 1}\inf_{P\in\P_{f,\cost<\costt}(\X)}\{uR-f(P,u)\},
\end{align}
where 
\begin{align}\label{aux funct2}
f(P,u):=u\chi_{\frac{1}{1-u}}\nw(W,P),\ds\ds\ds u\in[0,1],\ds P\in\P_f(\X);
\end{align}
the case $u=1$ is to be interpreted as $\lim_{u\searrow 1}u\chi_{\frac{1}{1-u}}\nw(W,P)=\chi_{\infty}\nw(W,P)$.
It is clear that $f$ is a concave function of $P$ for any fixed $u$.
A highly non-trivial counterpart, proved very recently in 
\cite{Cheng-Li-Hsieh2018}, is the following:
\begin{lemma}\label{lemma:aux conv}
For any fixed $P\in\P_f(\X)$, $f(P,.)$ is convex on $[0,1]$.
\end{lemma}

\begin{cor}\label{cor:cost const minimax}
For any cost function $\cost:\,\X\to\bR$, and any $\costt\in\bR$, 
\begin{align}\label{cost const minimax}
\sup_{\alpha\in[1,+\infty]}\frac{\alpha-1}{\alpha}\left[R-\sup_{P\in\P_{f,\cost<\costt}(\X)}
\chi_{\alpha}\nw(W,P)\right]
=
\inf_{P\in\P_{f,\cost<\costt}(\X)}\sup_{\alpha\in[1,+\infty]}\frac{\alpha-1}{\alpha}\left[R-
\chi_{\alpha}\nw(W,P)\right].
\end{align}
\end{cor}
\begin{proof}
It is clear that the function $h(u,P):= uR-f(P,u)$ is convex in $P$ for any fixed $u$. 
For any fixed $P$, $h(.,P)$ is clearly continuous, and, 
by Lemma \ref{lemma:aux conv}, it is concave on $[0,1]$. Thus, by Lemma \ref{lemma:KF+ minimax},
\begin{align}\label{minimax3}
\inf_{P\in\P_{f,\cost<\costt}(\X)}\sup_{0\le u\le 1}\{uR-f(P,u)\}=
\sup_{0\le u\le 1}\inf_{P\in\P_{f,\cost<\costt}(\X)}\{uR-f(P,u)\}.
\end{align}
Changing the variable $u$ to $\alpha:=1/(1-u)$ yields \eqref{cost const minimax}.
\end{proof}

The exact strong converse exponent with cost constraint is given by the following:

\begin{theorem}\label{thm:sc cost const}
In the above setting,
\begin{align}\label{cost const sc}
\sci_{\cost<\costt}(W,R)=\scs_{\cost<\costt}(W,R)=
\sup_{\alpha\in[1,+\infty]}\frac{\alpha-1}{\alpha}\left[R-\sup_{P\in\P_{f,\cost<\costt}(\X)}\chi_{\alpha}\nw(W,P)\right].
\end{align}
\end{theorem}
\begin{proof}
By Lemma \ref{lemma:sc optimality cost}, it is sufficient to prove that $\scs_{\cost<\costt}(W,R)$ is upper bounded by the RHS of 
\eqref{cost const sc}. Let $P\in\P_f(\X)$ by such that $\Exp_P(\cost)<\costt$. 
By Proposition \ref{prop:upper reg}, there exists a sequence of constant composition codes $(\C_n)_{n\in\bN}$ 
with asymptotic composition $P$ such that 
\begin{align*}
\limsup_n-\frac{1}{n}\log P_s(W^{\otimes n},\C_n)\le
\sup_{\alpha\in[1,+\infty]}\frac{\alpha-1}{\alpha}\left[R-\chi_{\alpha}\nw(W,P)\right].
\end{align*}
Moreover, this code sequence can be chosen so that $P_n:=P_{\E_n(k)}$, $k=1,\ldots,|\C_n|$, satisfies 
$\supp P_n\subseteq \supp P$, $n\in\bN$, (see Corollary \ref{cor:random coding exponent}). 
Since $\lim_n\norm{P_n-P}_1=0$, we have $\cost(\C_n)=\Exp_{P_n}(\cost)<\costt$ for all large enough $n$. 
Thus, 
\begin{align}
\scs_{\cost<\costt}(W,R)
\le
\sup_{\alpha\in[1,+\infty]}\frac{\alpha-1}{\alpha}\left[R-\chi_{\alpha}\nw(W,P)\right].
\end{align}
Since this holds for every $P\in\P_f(\X)$ with $\Exp_P(\cost)<\costt$, we finally have
\begin{align*}
\scs_{\cost<\costt}(W,R)
&\le
\inf_{\P_{f,\cost<\costt}}\sup_{\alpha\in[1,+\infty]}\frac{\alpha-1}{\alpha}\left[R-\chi_{\alpha}\nw(W,P)\right]
=
\sup_{\alpha\in[1,+\infty]}\frac{\alpha-1}{\alpha}\left[R-\sup_{P\in\P_{f,\cost<\costt}(\X)}\chi_{\alpha}\nw(W,P)\right],
\end{align*}
where the equality is due to Corollary \ref{cost const minimax}.
\end{proof}
\medskip

As it has been shown in \cite{MO-cqconv}, the strong converse exponent of a cq channel $W$ with unconstrained 
coding is given by 
\begin{align}\label{sc constraint free}
\mathrm{sc}(W,R):=\sci(W,R)=\scs(W,R)=\sup_{\alpha>1}\frac{\alpha-1}{\alpha}\left[R-\chi_{\alpha}\nw(W)\right]
\end{align}
for any $R>0$, where $\chi_{\alpha}\nw(W)=\sup_{P\in\P_f(\X)}\chi_{\alpha}\nw(W,P)$, and 
$\sci(W,R)$ and $\scs(W,R)$ are defined analogously to \eqref{sci}--\eqref{scs} by dropping the constant composition constraint. 
 It is natural to ask whether this optimal value can be achieved, or at least arbitrarily well approximated, by constant composition codes, i.e., whether we have
\begin{align}\label{constant approximation}
\inf_{P\in\P_f(\X)}\sc(W,R,P)=\sc(W,R).
\end{align}
In view of Theorem \ref{thm:main result}, this is equivalent to whether
\begin{align}\label{minimax1}
\inf_{P\in\P_f(\X)}\sup_{\alpha\in[1,+\infty]}\frac{\alpha-1}{\alpha}\left[R-\chi_{\alpha}\nw(W,P)\right]
=
\sup_{\alpha\in[1,+\infty]}\inf_{P\in\P_f(\X)}\frac{\alpha-1}{\alpha}\left[R-\chi_{\alpha}\nw(W,P)\right].
\end{align}
The answer to this question is affirmative: Indeed, by choosing a constant cost function 
$\cost\equiv\cost_0$, and any $\costt>\cost_0$, \eqref{minimax1} follows as a special case of
Corollary \ref{cor:cost const minimax}. 
In fact, \eqref{sc constraint free} follows as a special case of
Theorem \ref{thm:sc cost const} with the above trivial choice of $c$ and $\gamma$.

\begin{rem}
The identity \eqref{minimax1} was also stated in \cite{Cheng-Li-Hsieh2018}, although only as a formal 
identity, as the operational interpretation of 
$\sup_{\alpha\in[1,+\infty]}\frac{\alpha-1}{\alpha}\left[R-\chi_{\alpha}\nw(W,P)\right]$, (i.e., 
Theorem \ref{thm:main result}), and hence the equivalence of 
\eqref{minimax1} and \eqref{constant approximation},
 had not yet been known then.
\end{rem}

\subsection{Generalized cutoff rates for constant compositions and for cost constraint}

The convexity of $f(P,.)$, defined in \eqref{aux funct2}, also plays an important role in establishing the weighted sandwiched 
R\'enyi divergence radii as generalized cutoff rates in the sense of Csisz\'ar \cite{Csiszar}.
Following \cite{Csiszar}, for a fixed $\kappa>0$, 
we define the 
generalized $\kappa$-cutoff rate $C_{\kappa}(W,P)$ for a cq channel $W$ and input distribution $P$ as
the smallest number $R_0$ satisfying 
\begin{align}\label{cutoff def}
\sci(W,R,P)\ge \kappa(R-R_0),\ds\ds\ds R>0.
\end{align}
The following extends the analogous result for classical channels in \cite{Csiszar} to classical-quantum channels,
and gives a direct operational interpretation of the weighted sandwiched R\'enyi divergence radius of a cq channel as a generalized cutoff rate.

\begin{prop}\label{prop:cutoff}
For any $\kappa\in(0,1)$, 
\begin{align*}
C_{\kappa}(W,P)=\chi\nw_{\frac{1}{1-\kappa}}(W,P),
\end{align*}
or equivalently, for any $\alpha>1$, 
\begin{align*}
\chi\nw_{\alpha}(W,P)=C_{\frac{\alpha-1}{\alpha}}(W,P).
\end{align*}
\end{prop}
\begin{proof}
By Theorem \ref{thm:main result}, we have 
\begin{align*}
\sci(R,W,P)=\sup_{0<u<1}\{uR-f(P,u)\}\ge \kappa R-f(P,\kappa)=\kappa\bz R-\frac{1}{\kappa} f(P,\kappa)\jz,\ds\ds\ds
\kappa\in(0,1),
\end{align*}
where the inequality is trivial. Since $f(P,.)$ is convex according to \cite{Cheng-Li-Hsieh2018}, its left and right derivatives at $\kappa$, $\derleft{f(P,.)}(\kappa)$ and $\derright{f(P,.)}(\kappa)$, exist. Obviously, 
for any $\derleft{f(P,.)}(\kappa)\le R\le \derright{f(P,.)}(\kappa)$, 
\begin{align*}
\sup_{0<u<1}\{uR-f(P,u)\}= \kappa R-f(P,\kappa)=\kappa\bz R-\frac{1}{\kappa} f(P,\kappa)\jz,
\end{align*}
showing that $\frac{1}{\kappa} f(P,\kappa)=\chi\nw_{\frac{1}{1-\kappa}}(W,P)$ is the minimal $R_0$ for which \eqref{cutoff def} holds for all $R>0$.
\end{proof}
\medskip

We can analogously define the generalized $\kappa$-cutoff rate $C_{\kappa,\cost<\costt}(W)$
for a cq channel $W$, cost function $\cost$ and threshold $\costt$, as
the smallest number $R_0$ satisfying 
\begin{align}\label{cutoff def cost const}
\sci_{\cost<\costt}(W,R)\ge \kappa(R-R_0),\ds\ds\ds R>0.
\end{align}
A completely similar argument as in the proof of Proposition \ref{prop:cutoff}, using that 
$\sup_{P\in\P_{f,\cost<\costt}}(P,.)$, as the supremum of convex functions, is convex, yields the following:

\begin{prop}\label{prop:cutoff cost const}
For any $\kappa\in(0,1)$, 
\begin{align*}
C_{\kappa,\cost<\costt}(W)=\sup_{P\in\P_{f,\cost<\costt}}\chi\nw_{\frac{1}{1-\kappa}}(W,P),
\end{align*}
or equivalently, for any $\alpha>1$, 
\begin{align*}
\sup_{P\in\P_{f,\cost<\costt}}\chi\nw_{\alpha}(W,P)=C_{\frac{\alpha-1}{\alpha},\cost<\costt}(W,P).
\end{align*}
\end{prop}

Finally, the generalized $\kappa$-cutoff rate $C_{\kappa}(W)$ without constraints is
the smallest number $R_0$ satisfying 
\begin{align*}
\sci(W,R)\ge \kappa(R-R_0),\ds\ds\ds R>0.
\end{align*}
Using again a constant cost function $\cost\equiv\cost_0$ and $\costt>\cost_0$, 
Proposition \ref{prop:cutoff cost const} yields the 
generalized $\kappa$-cutoff rate representation of the R\'enyi capacities $\chi_{\alpha}\nw(W)=\sup_{P\in\P_f(\X)}\chi_{\alpha}\nw(W,P)$
for any $\alpha\in(1,+\infty)$ in the context of constraint-free cq channel coding
as follows:
\begin{cor}
For any $\kappa\in(0,1)$,
\begin{align*}
C_{\kappa}(W)=\chi\nw_{\frac{1}{1-\kappa}}(W),
\end{align*}
or equivalently, for any $\alpha\in(1,+\infty)$,
\begin{align*}
\chi\nw_{\alpha}(W)=C_{\frac{\alpha-1}{\alpha}}(W).
\end{align*}
\end{cor}

\appendix

\section{Further properties of divergence radii}
\label{sec:divrad further}

\subsection{General divergences}
\label{sec:general divrad}

Here we consider, among others, the connection between the divergence radius and the weighted divergence radius for a general divergence $\divv$. The following 
is a common generalization and simplification of several results of the same kind for various R\'enyi divergences \cite{KW,MH,MO-cqconv,WWY}. 

\begin{prop}\label{prop:radius equality}
Assume that $\divv$ 
is convex and lower semi-continuous 
in its second argument. Then 
\begin{align}\label{radius Holevo equality}
\sup_{P\in\P_f(S)}R_{\divv,P}(S)= R_{\divv}(S)
\end{align}
for any $S\subseteq\B(\hil)_+$.
\end{prop}
\begin{proof}
We have
\begin{align*}
R_{\divv}(S)
=
\inf_{\sigma\in\S(\hil)}\sup_{\rho\in S}\divv(\rho\|\sigma)
&=
\inf_{\sigma\in\S(\hil)}\sup_{P\in\P_f(S)}\sum_{\rho\in S}P(\rho)\divv(\rho\|\sigma)\\
&=
\sup_{P\in\P_f(S)}\inf_{\sigma\in\S(\hil)}\sum_{\rho\in S}P(\rho)\divv(\rho\|\sigma)
=
\sup_{P\in\P_f(S)}R_{\divv,P}(S).
\end{align*}
The first equality above is by definition, and the second one is trivial. The third one follows from 
Lemma \ref{lemma:KF+ minimax} by noting that $\sum_{\rho\in S}P(\rho)\divv(\rho\|\sigma)$ is convex and lower semi-continuous in $\sigma$ on the compact set $\S(\hil)$, and it is trivially concave (in fact, affine) on the convex set $\P_f(S)$. The last equality is again by definition.
\end{proof}

Note that if $\divv$ is additive on tensor products then the corresponding divergence radius is subadditive by definition: for any $S^{(i)}\subseteq\B(\hil^{(i)})_+$, $i=1,2$,
\begin{align}
R_{\divv}(S^{(1)}\dotimes S^{(2)})
&=
\inf_{\sigma_{12}\in\S(\hil^{(1)}\otimes \hil^{(2)})}
\sup_{\rho_1\in\S^{(1)},\rho_2\in\S^{(2)}}\divv(\rho_1\otimes\rho_2\|\sigma_{12})\nonumber\\
&\le
\inf_{\sigma_{1}\in\S(\hil^{(1)}),\,\sigma_2\in\otimes \S(\hil^{(2)})}
\sup_{\rho_1\in\S^{(1)},\rho_2\in\S^{(2)}}\divv(\rho_1\otimes\rho_2\|\sigma_1\otimes\sigma_2)\nonumber\\
&=
R_{\divv}(S^{(1)})+R_{\divv}(S^{(2)}),\label{divrad subadd}
\end{align}
where $S^{(1)}\dotimes S^{(2)}:=\{\rho_1\otimes\rho_2:\,\rho_i\in S^{(i)},\,i=1,2\}$.
It is easy to see that if $\divv$ satisfies the conditions of Proposition \ref{prop:radius equality}, and the 
weighted $\divv$-radius is additive, then so is the $\divv$-radius:
\begin{prop}\label{prop:divrad additive}
Assume that $\divv$ is convex and lower semi-continuous 
in its second argument and it is additive on tensor products, and the weighted $\divv$-radius is additive in the sense that 
$R_{\divv,P^{(1)}\otimes P^{(2)}}=R_{\divv,P^{(1)}}+R_{\divv,P^{(2)}}$ for any 
$P^{(i)}\in\P_f(\B(\hil^{(i)})_+$, $i=1,2$. Then the $\divv$-radius is additive, i.e., 
\eqref{divrad subadd} holds as an equality.
\end{prop}
\begin{proof}
We have already seen subadditivity of the $\divv$-radius in \eqref{divrad subadd}. The converse inequality follows by 
\begin{align*}
R_{\divv}(S^{(1)}\dotimes S^{(2)})
&=
\sup_{P_{12}\in\P_f(S^{(1)}\otimes S^{(2)})}R_{\divv,P_{12}}(S^{(1)}\dotimes S^{(2)})\\
&\ge
\sup_{P_1\in\P_f(S^{(1)}),\,P_2\in\P_f(S^{(2)})}R_{\divv,P_1\otimes P_2}(S^{(1)}\otimes S^{(2)})\\
&=
\sup_{P_1\in\P_f(S^{(1)}),\,P_2\in\P_f(S^{(2)})}\left\{R_{\divv,P_1}(S^{(1)})+R_{\divv,P_2}(S^{(2)})\right\}\\
&=
R_{\divv}(S^{(1)})+R_{\divv}(S^{(2)}),
\end{align*}
where the first and the last identities are due to Proposition \ref{prop:radius equality},
the inequality is trivial, and the second equality follows by the additivity assumption on the weighted 
$\divv$-radius.
\end{proof}

\begin{rem}\label{rem:capacity additivity}
In the formalism of gcq channels, Proposition \ref{prop:divrad additive} says that 
if $\divv$ is convex and lower semi-continuous 
in its second argument then
\begin{align*}
\chi_{\divv}(W)=R_{\divv}(\ran W)
\end{align*}
for any gcq channel $W$, and if $\divv$ is also additive on tensor products, then 
the additivity 
\begin{align*}
\chi_{\divv}(W^{(1)}\otimes W^{(2)},P^{(1)}\otimes P^{(2)})=\chi_{\divv}(W^{(1)},P^{(1)})+\chi_{\divv}(W^{(2)},P^{(2)})
\end{align*}
 for some gcq channels $W^{(i)}:\,\X^{(i)}\to\B(\hil^{(i)})_+$ and any $P^{(i)}\in\P_f(\X^{(i)})$, $i=1,2$, implies the additivity 
\begin{align*}
\chi_{\divv}(W^{(1)}\otimes W^{(2)})=\chi_{\divv}(W^{(1)})+\chi_{\divv}(W^{(2)}).
\end{align*}
We discuss this further for the case of $\alpha$-$z$ divergences at the end of Section \ref{sec:generalized mutual informations}.
\end{rem}

%

When $\divv$ is non-negative on $\supp P$ for some $P\in\P_f(\B(\hil))$ in the sense that 
$\divv(\rho\|\sigma)\ge 0$ for all $\rho\in\supp P$ and $\sigma\in\S(\hil)$, it is possible to define a continuous approximation between 
the $\divv$-radius and the $P$-weighted $\divv$-radius as follows.
Define for every 
$\beta\in[1,+\infty]$ the \ki{$(P,\beta)$-weighted divergence radius} of a set $S\subseteq\B(\hil)_+$ as 
\begin{align}\label{P beta div rad}
R_{\divv,P,\beta}(S):=\inf_{\sigma\in\S(\hil)}
\norm{\divv(.\|\sigma)}_{P,\beta}:=\inf_{\sigma\in\S(\hil)}\begin{cases}
\bz\sum_{\rho\in S}P(\rho)\divv(\rho\|\sigma)^{\beta}\jz^{1/\beta},&\beta\in[1,+\infty),\\
\sup_{\rho\in\supp P}\divv(\rho\|\sigma),&\beta=+\infty.
\end{cases}
\end{align}
Just like before, a $\sigma\in\S(\hil)$ is called a \ki{$(P,\beta)$-weighted $\divv$-centre} if it attains the infimum in \eqref{P beta div rad}.
Again, $R_{\divv,P,\beta}(S)$ is in fact independent of $S$, and hence we will often drop it from the notation.

Note that $R_{\divv,P,1}=R_{\divv,P}$, and
when $S$ is finite and $\supp P=S$ then $R_{\divv,P,+\infty}(S)=R_{\divv}(S)$.
In general, though, we need a further optimization to recover 
$R_{\divv}$ from $R_{\divv,P,+\infty}$.
According to well-known properties of the $\beta$-norms,
\begin{align}\label{beta monotonicity}
R_{\divv,P,\beta_1}\le R_{\divv,P,\beta_2}\ds\text{when}\ds\beta_1\le\beta_2,\ds\text{and}\ds
R_{\divv,P,\beta}\nearrow R_{\divv,P,+\infty}\ds\text{as}\ds\beta\nearrow+\infty
\end{align}
for any 
$P\in\P_f(\B(\hil)_+)$. Moreover, it is clear from the definitions that
\begin{align}\label{radius inequality}
\sup_{P\in\P_f(S)}R_{\divv,P,\beta}\le R_{\divv}(S)
\end{align}
for any $S\subseteq\B(\hil)_+$ and $\beta\in[1,+\infty]$.
Under the conditions of Proposition \ref{prop:radius equality}, the above holds as an equality:

\begin{prop}\label{prop:radius equality2}
Assume that $\divv$  is non-negative on some $S\subseteq\B(\hil)_+$, and
convex and lower semi-continuous 
in its second argument. Then 
\begin{align}\label{radius Holevo equality beta}
\sup_{P\in\P_f(S)}R_{\divv,P,\beta}(S)= R_{\divv}(S)
\end{align}
for any $\beta\in[1,+\infty]$.
\end{prop}
\begin{proof}
Due to \eqref{radius inequality} and the monotonicity \eqref{beta monotonicity}, it is enough to prove the assertion for $\beta=1$, which has already been established in Proposition \ref{prop:radius equality}.
\end{proof}

%
%
%
\medskip

In the rest of the section we explore some properties of the divergence radius $R_{\divv}$.
We will denote the set of $\divv$-centers of $S$ by $C_{\divv}(S)$.

\begin{lemma}\label{lemma:divrad basic properties}
The $\divv$-radius satisfies the following simple properties.
\begin{enumerate}
\item\label{div rad prop1}
The $\divv$-radius is a monotone function of $S$, i.e., if $S\subseteq S'$ then 
$R_{\divv}(S)\le R_{\divv}(S')$. 
\item\label{monotone divcentre}
If $S\subseteq S'$ and $R_{\divv}(S)= R_{\divv}(S')$ then 
$C_{\divv}(S')\subseteq C_{\divv}(S)$.
\item\label{div rad prop2}
If $\divv$ is quasi-convex in its first argument then $R_{\divv}(S)=R_{\divv}(\conv(S))$ and 
$C_{\divv}(S)=C_{\divv}(\conv(S))$.
\item\label{div rad prop3}
If $\divv$ is lower semi-continuous in its first argument then 
$R_{\divv}(S)=R_{\divv}(\oll S)$ and $C_{\divv}(S)=C_{\divv}(\oll S)$.
\end{enumerate}
\end{lemma}
\begin{proof}
The monotonicity in \ref{div rad prop1} is obvious. 

Assume that the conditions of \ref{monotone divcentre} hold, and that 
$\sigma$ is a $\divv$-centre for $S'$ (if $C_{\divv}(S)=\emptyset$ then the assertion holds trivially). Then
\begin{align*}
R_{\divv}(S)\le\sup_{\rho\in S}\divv(\rho\|\sigma)\le
\sup_{\rho\in S'}\divv(\rho\|\sigma)=R_{\divv}(S')=R_{\divv}(S),
\end{align*}
from which $\sup_{\rho\in S}\divv(\rho\|\sigma)=R_{\divv}(S)$, i.e., $\sigma$ is a $\divv$-centre of $S$.

As a consequence of \ref{div rad prop1} and \ref{monotone divcentre}, in \ref{div rad prop2} 
we only have to prove that 
$R_{\divv}(S)\ge R_{\divv}(\conv(S))$ and $C_{\divv}(S)\subseteq C_{\divv}(\conv(S))$, and analogously in 
\ref{div rad prop3}, with $\oll S$ in place of $\conv(S)$.

Assume that $\divv$ is quasi-convex in its first argument. For any 
$\rho'\in\conv(S)$, there exists a finitely supported probability distribution $P_{\rho'}\in\P_f(S)$ such that 
$\rho'=\sum_{\rho\in S}P_{\rho'}(\rho)\rho$, and hence
$\divv(\rho'\|\sigma)\le\max_{\rho\in \supp P_{\rho'}}\divv(\rho\|\sigma)\le\sup_{\rho\in S}\divv(\rho\|\sigma)$
for any $\sigma\in\S(\hil)$. Taking the supremum in $\rho'\in\conv(S)$ and then the infimum in $\sigma\in\S(\hil)$ yields 
$R_{\divv}(\conv(S))\le R_{\divv}(S)$.
If $\sigma\in C_{\divv}(S)$ then 
\begin{align*}
R_{\divv}(\conv(S))\le\sup_{\rho'\in \conv(S)}\divv(\rho'\|\sigma)\le
\sup_{\rho'\in \conv(S)}\sup_{\rho\in \supp P_{\rho'}}\divv(\rho\|\sigma)
=
\sup_{\rho\in S}\divv(\rho\|\sigma)
&=R_{\divv}(S)\\
&=R_{\divv}(\conv(S)),
\end{align*}
from which $R_{\divv}(\conv(S))=\sup_{\rho'\in \conv(S)}\divv(\rho'\|\sigma)$, i.e., $\sigma\in C_{\divv}(\conv (S))$.

Assume now that $\divv$ is l.s.c.~in its first argument, and let $\rho'\in\oll S$. Then there exists a sequence
$(\rho_n)_{n\in\bN}\subseteq S$ converging to $\rho'$, and hence, by lower semi-continuity,
\begin{align}\label{lsc proof1}
\divv(\rho'\|\sigma)\le\liminf_{n\to+\infty}\divv(\rho_n\|\sigma)\le\sup_{\rho\in S}\divv(\rho\|\sigma),\ds\ds\ds\sigma\in\S(\hil).
\end{align}
Taking the supremum in $\rho'\in\oll S$ and then the infimum in $\sigma\in\S(\hil)$ yields 
$R_{\divv}(\oll S)\le R_{\divv}(S)$.
If $\sigma\in C_{\divv}(S)$ then 
\begin{align*}
R_{\divv}(\oll S)\le\sup_{\rho'\in \oll S}\divv(\rho'\|\sigma)\le
\sup_{\rho\in S}\divv(\rho\|\sigma)
=R_{\divv}(S)
=R_{\divv}(\oll S),
\end{align*}
where the second inequality is due to \eqref{lsc proof1}. 
This yields that $R_{\divv}(\oll S)=\sup_{\rho'\in \oll S}\divv(\rho'\|\sigma)$, i.e., 
$\sigma\in C_{\divv}(\oll S)$.
\end{proof}

\begin{cor}
If $\divv$ is quasi-convex and lower semi-continuous in its first argument then 
\begin{align*}
R_{\divv}(S)=R_{\divv}(\oll\conv(S))\ds\ds\ds\text{and}\ds\ds\ds
C_{\divv}(S)=C_{\divv}(\oll\conv(S)).
\end{align*}
for any $S$.
\end{cor}

As a consequence, when studying the divergence radius for a divergence with the above properties, we can often restrict our investigation to closed convex sets without loss of generality.

\begin{prop}
Assume that $\divv$ satisfies \eqref{radius Holevo equality} for any 
$S\subseteq\B(\hil)_+$. Then $R_{\divv}$ is continuous on monotone increasing nets of subsets of $\B(\hil)_+$, i.e., if $\S\subseteq P(\B(\hil)_+)$ (all the subsets of $\B(\hil)_+$) such that 
for all $S,S'\in\S$ there exists an $S''\in\S$ with $S\cup S'\subseteq S''$ then 
\begin{align}\label{mon cont}
R_{\divv}(\cup\S)=\sup_{S\in\S}R_{\divv}(S).
\end{align}
In particular, for any $S\subseteq\B(\hil)_+$, 
\begin{align}\label{mon cont2}
R_{\divv}(S)=\sup\{R_{\divv}(S'):\,S'\subseteq S,\,S'\text{ finite}\}.
\end{align}
\end{prop}
\begin{proof}
It is clear from the monotonicity stated in Lemma \ref{lemma:divrad basic properties} that 
\begin{align*}
R_{\divv}(\cup\S)\ge\sup_{S\in\S}R_{\divv}(S),
\end{align*}
and hence we only have to prove the converse inequality. To this end, let 
\begin{align*}
c<R_{\divv}(\cup\S)=\sup_{P\in\P_f(\cup\S)}R_{\divv,P}(\cup\S),
\end{align*}
where the equality holds by assumption. Then there exists a $P\in\P_f(\cup\S)$ for which 
$c<R_{\divv,P}(\cup\S)$, and there exists an $S\in\S$ such that $\supp P\subseteq S$, and hence
$R_{\divv,P}(\cup\S)=R_{\divv,P}(S)$. From this we get \eqref{mon cont}, and 
\eqref{mon cont2} follows immediately.
\end{proof}

\begin{rem}\label{rem:OPW}
Theorem 3.5 in \cite{OPW} states \ref{radius Holevo equality} for the relative entropy. 
However, in their proof they use \eqref{mon cont2} without any explanation. In the proof above we assumed that 
\ref{radius Holevo equality} holds, so the question arises whether the proof in \cite{OPW} can be made complete in some other way.
\end{rem}
\smallskip

\subsection{Generalized mutual information}
\label{sec:generalized mutual informations}

For any gcq channel $W:\,\X\to\B(\hil)_+$, we define the lifted channel
\begin{align*}
\ext{W}:\,\X\to\S(\hil_\X\otimes\hil),\ds\ds\ds
\ext{W}(x):=\pr{x}\otimes W(x).
\end{align*}
Here, $\hil_{\X}$ is an auxiliary Hilbert space, and $\{\ket{x}:\,x\in\X\}$ is an orthonormal basis in it.
As a canonical choice, one can use $\hil_{\X}=l^2(\X)$, the $L^2$-space on $\X$ with respect to the counting measure,
and choose $\ket{x}:=\egy_{\{x\}}$ to be the characteristic function (indicator function) of the singleton $\{x\}$.
Note that this is well-defined irrespectively of the cardinality of $\X$.
The classical-quantum state 
\begin{align*}
\ext{W}(P):=\sum_{x\in\X}P(x)\pr{x}\otimes W(x)
\end{align*} 
plays the role of the joint distribution of the input and the output of the channel for a fixed finitely supported input probability distribution $P\in\P_f(\X)$.

For a general divergence $\divv$
and a gcq channel $W:\,\X\to\B(\hil)_+$, we may define the \ki{$\divv$-mutual information}
between the classical input and the quantum output of the channel for a fixed input distribution $P\in\P_f(\X)$ as 
\begin{align}\label{mutual info def}
I_{\divv}(W,P):=\inf_{\sigma\in\S(\hil)}\divv(\ext{W}(P)\|P\otimes \sigma).
\end{align}
The mutual information and the weighted divergence radius are different quantities in general (as we see below). 
However, they are equal if the divergence satisfies some simple properties:

\begin{lemma}\label{lemma:info radius equal}
Assume that $\divv$ is block additive and homogeneous. Then 
\begin{align*}
I_{\divv}(W,P)=\chi_{\divv}(W,P)
\end{align*}
for any gcq channel $W$ and input distribution $P$.
\end{lemma}
\begin{proof}
We have
\begin{align*}
I_{\divv}(W,P)
&=
\inf_{\sigma\in\S(\hil)}\divv(\ext{W}(P)\|P\otimes \sigma)\\
&=
\inf_{\sigma\in\S(\hil)}\divv\bz\sum_{x\in\X}P(x)\pr{x}\otimes W(x)\Big\|\sum_{x\in\X}P(x)\pr{x}\otimes \sigma\jz\\
&=
\inf_{\sigma\in\S(\hil)}\sum_{x\in\X}\divv\bz P(x)\pr{x}\otimes W(x)\|P(x)\pr{x}\otimes \sigma\jz\\
&=
\inf_{\sigma\in\S(\hil)}\sum_{x\in\X}P(x)\divv\bz\pr{x}\otimes W(x)\|\pr{x}\otimes \sigma\jz\\
&=
\inf_{\sigma\in\S(\hil)}\sum_{x\in\X}P(x)\divv\bz W(x)\| \sigma\jz\\
&=\chi_{\divv}(W,P),
\end{align*}
where the first two equalities are by definition, the third follows by block additivity,
the fourth by homogeneity, the fifth by isometric invariance (see the beginning of Section 
\ref{sec:gendivrad}), and the last one is again by definition.
\end{proof}

In particular, all $\oll Q_{\alpha,z}$ are block additive and homogenous, and hence we have
\begin{cor}\label{cor:info radius equal}
For any $(\alpha,z)$, any gcq channel $W$ and input distribution $P\in\P_f(\X)$, we have 
\begin{align*}
I_{\oll Q_{\alpha,z}}(W,P)=\chi_{\oll Q_{\alpha,z}}(W,P).
\end{align*}
\end{cor}

Note that the relative entropy $D=D_1$ is also block additive and homogeneous, and hence the corresponding mutual 
information and weighted $D$-radius coincide for any input distribution $P$:
\begin{align}\label{Holevo quantity}
\chi_{1}(W,P)=I_1(W,P).
\end{align}
Moreover, the relative entropy is even more special, as for any cq channel $W$, the $P$-weighted $D$-center 
coincides with the minimizer in \eqref{mutual info def}, and can be explicitly given as $W(P)$. Indeed,
for any state $\sigma\in\S(\hil)$, we have the simple identities (attributed to Donald)
\begin{align*}
D\bz\ext{W}(P)\|P\otimes \sigma\jz
=
\sum_{x\in\X}P(x)D(W(x)\|\sigma)
&=
D(W(P)\|\sigma)+\sum_{x\in\X}P(x)D(W(x)\|W(P))\\
&=
D(W(P)\|\sigma)+D\bz\ext{W}(P)\|P\otimes W(P)\jz,
\end{align*}
and the assertion follows from the strict positivity of the relative entropy on pairs of states.
(Note that for this it is necessary that all the $W(x)$ are normalized on $\supp P$.)
$\chi_{1}(W,P)=I_1(W,P)$
is called the \ki{Holevo quantity} of the ensemble 
$\{W(x),P(x)\}_{x\in\supp P}$ in the quantum information theory literature.
\medskip

It is easy to see that $D_{\alpha,z}$ is not block additive if $\alpha\ne 1$, and in general 
the $D_{\alpha,z}$ mutual information $I_{\alpha,z}(W,P):=I_{D_{\alpha,z}}(W,P)$ 
and the weighted $D_{\alpha,z}$ divergence radius 
$\chi_{\alpha,z}(W,P)$ are different quantities. 
However, they can be related by a simple inequality, as 
\begin{align}
I_{\alpha,z}(W,P)
&=
\inf_{\sigma\in\S(\hil)}\frac{1}{\alpha-1}
\log Q_{\alpha,z}\bz\sum_{x\in\X}P(x)\pr{x}\otimes W(x)\Big\|\sum_{x\in\X}P(x)\pr{x}\otimes \sigma\jz
\label{D info equal Q radius}\\
&=
\inf_{\sigma\in\S(\hil)}\frac{1}{\alpha-1}\log\sum_{x\in\X}P(x)Q_{\alpha,z}(W(x)\|\sigma)
\label{D info equal Q radius2}\\
&\le
\inf_{\sigma\in\S(\hil)}\sum_{x\in\X}P(x)\frac{1}{\alpha-1}\log Q_{\alpha,z}(W(x)\|\sigma)\ds\ds\ds
(\text{if }\alpha\in(0,1))\nonumber\\
&=
\inf_{\sigma\in\S(\hil)}\sum_{x\in\X}P(x)D_{\alpha,z}(W(x)\|\sigma)\nonumber\\
&=
\chi_{\alpha,z}(W,P),\nonumber
\end{align}
where the second equality follows as in the proof of Lemma \ref{lemma:info radius equal}, and the inequality is due to the concavity of the logarithm. Obviously, the inequality holds in the opposite direction when $\alpha>1$.

The above inequality is in general strict. As an example, consider a noiseless channel as in Example \ref{ex:noiseless}
with a non-uniform input distribution $P$. Then we have 
\begin{align*}
I_{\alpha,z}(W,P)
&\le D_{\alpha,z}(\ext{W}(P)\|P\otimes W(P))
=
\frac{1}{\alpha-1}\log\sum_{x\in\X} P(x) Q_{\alpha,z}(W(x)\|W(P))\\
&=
\frac{1}{\alpha-1}\log\sum_{x\in\X} P(x) P(x)^{1-\alpha}
<
\sum_{x\in\X} P(x)\frac{1}{\alpha-1}\log P(x)^{1-\alpha}\\
&=H(P)=\chi_{\alpha,z}(W,P),
\end{align*}
where the first inequality is by definition, the strict inequality follows from the strict concavity of the logarithm, and the last equality is due to Example \ref{ex:noiseless}.

While the mutual information $I_{\alpha,z}$ for $D_{\alpha,z}$ differs from the weighted channel radius 
$\chi_{\alpha,z}$ for $D_{\alpha,z}$, 
it is a simple function of the weighted channel radius for $Q_{\alpha,z}$ (and thus, by Corollary 
\ref{cor:info radius equal}, also of the mutual information for $Q_{\alpha,z}$); indeed, 
moving the infimum over $\sigma$ behind the logarithm in 
\eqref{D info equal Q radius} and \eqref{D info equal Q radius2}, respectively, yields
\begin{align}\label{az mutual info}
I_{\alpha,z}(W,P)
=
\frac{1}{\alpha-1}\log s(\alpha)I_{\oll Q_{\alpha,z}}(W,P)
=
\frac{1}{\alpha-1}\log s(\alpha)\chi_{\oll Q_{\alpha,z}}(W,P).
\end{align}
Moreover, the difference between $I_{\alpha,z}$ and $D_{\alpha,z}$ vanishes after optimizing over the input distribution:
\begin{cor}\label{cor:alpha-z capacities}
If $(\alpha,z)$ are such that $D_{\alpha,z}$ is convex in its second argument then 
\begin{align}\label{az radius1}
\sup_{P\in\P_f(\X)}\chi_{\alpha,z}(W,P)=R_{D_{\alpha,z}}(\ran W),
\end{align}
and if $\alpha=1$ or $\alpha\ne 1$ and $(\alpha,z)$ are such that $\oll Q_{\alpha,z}$ is convex in its second argument 
 then 
\begin{align}\label{az radius2}
\sup_{P\in\P_f(\X)}I_{\alpha,z}(W,P)
&=
R_{D_{\alpha,z}}(\ran W)
\end{align}
for any gcq channel $W:\,\X\to\S(\hil)$.
\end{cor}
\begin{proof}
The identity in \eqref{az radius1} is immediate from Proposition \ref{prop:radius equality}.
When $\alpha=1$, we have $I_{1,z}(W,P)=\chi_{1,z}(W,P)$, and hence for the rest we assume that $\alpha\ne 1$.
The identity in \eqref{az radius2} follows as
\begin{align*}
\sup_{P\in\P_f(\X)}I_{\alpha,z}(W,P)
&=
\frac{1}{\alpha-1}\log s(\alpha)\sup_{P\in\P_f(\X)}\chi_{\oll Q_{\alpha,z}}(W,P)\\
&=
\frac{1}{\alpha-1}\log s(\alpha)R_{\oll Q_{\alpha,z}}(\ran W)\\
&=
\frac{1}{\alpha-1}\log s(\alpha)\inf_{\sigma\in\S(\hil)}\sup_{x\in\X}\oll Q_{\alpha,z}(W(x)\|\sigma)\\
&=
\inf_{\sigma\in\S(\hil)}\sup_{x\in\X}\frac{1}{\alpha-1}\log s(\alpha)\oll Q_{\alpha,z}(W(x)\|\sigma)\\
&=
\inf_{\sigma\in\S(\hil)}\sup_{x\in\X}D_{\alpha,z}(W(x)\|\sigma)\\
&=
R_{D_{\alpha,z}}(\ran W)
\end{align*}
where the first equality is due to \eqref{az mutual info}, 
the second one is due to Proposition \ref{prop:radius equality}, and the rest are obvious.
\end{proof}

\begin{rem}
The above proof method is due to Csisz\'ar \cite{Csiszar}, and extends various prior results
for different quantum R\'enyi divergences and ranges of parameters $(\alpha,z)$
in \cite{KW,MH,MO-cqconv,WWY}. The special case
\begin{align*}
\sup_{P\in\P_f(\X)}\chi_{1}(W,P)=R_{D}(\ran W)
=\sup_{P\in\P_f(\X)}I_{1}(W,P)
\end{align*}
was proved by different methods in \cite{OPW,SW2,TomamichelTan2015}.
\end{rem}

\begin{rem}
Mutual information quantifies the amount of correlation in a bipartite state by measuring its distance from the 
set of uncorrelated states. Following this general idea, one may give two alternative definitions of mutual information between the input and the output of a gcq channel $W$ for a fixed input distribution $P$ using a general divergence $\divv$:
the $\divv$-``distance'' of $\ext{W}(P)$ from the product of its marginals $P\otimes W(P)$:
\begin{align*}
I_{\divv}^{(2)}(W,P):=\divv\bz\ext{W}(P)\|P\otimes W(P)\jz,
\end{align*}
or the $\divv$-``distance'' of $\ext{W}(P)$ from the set of uncorrelated states:
\begin{align*}
I_{\divv}^{(1)}(W,P):=\inf_{\omega\in\S(\hil_{\X}),\,\sigma\in\S(\hil)}\divv\bz\ext{W}(P)\|\omega\otimes \sigma\jz.
\end{align*}
Both of these options seem more natural than the curiously asymmetric optimization in \eqref{mutual info def}, 
and one might wonder why it is nevertheless the capacity 
$\chi_{\alpha}\nw(W):=\sup_{P\in\P_f(\X)} I_{\alpha}\nw(W,P)$ corresponding to the version 
\eqref{mutual info def} (for the sandwiched R\'enyi divergence)
that seems to obtain an operational significance in channel coding, according to \cite{Csiszar,MO-cqconv}. 

There seems to be at least two different resolutions of this question: First, as the more detailed analysis
of constant composition channel coding shows (according to \cite{Csiszar} and our Corollary \ref{prop:cutoff}), it 
is in fact not the mutual information, but the weighted divergence radius that obtains a natural operational interpretation in channel coding; the mutual information only enters the picture because its optimized version over all input distributions ``happens to'' coincide with the optimized version of the divergence radius. 
In this context it would be interesting to know if any of the trivial inequalities 
\begin{align*}
\sup_{p\in\P_f(\X)}I_{\divv}^{(1)}(W,P)
\le
\sup_{p\in\P_f(\X)}I_{\divv}(W,P)
\le
\sup_{p\in\P_f(\X)}I_{\divv}^{(2)}(W,P)
\end{align*}
holds as an equality for general $W$ and $P$ (at least for $\divv=D_{\alpha}\nw$ with $\alpha>1$).

Second, according to \eqref{az mutual info}, the mutual information for R\'enyi divergences is simply a function of another weighted divergence radius, corresponding to the $\oll Q$ quantities rather than the R\'enyi divergences, 
and the optimization over $\sigma$ is simply the optimization over the candidates for the weighted divergences centers, which is very natural, and in this context the problem of asymmetry does not even makes sense. 
\end{rem}
\medskip

The same arguments as in Sections \ref{sec:Renyi center} and \ref{sec:additivity} yield the following statements, and hence we omit their proofs.

\begin{lemma}\label{lemma:minimizer support Q}
Let $\sigma$ be a $P$-weighted $\oll Q_{\alpha,z}$ center for $W$.
\smallskip

\noindent (1) If $(\alpha,z)$ is such that $\oll Q_{\alpha,z}$ is quasi-convex in its second argument
then $\sigma^0\le W(P)^0$. 
\smallskip

\noindent (2) If
$\alpha> 1$ or $\alpha\in(0,1)$ and 
$1-\alpha<z<+\infty$ then 
$W(P)^0\le\sigma^0$.
\end{lemma}

Let us define $\Gamma_{\oll Q}$ to be the set of $(\alpha,z)$ values such that 
for any gcq channel $W$ and any input probability distribution $P$, 
any $P$-weighted $\oll Q_{\alpha,z}$ center $\sigma$ for $W$ satisfies
$\sigma^0=W(P)^0$.
Then Lemmas \ref{lemma:az monotonicity}, \ref{lemma:2nd convexity}, and \ref{lemma:minimizer support Q} yield
\begin{align*}
\Gamma_{\oll Q}\supseteq\left\{(\alpha,z):\,\alpha\in(0,1),\,1-\alpha< z+\infty\right\}
\cup
\left\{(\alpha,z):\,\alpha>1,\, z\ge\max\{\alpha/2,\alpha\}\right\}.
\end{align*}

\begin{prop}\label{prop:fixed point characterization Q}
Assume that $(\alpha,z)\in\Gamma_{\oll Q}$ are such that 
$\oll Q_{\alpha,z}$ is convex in its second variable.
Then $\sigma$ is a $P$-weighted $\oll Q_{\alpha,z}$ center for $W$
if and only if it is a fixed point of 
the map
\begin{align}\label{Q fixedpoint eq}
\map_{W,P,\oll Q_{\alpha,z}}(\sigma):=&
\frac{1}{\tau}\sum_{x\in\X}P(x)\bz \sigma^{\frac{1-\alpha}{2z}} W(x)^{\frac{\alpha}{z}}\sigma^{\frac{1-\alpha}{2z}}\jz^{z},
\end{align}
$\sigma\in\S_{W,P}(\hil)_{++}$, where $\tau$ is the normalization factor
\begin{align*}
\tau := \sum_{x\in\X}P(x)Q_{\alpha,z}(W(x)\|\sigma).
\end{align*}
\end{prop}

\begin{rem}\label{mutual info remark1}
Note that the fixed point equations in \eqref{Q fixedpoint eq} and
\eqref{D fixed point eq} look very similar, except that the normalization of sigma is obtained 
``globally'' in the former, and for each $x$ individually in the latter.
\end{rem}

\begin{rem}\label{rem:z=1 power mean}
As we have mentioned above, the relative entropy ($\alpha=1$) is special in that for cq channels the minimizer for the mutual information can be explicitly determined (as $W(P)$), and hence also the mutal information can be given by an explicit formula. The only other known family of quantum R\'enyi divergences with these properties are the 
Petz-type R\'enyi divergences (corresponding to $z=1$). Indeed, it is easy to verify that 
in this case we have the classical-quantum Sibson identity \cite[Lemma 2.2]{KW}
$D_{\alpha,1}(\ext{W}(P)\|P\otimes\sigma)=D_{\alpha,1}(\sigma_{\alpha,1}\|\sigma)+\frac{1}{\alpha-1}
\log(\Tr\omega(\alpha))^{\alpha}$, where 
\begin{align*}
\omega(\alpha):=\bz\sum_x P(x)W(x)^{\alpha}\jz^{1/\alpha},\ds\ds\text{and}\ds\ds
\sigma_{\alpha,1}:=\omega(\alpha)/\Tr\omega(\alpha).
\end{align*}
As a consequence, $\sigma_{\alpha,1}$ is the unique 
minimizer for $I_{\alpha,1}(W,P)$, and, by \eqref{az mutual info}, it is also the unique
$P$-weighted $\oll Q_{\alpha,1}$ center. Moreover,
\begin{align*}
\chi_{\oll Q_{\alpha,1}}(W,P)&=\sum_{x\in\X}P(x)\oll Q_{\alpha,1}(W(x)\|\sigma_{\alpha,1})=
s(\alpha)(\Tr\omega(\alpha))^{\alpha}
=
s(\alpha)\bz\Tr\bz\sum_x P(x)W(x)^{\alpha}\jz^{1/\alpha}\jz^{\alpha}.
\end{align*}
\end{rem}

\begin{prop}\label{prop:radius multiplicativity}
\ki{(Additivity of the $D_{\alpha,z}$ mutual information)}
Let $W^{(1)}:\,\X^{(1)}\to\S(\hil^{(1)})$ and 
$W^{(2)}:\,\X^{(2)}\to\S(\hil^{(2)})$ be gcq channels, 
and $P^{(i)}\in\P_f(\X^{(i)})$, $i=1,2$, be input distributions.
Assume, moreover, that $\alpha$ and $z$ satisfy the conditions of 
Proposition \ref{prop:fixed point characterization Q}.
Then
\begin{align}
\chi_{\oll Q_{\alpha,z}}\bz W^{(1)}\otimes W^{(2)},P^{(1)}\otimes P^{(2)}\jz=
\chi_{\oll Q{\alpha,z}}\bz W^{(1)},P^{(1)}\jz
\cdot
\chi_{\oll Q_{\alpha,z}}\bz W^{(2)},P^{(2)}\jz,\label{Qrad mult}\\
I_{\alpha,z}\bz W^{(1)}\otimes W^{(2)},P^{(1)}\otimes P^{(2)}\jz=
I_{\alpha,z}\bz W^{(1)},P^{(1)}\jz
+
I_{\alpha,z}\bz W^{(2)},P^{(2)}\jz.\label{mutual info add}
\end{align}
\end{prop}

\begin{rem}
In \cite{HT14}, a generalized notion of mutual information was studied, where the bipartite state need not be 
classical-quantum, and the first marginal of the second argument is fixed but not necessarily equal to the first marginal of the first argument. For sandwiched R\'enyi divergences a characterization of the optimal state in terms of a fixed point equation, as well as the additivity of the generalized mutual information was obtained. Our approach is essentially the same as that of \cite{HT14}, and 
Propositions \ref{prop:fixed point characterization Q} and 
\ref{prop:radius multiplicativity} are special cases of the results of \cite{HT14} when 
 $z=\alpha$.
\end{rem}

\begin{rem}
Note that for the special case $z=1$, the relations \eqref{Qrad mult} and \eqref{mutual info add} follow immediately from the explicit expression for the minimizer in Remark \ref{rem:z=1 power mean}.
\end{rem}

Finally, Corollary \ref{cor:alpha-z capacities}, Propositions \ref{prop:divrad additive} and \ref{prop:radius multiplicativity}, and Theorem \ref{thm:additivity} yield the following:

\begin{prop}\label{prop:Renyi cap additivity}
Assume that $(\alpha,z)$ are such that $(\alpha,z)\in\Gamma_D$ and 
$D_{\alpha,z}$ is convex in its second variable, or 
$(\alpha,z)\in\Gamma_{\oll Q}$ and
$\oll Q_{\alpha,z}$ is convex in its second variable.
Then the $(\alpha,z)$-capacity is additive in the sense that for any qcq channels
$W^{(i)}:\,\X^{(i)}\to\B(\hil^{(i)})_+$, $i=1,2$,
\begin{align*}
\chi_{\alpha,z}(W^{(1)}\otimes W^{(2)})=
\chi_{\alpha,z}(W^{(1)})+\chi_{\alpha,z}(W^{(2)}).
\end{align*}
\end{prop}
\begin{proof}
The assertion under the first set of conditions follows immediately from 
Theorem \ref{thm:additivity}, Corollary \ref{cor:alpha-z capacities}, and Proposition \ref{prop:divrad additive}, as we have already essentially stated in Remark \ref{rem:capacity additivity}. 
Corollary \ref{cor:alpha-z capacities}, and Propositions \ref{prop:divrad additive} and \ref{prop:radius multiplicativity}
yield the assertion under the second set of conditions by a completely similar argument.
\end{proof}

\begin{rem}\label{rem:Renyi cap additivity}
Additivity of the $\chi_{\alpha,1}$ capacities for classical-quantum channels and $\alpha>1$ was first proved in \cite[Lemma 2]{ON99} by different methods from the above, following the method of Arimoto \cite{A73}.
We credit the book \cite{Hayashibook2}  (especially Exercise 4.30)
for the idea of the above approach, using the identity in \eqref{az radius2} and the additivity 
of the mutual information. We remark that in \cite{Hayashibook2}, only the case 
$z=1$ was considered, where additivity of the mutual information can be obtained without Proposition \ref{prop:radius multiplicativity}, as pointed out in Remark \ref{rem:z=1 power mean}.
\end{rem}

\subsection{PSD divergence center and radius}

Note that while we defined the divergence radius and center for an arbitrary non-empty subset $S$ of $\B(\hil)_+$, in 
the definition of the center we restricted to density operators, i.e., PSD operators with trace $1$. 
This is a natural choice when the set $S$ consists of density operators itself, and, even more importantly, it leads to operationally relevant information measures as our main result, Theorem \ref{thm:main result} shows. 

Moreover, restricting the set of possible divergence centers to density operators may be operationally motivated even when the elements of the set $S$ are not normalized to have trace $1$. Indeed, the optimal success probability of discriminating quantum states $\rho_1,\ldots,\rho_r$ with prior probabilities $p_1,\ldots,p_r$ can be expressed as
\begin{align*}
P_s^*=\exp\bz R_{D_{\infty}\nw}(\{p_1\rho_1,\ldots,p_r\rho_r\})\jz,
\end{align*}
where $D_{\infty}\nw:=\lim_{\alpha\to+\infty}D_{\alpha}\nw=\lim_{\alpha\to+\infty}D_{\alpha,\alpha}$ is the max-relative entropy \cite{RennerPhD,Datta,Renyi_new}. This follows by simply rewriting the optimal success probability $P_s^*:=\max\{\sum_{i=1}^rp_i\Tr\rho_i M_i:\,(M_i)_{i=1}^r\text{ POVM}\}$ using the duality of linear programming, as was done in \cite{YKL} (see also \cite{KRS}).

In this section we consider the alternative approach where the divergence center is allowed to be a general PSD 
operator. This is largely motivated by recent investigations in matrix analysis regarding various concepts of 
multi-variate geometric matrix means, in particular, by the approach of \cite{BGJ2019}. We comment on this in more detail at the end of the section.

For a general divergence $\divv$, we define the \ki{$P$-weighted PSD $\divv$-radius} as
\begin{align*}
\wtilde R_{\divv,P}(S):=\inf_{\sigma\in\B(\hil)_+}\sum_{\rho\in S}P(\rho)\divv(\rho\|\sigma),
\end{align*} 
and we call any $\sigma\in\B(\hil)_+$ that attains the above infimum a \ki{$P$-weighted PSD $\divv$-center}. For a gcq channel $W:\,\X\to\B(\hil)_+$ and $P\in\P_f(\X)$ we define the
$P$-weighted PSD $\divv$-radius of $W$ as before:
\begin{align*}
\wtilde\chi_{\divv}(W,P):=\tilde R_{\divv,P\circ W\inv}(\ran W).
\end{align*}
Any PSD minimizer in the definition of $\tilde R_{\divv,P\circ W\inv}(\ran W)$ will be called a 
$P$-weighted PSD $\divv$-center for the channel $W$.

It is easy to see that these quantities are meaningless for the R\'enyi divergences considered before, as 
we always have
\begin{align*}
\wtilde R_{D_{\alpha,z},P}(S)=-\infty,\ds\ds\ds
\wtilde R_{\oll Q_{\alpha,z},P}(S)=\begin{cases}-\infty,&\alpha\in(0,1),\\ 0,&\alpha>1,\end{cases}
\end{align*}
due to the scaling laws
\begin{align*}
D_{\alpha,z}(\rho\|\lambda\sigma)=D_{\alpha,z}(\rho\|\sigma)-\log\lambda,\ds\ds\ds
\oll Q_{\alpha,z}(\rho\|\lambda\sigma)=\lambda^{1-\alpha}\oll Q_{\alpha,z}(\rho\|\sigma),\ds\ds\ds
\lambda\in(0,+\infty).
\end{align*}
Hence, in order to make sense of the PSD divergence radius, it seems necessary to modify the notion of 
R\'enyi divergence for PSD operators. 

We consider two such options, motivated by Proposition \ref{lemma:non-negative} in Section \ref{sec:positive Renyi}. One is a simple rescaling of $D_{\alpha,z}$, defined as
\begin{align*}
\what D_{\alpha,z}(\rho\|\sigma)&:=
\frac{1}{\alpha-1}\log\frac{Q_{\alpha,z}(\rho\|\sigma)}{(\Tr\rho)^{\alpha}(\Tr\sigma)^{1-\alpha}}
=
\frac{1}{\alpha-1}\log Q_{\alpha,z}(\rho\|\sigma)-\frac{\alpha}{\alpha-1}\log\Tr\rho+\log\Tr\sigma\\
&=
D_{\alpha,z}-\log\Tr\rho+\log\Tr\sigma
=
D_{\alpha,z}\bz\frac{\rho}{\Tr\rho}\Big\|\frac{\sigma}{\Tr\sigma}\jz.
\end{align*}
The limit $\alpha\to 1$ yields
\begin{align*}
\what D_1(\rho\|\sigma):=\lim_{\alpha\to 1}\what D_{\alpha,z(\alpha)}(\rho\|\sigma)
=
D_1(\rho\|\sigma)+\log\Tr\rho-\log\Tr\sigma=D\bz\frac{\rho}{\Tr\rho}\Big\|\frac{\sigma}{\Tr\sigma}\jz
\end{align*}
for any function $\alpha\mapsto z(\alpha)$ that is continuously differentiable in a neighbourhood of $1$, 
on which $z(\alpha)\ne 0$
\cite{LinTomamichel15}.
Obviously, $\what D_{\alpha,z}(\rho\|\sigma)=D_{\alpha,z}(\rho\|\sigma)$ for any pair of states
$\rho,\sigma$.
Note that $\what D_{\alpha,z}$ is a \ki{projective divergence}, i.e., 
$\what D_{\alpha,z}(\lambda\rho\|\sigma)=\what D_{\alpha,z}(\rho\|\lambda\sigma)=\what D_{\alpha,z}(\rho\|\sigma)$
for any $\rho,\sigma\in\B(\hil)_+$ and $\lambda\in(0,+\infty)$.
As an immediate consequence, we have
\begin{align*}
\wtilde R_{\what D_{\alpha,z},P}(S)=R_{D_{\alpha,z},P}(S)-\sum_{\rho}P(\rho)\log\Tr\rho,
\end{align*}
i.e., we essentially recover the previously considered concept of the 
$D_{\alpha,z}$-radius, apart from an uninteresting term. Obviously, if $\sigma$ is a PSD $\what D_{\alpha,z}$-center then so is $\lambda\sigma$ for all $\lambda\in(0,+\infty)$. Moreover, the normalized PSD $\what D_{\alpha,z}$-centers can be characterized by the fixed point equation in Theorem \ref{prop:fixed point characterization}
for the $(\alpha,z)$ pairs for which Theorem \ref{prop:fixed point characterization} holds.

The other option we consider is a modification of the Tsallis quantum R\'enyi divergences, 
or Tsallis relative entropies,
defined as
\begin{align*}
T_{\alpha,z}(\rho\|\sigma):=\frac{1}{1-\alpha}\bz\alpha\Tr\rho+(1-\alpha)\Tr\sigma-Q_{\alpha,z}(\rho\|\sigma)\jz.
\end{align*}
We will call these quantities Tsallis $(\alpha,z)$-divergences.
The limit $\alpha\to 1$ yields
\begin{align*}
T_1(\rho\|\sigma):=\lim_{\alpha\to 1}T_{\alpha,z(\alpha)}(\rho\|\sigma)
=
D(\rho\|\sigma)-\Tr\rho+\Tr\sigma,
\end{align*}
for any function $\alpha\mapsto z(\alpha)$ that is continuously differentiable in a neighbourhood of $1$, 
on which $z(\alpha)\ne 0$
\cite{LinTomamichel15}.

\begin{remark}
The usual way to define the quantum Tsallis relative entropy is 
\begin{align*}
T'_{\alpha}(\rho\|\sigma):=\frac{1}{1-\alpha}\bz\Tr\rho-Q_{\alpha,1}(\rho\|\sigma)\jz.
\end{align*}
Note that $T'_{\alpha}$ coincides with $T_{\alpha,1}$ on pairs of density operators but,
unlike $T_{\alpha,1}$, $T'_{\alpha}$ is not positive on pairs of PSD operators for any $(\alpha,z)$. Moreover, the 
PSD divergence radius problem is trivial for these quantities, as
\begin{align*}
\wtilde R_{T'_{\alpha,z}}(S)=\begin{cases}-\infty,&\alpha\in(0,1),\\ \frac{1}{1-\alpha}\sum_{\rho}P(\rho)\rho,&\alpha>1.\end{cases}
\end{align*}
\end{remark}

We consider the PSD divergence center for $T_{\alpha,z}$ in the gcq channel formalism, for easier comparison with 
Theorem \ref{prop:fixed point characterization} and Proposition \ref{prop:fixed point characterization Q}.

\begin{prop}\label{prop:Tsallis center}
Let $W:\,\X\to\B(\hil)_+$ be a qcq channel.
\begin{enumerate}
\item\label{Tsallis rad1}
For any $\alpha\in(0,+\infty)\setminus\{1\}$ and $z\in(0,+\infty)$, 
\begin{align}\label{Tsallis radius}
\wtilde \chi_{T_{\alpha,z}}(W,P)=\frac{\alpha}{1-\alpha}\left[\Tr W(P)-\bz s(\alpha)\chi_{\oll Q_{\alpha,z}}(W,P)\jz^{1/\alpha}\right].
\end{align}
\item\label{Tsallis rad2}
If $\sigma$ is a $P$-weighted PSD $T_{\alpha,z}$-center for $W$ 
then $\Tr\sigma=\sum_x P(x)Q_{\alpha,z}(W(x)\|\sigma)$, and
$\oll\sigma=\sigma/\Tr\sigma$ is a $P$-weighted $\oll Q_{\alpha,z}$-center for $W$.
\item\label{Tsallis rad3}
If $\oll\sigma\in\S(\hil)$ is a $P$-weighted $\oll Q_{\alpha,z}$-center for $W$ 
then 
\begin{align*}
\sigma:=\bz\sum_x P(x)Q_{\alpha,z}(W(x)\|\oll\sigma)\jz^{1/\alpha}\oll\sigma
\end{align*}
is a $P$-weighted PSD $T_{\alpha,z}$-center for $W$.

\item\label{Tsallis rad4}
If $(\alpha,z)$ satisfy the conditions of Proposition \ref{prop:fixed point characterization Q} then 
$\sigma$ is a $P$-weighted PSD $T_{\alpha,z}$-center for $W$
if and only if it is a solution of the fixed point equation
\begin{align}\label{Tsallis center}
\sigma=\sum_{x\in\X}P(x)\bz\sigma^{\frac{1-\alpha}{2z}}W(x)^{\frac{\alpha}{z}}\sigma^{\frac{1-\alpha}{2z}}\jz^z.
\end{align}
\end{enumerate}
\end{prop}
\begin{proof}
We have 
\begin{align*}
\wtilde \chi_{T_{\alpha,z}}(W,P)&=
\inf_{\sigma\in\B(\hil)_+}\sum_x P(x)T_{\alpha,z}(W(x)\|\sigma)\\
&=
\frac{\alpha}{1-\alpha}\Tr W(P)+\inf_{\sigma\in\B(\hil)_+}\left[\Tr\sigma-\frac{1}{1-\alpha}\sum_x P(x)Q_{\alpha,z}(W(x)\|\sigma)\right]\\
&=
\frac{\alpha}{1-\alpha}\Tr W(P)+\inf_{\oll\sigma\in\S(\hil)}\inf_{\lambda>0}
\left[\lambda-\frac{1}{1-\alpha}\lambda^{1-\alpha}\sum_x P(x)Q_{\alpha,z}(W(x)\|\oll\sigma)\right].
\end{align*}
Differentiating w.r.t.~$\lambda$ yields that the optimal $\lambda$ is 
\begin{align}\label{lambda}
\lambda
&=
\bz\sum_x P(x)Q_{\alpha,z}(W(x)\|\oll\sigma)\jz^{1/\alpha}
=
\sum_x P(x)Q_{\alpha,z}(W(x)\|\sigma),
\end{align}
with $\sigma:=\lambda\oll\sigma$.
Writing it back to the previous equation, we get 
\begin{align*}
\wtilde \chi_{T_{\alpha,z}}(W,P)&=
\frac{\alpha}{1-\alpha}\Tr W(P)+
\inf_{\oll\sigma\in\S(\hil)}\frac{\alpha}{\alpha-1}\bz\sum_x P(x)Q_{\alpha,z}(W(x)\|\oll\sigma)\jz^{1/\alpha},
\end{align*}
which is exactly \eqref{Tsallis radius}. This proves \ref{Tsallis rad1}, and \ref{Tsallis rad2} and \ref{Tsallis rad3} are clear from the above argument.
Finally, \ref{Tsallis rad4} is immediate from the above and Proposition \ref{prop:fixed point characterization Q}.
\end{proof}

\begin{rem}\label{rem:general Tsallis center}
Note that \ref{Tsallis rad1}--\ref{Tsallis rad3} of Proposition \ref{prop:Tsallis center} hold true for any pair of divergences 
$T_{\alpha}^q$ and $Q_{\alpha}^q$ related as 
\begin{align*}
T^q_{\alpha}(\rho\|\sigma)=\frac{1}{1-\alpha}\bz\alpha\Tr\rho+(1-\alpha)\Tr\sigma-Q^q_{\alpha}(\rho\|\sigma)\jz,
\end{align*}
provided that $Q_{\alpha}^q$ is a non-negative divergence that satisfies the scaling property 
$Q_{\alpha}^q(\rho\|\lambda\sigma)=\lambda^{\alpha-1}Q_{\alpha}^q(\rho\|\sigma)$, and that there exists a $\sigma\in\B(\hil)_+$ such that $Q_{\alpha}^q(W(x)\|\sigma)<+\infty$ for all $x\in\supp P$. 
\end{rem}

The above Proposition extends Theorem 8 in \cite{BJL2018}, where the fixed point characterization \eqref{Tsallis center} was obtained in the case $z=\alpha=1/2$, by 
a somewhat different proof than above.
The existence of a solution of the fixed point equation \eqref{Tsallis center}
was studied in \cite[Theorem 6.1]{AC2011} for the case $z=\alpha=1/2$, and their proof 
extends without alteration for more general $(\alpha,z)$ pairs as below.

\begin{prop}
Let $\alpha\in(0,1)$.
If there exist positive numbers $\lambda,\eta$, such that 
$W(x)\in[\lambda I,\eta I]:=\{A\in\B(\hil)_+:\,\lambda I\le A\le\eta I\}$ 
for all $x\in\supp P$ then the fixed point equation \eqref{Tsallis center}
has a solution, which is also in $[\lambda I,\eta I]$.
\end{prop}
\begin{proof}
It is easy to see that the map on the RHS of \eqref{Tsallis center} maps the compact convex set $[\lambda I,\eta I]$ into itself, and hence, by Brouwer's fixed point theorem, it has a fixed point.
\end{proof}

\begin{rem}\label{rem:z=1 power mean2}
The case of the Petz-type Tsallis divergences ($z=1$) is again special in that the weighted PSD
$T_{\alpha,1}$ center and radius can be given explicitly. Indeed, by Remark \ref{rem:z=1 power mean}
and Proposition \ref{prop:Tsallis center}, the unique $T_{\alpha,1}$ center is given by 
\begin{align*}
\tilde\sigma_{\alpha,1}=\bz\sum_x P(x)Q_{\alpha,1}(W(x)\|\sigma_{\alpha,1})\jz^{1/\alpha}
\frac{\omega(\alpha)}{\Tr\omega(\alpha)}
=
\omega(\alpha)=\bz\sum_x P(x)W(x)^{\alpha}\jz^{1/\alpha},
\end{align*}
and
\begin{align*}
\wtilde \chi_{T_{\alpha,z}}(W,P)&=
\frac{\alpha}{\alpha-1}\left[\Tr W(P)-\Tr\bz\sum_x P(x)W(x)^{\alpha}\jz^{1/\alpha}\right].
\end{align*}
\end{rem}
\bigskip

Let us now explain how the above considerations are related to recent research in matrix analysis. First, we 
recall that if $\rho_1,\ldots\rho_r$ are points in a metric space $(M,d)$, and $(p_i)_{i=1}^r$ is a probability 
distribution then the 
\ki{$p$-weighted Fr\'echet variance} is $\inf_{\sigma\in M}\sum_i p_id^2(\rho_i,\sigma)$,
and if $\sigma$ attains this infimum then it is called a
\ki{$p$-weighted Fr\'echet mean} of $\rho_1,\ldots\rho_r$.
These are analogous to our notions of weighted divergence radius and center. Note, however, that we only defined 
divergences on PSD operators, which is more restrictive than the setting of the Fr\'echet means, while 
it is also more general in the sense that a divergence does not need to be the square of a metric. 
(Although some properties of the relative entropy are reminiscent to those of a squared Euclidean distance.)
In particular, if 
$f$ is an injective continuous function on an interval $M\subseteq\bR$, then $d(x,y):=|f(x)-f(y)|$ is a metric on 
$M$, and a straightforward computation shows that 
for any $\rho_1,\ldots,\rho_r\in M$, and any probability distribution $(p_i)_{i=1}^r$,
there is a unique Fr\'echet mean, which is exactly the generalized $f$-mean
$f\inv\bz\sum_{i=1}^rp_if(\rho_i)\jz$, and the Fr\'echet variance is 
$\sum_ip_if(\rho_i)^2-\bz\sum_ip_if(\rho_i)\jz^2$. 

Of particular importance to us are the 
cases $M:=(0,+\infty)$, $f(t):=t^{\alpha}$, which yields the \ki{$\alpha$-power mean}
$\bz\sum_i p_i\rho_i^{\alpha}\jz^{1/\alpha}$, and $f(t):=\log t$ on the same set, which yields the 
\ki{geometric mean} $\prod_i\rho_i^{p_i}$. These special cases are closely related to each other, as the geometric mean can be recovered from the $\alpha$-power mean in the limit $\alpha\searrow 0$,
\begin{align*}
\lim_{\alpha\searrow 0}\bz\sum_i p_i\rho_i^{\alpha}\jz^{1/\alpha}=\prod_i\rho_i^{p_i}, 
\end{align*}
while the $\alpha$-power mean is the unique solution of the fixed point equation 
\begin{align}\label{alpha power fixed point}
\sigma=\sum_i p_i\G_{\alpha}(\rho_i\|\sigma),
\end{align}
where $\G_{\alpha}(\rho_i\|\sigma):=\rho_i^{\alpha}\sigma_i^{1-\alpha}$ is the $\alpha$-geometric mean of $\rho_i$ and $\sigma$.

There are various ways to extend the $\alpha$-geometric mean to operators, and these are closely related to different definitions of R\'enyi divergences. Indeed, for every $(\alpha,z)$, the quantity 
\begin{align*}
\G_{\alpha,z}(\rho\|\sigma)
:=
\bz\sigma^{\frac{1-\alpha}{2z}}\rho^{\frac{\alpha}{z}}\sigma^{\frac{1-\alpha}{2z}}\jz^z
\end{align*}
is well-defined if $\alpha\in(0,1)$, or $\alpha>1$ and $\rho^0\le\sigma^0$, and it
reduces to the $\alpha$-geometric mean for scalars for every $z>0$.
This is related to the $(\alpha,z)$-R\'enyi divergences via $Q_{\alpha,z}(\rho\|\sigma)=\Tr\G_{\alpha,z}(\rho\|\sigma)$. The fixed point equation \eqref{Tsallis center} is an exact analogue of \eqref{alpha power fixed point}, and hence we may call the solution of \eqref{Tsallis center} the $P$-weighted $(\alpha,z)$-power mean 
of $(W(x))_{x\in\X}$, and denote it as $\P_{\alpha,z}(W,P)$, provided that it exists and is unique. (From here we switch to the gcq channel formalism for better comparison with the preceding 
part of the paper, but this is equivalent to considering subsets of PSD operators and probability distributions on them.) That is, the $P$-weighted $(\alpha,z)$-power mean is nothing but the 
$P$-weighted $T_{\alpha,z}$-center for $W$ 
(or, in other terminology, the $P\circ W\inv$-weighted $T_{\alpha,z}$-center of $\ran W$). In general, there is no explicit formula for it;
the case $z=1$, discussed in Remark \ref{rem:z=1 power mean2}, is an exception, and the resulting formula 
is probably the most straightforward extension of the $\alpha$-power mean from numbers to operators.
Following the ideas of \cite{LimPalfia2012}, a family of multivariate geometric means for operators may be defined
as $\G(W,P):=\lim_{\alpha\searrow 0}\G_{\alpha,z(\alpha)}(W,P)$, where $z(\alpha)$ is some well-behaved funcion of 
$\alpha$. It is an interesting question if this limit always exists, how it depends on the choice of $z(\alpha)$, 
and how it relates to other notions of multivariate geometric means.

Probably the most studied notion of $\alpha$-geometric mean for a pair of operators is the \ki{Kubo-Ando geometric mean} \cite{KA}
\begin{align*}
\G_{\alpha}^{\max}(\rho\|\sigma):=\sigma^{1/2}\bz\sigma^{-1/2}\rho\sigma^{-1/2}\jz^{\alpha}\sigma^{1/2},
\end{align*}
introduced by Kubo and Ando for $\alpha\in[0,1]$. This is also a special instance of a \ki{maximal $f$-divergence} \cite{Ma,HiaiMosonyi2017} (with a minus sign for $\alpha\in(0,1)$), and can be extended to $\alpha>1$. It gives rise to the \ki{maximal R\'enyi divergence} \cite{TomamichelBook} $D_{\alpha}^{\max}$ and the maximal Tsallis divergence $T_{\alpha}^{\max}$ via
\begin{align*}
D_{\alpha}^{\max}(\rho\|\sigma)&:=\frac{1}{\alpha-1}\log\frac{1}{\Tr\rho} Q_{\alpha}^{\max}(\rho\|\sigma),\\
T_{\alpha}^{\max}(\rho\|\sigma)&:=\frac{1}{1-\alpha}\bz\alpha\Tr\rho+(1-\alpha)\Tr\sigma-Q_{\alpha}^{\max}(\rho\|\sigma)\jz,
\end{align*}
where $Q_{\alpha}^{\max}(\rho\|\sigma):=\Tr\G_{\alpha}^{\max}(\rho\|\sigma)$ for positive definite $\rho$ and $\sigma$, and it is extended to general PSD operators via the smoothing procedure in \eqref{ext0}. (In fact, $T_{\alpha}^{\max}$ is also a maximal $f$-divergence, corresponding to the convex function $f(t)=\frac{1}{1-\alpha}(\alpha t+(1-\alpha)-t^{\alpha})$, which is operator convex if and only if $\alpha\in[0,2]$.)
The positive version of $D_{\alpha}^{\max}$ can be defined again as 
$\what D_{\alpha}^{\max}(\rho\|\sigma):=D_{\alpha}^{\max}\bz\frac{\rho}{\Tr\rho}\big\|\frac{\sigma}{\Tr\sigma}\jz$. Lim and P\'alfia \cite{LimPalfia2012} showed that when all $W(x)$ are positive definite, the fixed point equation 
\begin{align*}
\sigma=\sum_x P(x)\G_{\alpha}^{\max}(W(x)\|\sigma)
\end{align*}
has a unique solution, which we will call the \ki{max $\alpha$-power mean} and denote it as
$\P_{\alpha}^{\max}(W,P)$. Moreover, 
\begin{align*}
\lim_{\alpha\searrow 0}\P_{\alpha}^{\max}(W,P)=\G_K(W,P),
\end{align*}
where the latter is the \ki{Karcher mean}, which, in our terminology, is nothing else but the 
$P\circ W\inv$-weighted PSD $(D_{\infty}\nw)^2$-center of $\ran W$, i.e., 
\begin{align*}
\G_K(W,P)=\argmin_{\sigma\in\B(\hil)_+}\sum_x P(x)(D_{\infty}\nw(W(x)\|\sigma))^2.
\end{align*}
By Remark \ref{rem:general Tsallis center}, 
\ref{Tsallis rad1}--\ref{Tsallis rad3} of Proposition \ref{prop:Tsallis center} hold true
for the pair $T_{\alpha}^{\max}$ and $Q_{\alpha}^{\max}$.
However, it has been shown in \cite{PV2019} that \ref{Tsallis rad4} of Proposition \ref{prop:Tsallis center}
is no longer true in this case. More precisely, an example is shown in \cite{PV2019}
for $W$ and $P$ for which the max $1/2$ power mean does not coincide with the 
$P$-weighted 
$T_{1/2}^{\max}$-center for $W$.

\subsection{Positive R\'enyi divergences}
\label{sec:positive Renyi}

In this section we establish the non-negativity of $\what D_{\alpha,z}$ and
$T_{\alpha,z}$ for all $\alpha\in(0,+\infty)\setminus\{1\}$ and $z\in(0,+\infty]$.

\begin{prop}\label{lemma:non-negative}
For every $\rho,\sigma\in\B(\hil)_+$, and every $z\in(0,+\infty]$, we have 
\begin{align}
Q_{\alpha,z}(\rho\|\sigma)&\le(\Tr\rho)^{\alpha}(\Tr\sigma)^{1-\alpha}\le\alpha\Tr\rho+(1-\alpha)\Tr\sigma,
& &\alpha\in(0,1),\label{nonneg1}\\
Q_{\alpha,z}(\rho\|\sigma)&\ge(\Tr\rho)^{\alpha}(\Tr\sigma)^{1-\alpha}\ge\alpha\Tr\rho+(1-\alpha)\Tr\sigma,
& &\alpha>1,\label{nonneg2}
\end{align}
or equivivalently, 
$\what D_{\alpha,z}(\rho\|\sigma)\ge 0$ and 
$T_{\alpha,z}(\rho\|\sigma)\ge 0$ for all $\alpha\in(0,+\infty)\setminus\{1\}$ and $z\in(0,+\infty]$.
\end{prop}
\begin{proof}
Note that the bounds are independent of $z$; in particular, it is enough to prove them for all finite $z$, and the case $z=+\infty$ follows by taking the limit $z\to+\infty$.
Note also that the second inequalities in \eqref{nonneg1}--\eqref{nonneg2} follow 
by the trivial identity $x^{\alpha}y^{1-\alpha}=y(x/y)^{\alpha}$, and lower bounding the convex function 
$t\mapsto s(\alpha)t^{\alpha}$ on $[0,+\infty)$ by its tangent line at $1$.

To prove the first inequalities, we may assume w.l.o.g.~that $\rho$ and $\sigma$ are positive definite, due to \eqref{ext0}. Assume first that $\rho$ and $\sigma$ are both diagonal in the same orthonormal basis with diagonal elements $r_1,\ldots,r_d$ and $s_1,\ldots,s_d$, respectively, where $d:=\dim\hil$.
Concavity of $t\mapsto t^{\alpha}$ for $\alpha\in(0,1)$ yields
\begin{align}\label{Renyi pos1}
\frac{(\Tr\rho)^{\alpha}}{(\Tr\sigma)^{\alpha}}
=
\bz\sum_{i=1}^d \frac{s_i}{\Tr\sigma}\cdot\frac{r_i}{s_i}\jz^{\alpha}
\ge
\sum_{i=1}^d \frac{s_i}{\Tr\sigma}\cdot\bz\frac{r_i}{s_i}\jz^{\alpha}
=
(\Tr\sigma)\inv\sum_{i=1}^d r_i^{\alpha}s_i^{1-\alpha},
\end{align}
which is exactly the first inequality in \eqref{nonneg1}.
When $\alpha>1$, the inequality in \eqref{Renyi pos1} is in the opposite direction, which yields the first inequality in \eqref{nonneg2}.

For the general case, we follow the approach at the beginning of Section 3 in \cite{AH2019}.
Let $s_1(X)\ge\ldots\ge s_d(X)$ denote the decreasingly ordered singular values of an operator $X$. 
By the Gelfand-Naimark majorization theorem (see, e.g.~\cite[Theorem 4.3.4]{Hiai_book}), we have 
\begin{align*}
\prod_{j=1}^k s_{i_j}\bz\sigma^{\frac{1-\alpha}{2z}}\jz s_{d+1-i_j}\bz\rho^{\frac{\alpha}{2z}}\jz
\le
\prod_{i=1}^k s_i\bz\sigma^{\frac{1-\alpha}{2z}}\rho^{\frac{\alpha}{2z}}\jz
\le
\prod_{i=1}^k s_i\bz\sigma^{\frac{1-\alpha}{2z}}\jz s_i\bz\rho^{\frac{\alpha}{2z}}\jz
\end{align*}
for every $1\le k\le d$ and $1\le i_1<\ldots<i_j\le d$.
 Taking it to the power $2z$ yields
\begin{align*}
\prod_{j=1}^k s_{i_j}\bz\sigma\jz^{1-\alpha} s_{d+1-i_j}\bz\rho\jz^{\alpha}
\le
\prod_{i=1}^k s_i\bz\bz\rho^{\frac{\alpha}{2z}}\sigma^{\frac{1-\alpha}{z}}\rho^{\frac{\alpha}{2z}}\jz^z\jz
\le
\prod_{i=1}^k s_i\bz\sigma\jz^{1-\alpha} s_i\bz\rho\jz^{\alpha}.
\end{align*}
Using that weak log-majorization implies weak majorization (see, e.g.~\cite[Proposition 4.1.6]{Hiai_book}), we obtain
\begin{align*}
\sum_{j=1}^k s_{i_j}\bz\sigma\jz^{1-\alpha} s_{d+1-i_j}\bz\rho\jz^{\alpha}
\le
\sum_{i=1}^k s_i\bz\bz\rho^{\frac{\alpha}{2z}}\sigma^{\frac{1-\alpha}{z}}\rho^{\frac{\alpha}{2z}}\jz^z\jz
\le
\sum_{i=1}^k s_i\bz\sigma\jz^{1-\alpha} s_i\bz\rho\jz^{\alpha}
\end{align*}
for every $1\le k\le d$.
Taking now $k=d$, we see that the middle sum is equal to $Q_{\alpha,z}(\rho\|\sigma)$,
the rightmost sum can be upper bounded by 
$\bz\sum_{i=1}^ds_i(\rho)\jz^{\alpha}\bz\sum_{i=1}^ds_i(\sigma)\jz^{1-\alpha}=\bz\Tr\rho\jz^{\alpha}\bz\Tr\sigma\jz^{1-\alpha}$ when $\alpha\in(0,1)$, 
by the same argument as in \eqref{Renyi pos1},
and similarly, the leftmost sum can be lower bounded by 
$\bz\Tr\rho\jz^{\alpha}\bz\Tr\sigma\jz^{1-\alpha}$
when $\alpha>1$.
\end{proof}

\begin{remark}
The first inequalities in \eqref{nonneg1}--\eqref{nonneg2} for $(\alpha,z)\in K_2\cup K_4\cup K_5\cup K_7$
are immediate from the monotonicity under taking the trace of both $\rho$ and $\sigma$, according to Lemma \ref{lemma:az monotonicity}.
\end{remark}

We can also establish strict positivity of $\what D_{\alpha,z}$ and $T_{\alpha,z}$, except for the region
\begin{align*}
K_0:\ds 0<\alpha<1,\,z<\min\{\alpha,1-\alpha\}.
\end{align*}
In the proof we also give an alternative proof for 
the first inequalities in \eqref{nonneg1}--\eqref{nonneg2} when 
$(\alpha,z)\notin K_0$.

\begin{prop}\label{prop:Renyi pos}
The second inequalities in \eqref{nonneg1}--\eqref{nonneg2} hold as equalities if and only if 
$\Tr\rho=\Tr\sigma$. 
If $(\alpha,z)\in\bz(0,+\infty)\times(0,+\infty]\jz\setminus K_0$ then 
the second inequalities in \eqref{nonneg1}--\eqref{nonneg2} hold as equalities if and only if 
$\rho/\Tr\rho=\sigma/\Tr\sigma$.
\end{prop}
\begin{proof}
The assertion about the second inequalities is trivial from the strict convexity of 
$t\mapsto s(\alpha)t^{\alpha}$ when $\alpha\in(0,+\infty)\setminus\{1\}$. Hence, for the rest we analyze the first inequalities.

It has been pointed out, e.g., in \cite[Proposition 1]{LinTomamichel15}, that the 
Araki-Lieb-Thirring inequality \cite{Araki,LT} implies the monotonicity
\begin{align}\label{mon in z}
Q_{\alpha,z_1}(\rho\|\sigma)=\Tr\bz \rho^{\frac{\alpha}{2z_1}}\sigma^{\frac{1-\alpha}{z_1}}\rho^{\frac{\alpha}{2z_1}}\jz^{\frac{z_1}{z_2}z_2}
\le
\Tr\bz \rho^{\frac{\alpha}{2z_2}}\sigma^{\frac{1-\alpha}{z_2}}\rho^{\frac{\alpha}{2z_2}}\jz^{z_2}
=
Q_{\alpha,z_2}(\rho\|\sigma),\ds\ds\ds z_2\le z_1.
\end{align}
Hence, for $\alpha>1$ we have
\begin{align*}
D_{\alpha,z}\bz\frac{\rho}{\Tr\rho}\Big\|\frac{\sigma}{\Tr\sigma}\jz
=
\what D_{\alpha,z}(\rho\|\sigma)
\ge
\what D_{\alpha,+\infty}(\rho\|\sigma)
=
D_{\alpha,+\infty}\bz\frac{\rho}{\Tr\rho}\Big\|\frac{\sigma}{\Tr\sigma}\jz
\ge 0,
\end{align*}
with equality if and only if $\rho/\Tr\rho=\sigma/\Tr\sigma$, according to 
\cite[Proposition 3.22]{MO-cqconv}. This is exactly the second inequality in 
\eqref{nonneg2} with the equality condition.
If $\alpha\in(0,1/2]$ and $z\ge \alpha$ then 
\begin{align*}
D_{\alpha,z}\bz\frac{\rho}{\Tr\rho}\Big\|\frac{\sigma}{\Tr\sigma}\jz
=
\what D_{\alpha,z}(\rho\|\sigma)
\ge
\what D_{\alpha,\alpha}(\rho\|\sigma)
=
D_{\alpha,\alpha}\bz\frac{\rho}{\Tr\rho}\Big\|\frac{\sigma}{\Tr\sigma}\jz
\ge 0,
\end{align*}
with equality if and only if $\rho/\Tr\rho=\sigma/\Tr\sigma$, according to 
\cite[Theorem 5]{Beigi} (see also \cite[Proposition 3.22]{MO-cqconv}).
If $\alpha\in[1/2,1)$ and $z\ge 1-\alpha$ then 
\begin{align*}
D_{\alpha,z}\bz\frac{\rho}{\Tr\rho}\Big\|\frac{\sigma}{\Tr\sigma}\jz
=
\what D_{\alpha,z}(\rho\|\sigma)
\ge
\what D_{\alpha,1-\alpha}(\rho\|\sigma)
&=
D_{\alpha,1-\alpha}\bz\frac{\rho}{\Tr\rho}\Big\|\frac{\sigma}{\Tr\sigma}\jz\\
&=
\frac{\alpha}{1-\alpha}D_{1-\alpha,1-\alpha}\bz\frac{\rho}{\Tr\rho}\Big\|\frac{\sigma}{\Tr\sigma}\jz
\ge 0,
\end{align*}
with equality if and only if $\rho/\Tr\rho=\sigma/\Tr\sigma$,
by the same argument as above.
\end{proof}

\begin{cor}
Let $\rho,\sigma\in\B(\hil)_+$ and $\alpha\in(0,+\infty)$, $z\in(0,+\infty]$. Then 
\begin{align}
D_{\alpha,z}(\rho\|\sigma)&\ge \log\Tr\rho-\log\Tr\sigma,\label{Renyi nonneg1}\\
\what D_{\alpha,z}(\rho\|\sigma)&\ge 0,\label{Renyi nonneg2}\\
T_{\alpha,z}(\rho\|\sigma)&\ge 0.\label{Tsallis nonneg}
\end{align}
If, moreover, $(\alpha,z)\notin K_0$ then equality holds in 
\eqref{Renyi nonneg1} or \eqref{Renyi nonneg2} if and only if 
$\rho=\lambda\sigma$ for some $\lambda\in(0,+\infty)$, and 
equality holds in \eqref{Tsallis nonneg} if and only if $\rho=\sigma$.
\end{cor}
\begin{proof}
Immediate from Proposition \ref{prop:Renyi pos}.
\end{proof}

\subsection{Further properties of the R\'enyi divergence radii}
\label{sec:minfa}

\begin{lemma}\label{lemma:minfa cont}
For any $\sigma\in\S(\hil)_{++}$ and any $\alpha\in(0,+\infty)\setminus\{1\}$, 
$z\in(0,+\infty)$, the map 
$x\mapsto D_{\alpha,z}(W(x)\|\sigma)$ is bounded on $\X$ for any cq channel $W$, and hence the map
$P\mapsto\sum_{x\in\X}P(x)D_{\alpha,z}(W(x)\|\sigma)$ is continuous on 
$\P_f(\X)$ in the variational norm.
\end{lemma}
\begin{proof}
We prove the case $\alpha>1$; the case $\alpha\in(0,1)$ follows the same way.
If $\sigma\in\S(\hil)_{++}$ then $\sigma\ge\lambda_{\min}(\sigma)I$, and hence
$\sigma^{\frac{1-\alpha}{z}}\le\lambda_{\min}(\sigma)^{\frac{1-\alpha}{z}}I$. Using the monotonicity 
of $A\mapsto\Tr A^{z}$ on $\B(\hil)_+$ for $z>0$, we get 
\begin{align*}
D_{\alpha,z}(W(x)\|\sigma)&\le
\frac{1}{\alpha-1}\log\Tr\bz W(x)^{\frac{\alpha}{2z}}\lambda_{\min}(\sigma)^{\frac{1-\alpha}{z}}IW(x)^{\frac{\alpha}{2z}}\jz^{z}\\
&=
-\log \lambda_{\min}(\sigma)+\frac{1}{\alpha-1}\log\Tr W(x)^{\alpha}
\le
-\log \lambda_{\min}(\sigma),
\end{align*}
proving the boundedness, and the assertion on continuity is immediate from this.
\end{proof}

\begin{cor}\label{cor:minfa usc}
The map $P\mapsto\chi_{\alpha}\nw(W,P)$ is concave and upper semi-continuous on $\P_f(\X)$ in the variational norm.
In particular,
if $P\in\P_f(\X)$ and $P_n\in\P_f(\X)$, $n\in\bN$, are such that $\lim_{n\to+\infty}\norm{P_n-P}_1=0$ then 
\begin{align}\label{minfa usc}
\limsup_{n\to+\infty}\chi_{\alpha}\nw(W,P_n)\le\chi_{\alpha}\nw(W,P).
\end{align}
\end{cor}
\begin{proof}
A combination of \eqref{positive minimization} and Lemma \ref{lemma:minfa cont} shows that
$\chi_{\alpha}\nw$, as the infimum of continuous affine functions, is upper semi-continuous and concave. Upper 
semi-continuity imples \eqref{minfa usc}.
\end{proof}

\section{Optimality of the strong converse exponents}
\label{sec:lower bound}

\textit{Proof of Lemma \ref{lemma:sc lower}:}
Let $\C_n=(\E_n,\D_n)$ be a code for $n$ uses of the channel. 
Define the classical-quantum states
\begin{align*}
R_n:=\frac{1}{|\C_n|}\sum_{k=1}^{|\C_n|}\pr{k}\otimes W^{\otimes n}(\E_n(k)),\ds\ds
S_n:=\frac{1}{|\C_n|}\sum_{k=1}^{|\C_n|}\pr{k}\otimes \sigma^{\otimes n},
\end{align*}
and the POVM element 
\begin{align*}
T_n:=\sum_{k=1}^{|\C_n|}\pr{k}\otimes\D_n(k),
\end{align*}
where $(\pr{k})_{k=1}^{|\C_n|}$ is a set of orthogonal rank $1$ projections in some Hilbert space, and
$\sigma\in\S(\hil)$ is an arbitrary state.
Then we have 
\begin{align*}
\Tr R_nT_n&=\frac{1}{|\C_n|}\sum_{k=1}^{|\C_n|}\Tr  W^{\otimes n}(\E_n(k))\D_n(k)= P_s(W^{\otimes n},\C_n),\\
\Tr S_nT_n&=\frac{1}{|\C_n|}\sum_{k=1}^{|\C_n|}\Tr  \sigma^{\otimes n}\D_n(k)=\frac{1}{|\C_n|}.
\end{align*}
For any $\sigma\in\S(\hil^{\otimes n})$ such that $W^{\otimes n}(\E_n(k))^0\le\sigma^0$ for all $k$, 
and for all $\alpha\in(1,+\infty)$, we get
\begin{align}
P_s(W^{\otimes n},\C_n)^{\alpha}\bz\frac{1}{|\C_n|}\jz^{1-\alpha}
&=
\bz\Tr R_nT_n\jz^{\alpha}\bz\Tr S_n T_n\jz^{1-\alpha}
\le
Q_{\alpha}\nw\bz R_n\|S_n\jz\nonumber\\
&=
\frac{1}{|\C_n|}\sum_{k=1}^{|\C_n|}Q_{\alpha}\nw\bz W^{\otimes n}(\E_n(k))\|\sigma^{\otimes n}\jz\nonumber\\
&=
\frac{1}{|\C_n|}\sum_{k=1}^{|\C_n|}\prod_{x\in\X}Q_{\alpha}\nw(W(x)\|\sigma)^{nP_{\E_n(k)}(x)}\nonumber\\
&=
\frac{1}{|\C_n|}\sum_{k=1}^{|\C_n|}\exp\bz n(\alpha-1)\sum_{x\in\X}P_{\E_n(k)}(x)D_{\alpha}\nw(W(x)\|\sigma)\jz\label{Nagaoka bound2}
\end{align}
where the inequality is due to the monotonicity of the sandwiched R\'enyi divergence for $\alpha>1$.
Note that 
the inequality between the first and the last expressions above holds trivially when 
$W^{\otimes n}(\E_n(k))^0\nleq\sigma^0$ for some $k$.

A simple rearrangement yields that for any $\alpha\in(1,+\infty)$ and any $\sigma\in\S(\hil)$,
\begin{align}\label{Nagaoka bound3}
\frac{1}{n}\log  P_s(W^{\otimes n},\C_n)
&\le
-\frac{\alpha-1}{\alpha}\left[
\frac{1}{n}\log|\C_n|-\max_{1\le k\le|\C_n|}\sum_{x\in\X}P_{\E_n(k)}(x)D_{\alpha}\nw(W(x)\|\sigma)
\right],
\end{align}
where we used that the logarithm function is monotone increasing, and hence quasi-convex, to upper bound the
logarithm of the convex combination in \eqref{Nagaoka bound2} by the logarithm of the maximum of the 
individual terms. A completely similar argument as above, using the monotonicity of the max-relative entropy 
$D_{\infty}\nw$, yields that \eqref{Nagaoka bound3} also holds for $\alpha=+\infty$, with the natural definition
$\frac{\alpha-1}{\alpha}\big\vert_{\alpha=+\infty}:=1$. Finally, the inequality in \eqref{Nagaoka bound3} holds trivially for $\alpha=1$, since then the RHS is zero.

When $\sigma\in\S(\hil)_{++}$, we may use the continuity bound in the proof of Lemma 
\ref{lemma:minfa cont} to obtain
\begin{align}
&\frac{1}{n}\log  P_s(W^{\otimes n},\C_n)\nonumber\\
&\ds\le
-\frac{\alpha-1}{\alpha}\left[
\frac{1}{n}\log|\C_n|-\sum_{x\in\X}P(x)D_{\alpha}\nw(W(x)\|\sigma)
+\max_{1\le k\le|\C_n|}\norm{P_{\E_n(k)}-P}_1\log\lambda_{\min}(\sigma)
\right].\label{Nagaoka bound4}
\end{align}

Now, if $(\C_n)_{n\in\bN}$ is a sequence of codes as in the definition \eqref{sci2} of $\sci(W,R,P)^*$ then
taking the limsup over $n$ in \eqref{Nagaoka bound4}, then the infimum over $\sigma\in\S(\hil)_{++}$, and finally the infimum over $\alpha>1$, yields
\begin{align}
\limsup_{n\to+\infty}\frac{1}{n}\log  P_s(W^{\otimes n},\C_n)
&\le
-\sup_{\alpha>1}\frac{\alpha-1}{\alpha}\left[R-\minfa\nw(W,P)\right],
\end{align}
which is what we wanted to prove.

\begin{rem}
When $\C_n$ is a constant composition code with composition $P_n:=P_{\E_n(k)}$, $k=1,\ldots,|\C_n|$,
the inequality in \eqref{Nagaoka bound3} yields, after taking the infimum over $\sigma\in\S(\hil)$ and 
$\alpha>1$,
\begin{align}\label{sc lower2}
\frac{1}{n}\log  P_s(W^{\otimes n},\C_n)&\le
-\sup_{\alpha>1}\frac{\alpha-1}{\alpha}\left[
\frac{1}{n}\log|\C_n|-\minfa\nw(W,P_n)
\right].
\end{align}
\end{rem}

\begin{rem}\label{rem:Sheverdyaev}
A similar bound to the one in \eqref{Nagaoka bound2} was used by Sheverdyaev \cite{Sheverdyaev1982} to bound the success 
probability of feedback-assisted coding for classical channels.
The inequality in \cite{Sheverdyaev1982} was derived using H\"older's inequality, instead of the monotonicity 
argument above, which is
adapted from Nagaoka's proof \cite{N}. 
We remark that for classical R\'enyi divergences, 
monotonicity under stochastic maps follows immediately from the convexity properties of the power functions, or 
equivalently, from a special case of H\"older's inequality, and vice versa, 
this case of H\"older's inequality can be viewed as a special case of the monotonicity of the R\'enyi divergences, 
hence the two approaches to obtain the upper bound on the success probability are the same on the technical level. 
On the other hand, proving monotonicity of the quantum R\'enyi 
divergences requires more involved techniques, and it is the monotonicity-based approach of Nagaoka that has 
proved to be fruitful in quantum information theory.
\end{rem}
\medskip

\noindent\textit{Proof of Lemma \ref{lemma:sc optimality cost}:}
Let $\C_n$ be a code for $n$ uses of the channel with cost $\cost(\C_n)< \costt$. For any 
$\sigma\in\S(\hil)$, and any message $k=1,\ldots,|\C_n|$,
\begin{align*}
\sum_{x\in\X}P_{\E_n(k)}(x)D_{\alpha}\nw(W(x)\|\sigma)
\le
\sup_{P\in\P_{f,\cost<\costt}(\X)}\sum_{x\in\X}P(x)D_{\alpha}\nw(W(x)\|\sigma),
\end{align*}
and hence the upper bound in \eqref{Nagaoka bound3} can be continued as
\begin{align}\label{cost const sc opt1}
\frac{1}{n}\log P_s(W^{\otimes n},\C_n)
\le
-\frac{\alpha-1}{\alpha}
\left[
\frac{1}{n}\log|\C_n|-\sup_{P\in\P_{f,\cost<\costt}(\X)}\sum_{x\in\X}P(x)D_{\alpha}\nw(W(x)\|\sigma)\right].
\end{align}

Let $(\C_n)_{n\in\bN}$ be a sequence of codes with rate at least $R$, and
$\limsup_n \cost(\C_n)<\costt$. Then 
$\cost(\C_n)<\costt$ for all large enough $n$, and taking the limsup over $n$ in \eqref{cost const sc opt1}, and then the infima over $\sigma\in\S(\hil)$ and $\alpha\in[1,+\infty]$, yields
\begin{align}\label{cost const sc opt2}
\limsup_{n\to+\infty}\frac{1}{n}\log P_s(W^{\otimes n},\C_n)
\le
-\sup_{\alpha\in[1,+\infty]}\frac{\alpha-1}{\alpha}
\left[R-\inf_{\sigma\in\S(\hil)}\sup_{P\in\P_{f,\cost<\costt}(\X)}\sum_{x\in\X}P(x)D_{\alpha}\nw(W(x)\|\sigma)\right].
\end{align}
Let $h_{\alpha}(P,\sigma):=\sum_{x\in\X}P(x)D_{\alpha}\nw(W(x)\|\sigma)$. Then $h$ is clearly concave in its first variable, and by Lemmas \ref{lemma:2nd convexity} and \ref{lemma:lsc}, it is convex and lower semi-continuous in its second variable. Thus, by Lemma \ref{lemma:KF+ minimax}, 
\begin{align*}
\inf_{\sigma\in\S(\hil)}\sup_{P\in\P_{f,\cost<\costt}(\X)}\sum_{x\in\X}P(x)D_{\alpha}\nw(W(x)\|\sigma)
&=
\sup_{P\in\P_{f,\cost<\costt}(\X)}\inf_{\sigma\in\S(\hil)}\sum_{x\in\X}P(x)D_{\alpha}\nw(W(x)\|\sigma)\\
&=
\sup_{P\in\P_{f,\cost<\costt}(\X)}\chi_{\alpha}\nw(W,P),
\end{align*}
and the bound in \eqref{cost const sc opt2} can be rewritten as
\begin{align*}
\liminf_{n\to+\infty}-\frac{1}{n}\log P_s(W^{\otimes n},\C_n)
\ge
\sup_{\alpha\in[1,+\infty]}\frac{\alpha-1}{\alpha}
\left[R-\sup_{P\in\P_{f,\cost<\costt}(\X)}\chi_{\alpha}\nw(W,P)\right].
\end{align*}
Since this holds for every sequence of codes as above, the assertion follows.

\section{Random coding exponent with constant composition}
\label{sec:random coding exponent}

Below we give a slightly different proof of the constant composition random coding bound first proved in 
\cite{ChengHansonDattaHsieh2018}.
We start with the following random coding bound, given in \cite{HN,Hayashicq,Ogawa_randomcoding}.
\begin{lemma}\label{lemma:random coding}
Let $W:\,\X\to\S(\hil)$ be a classical-quantum channel, $n\in\bN$, $R>0$, 
$M_n:=\ceil{e^{nR}}$, 
and
$Q_n\in\P_f(\X^n)$.
For every
$\dvecc{x}=(\vecc{x}_1,\ldots,\vecc{x}_{M_n})\in(\X^n)^{M_n}$, there exists a code
$\C_{n,\dvecc{x}}=(\E_{n,\dvecc{x}},\D_{n,\dvecc{x}})$ for $W^{\otimes n}$ such that 
$\E_{n,\dvecc{x}}(k)=\vecc{x}_k$, $k\in[M_n]$, and
\begin{align}
\Exp_{Q_n^{\otimes M_n}}P_e(W^{\otimes n},\C_{n,\dvecc{x}})
&\le
\sum_{\vecc{x}\in\X^n}Q_n(\vecc{x})\Tr W^{\otimes n}(\vecc{x})
\{ W^{\otimes n}(\vecc{x})-e^{nR}W^{\otimes n}(Q_n)\le 0\}\nonumber\\
&\ds\ds+
e^{nR}\sum_{\vecc{x}\in\X^n}Q_n(\vecc{x})\Tr W^{\otimes n}(Q_n)
\{ W^{\otimes n}(\vecc{x})-e^{nR}W^{\otimes n}(Q_n)>0\}\nonumber\\
&\le
e^{nR(1-\alpha)}\sum_{\vecc{x}\in\X^n}Q_n(\vecc{x})
\Tr W^{\otimes n}(\vecc{x})^{\alpha}W^{\otimes n}(Q_n)^{1-\alpha}.\label{rcbound}
\end{align}
In particular, there exists an $\vecc{x}\in\supp Q_n$ such that $P_e(W^{\otimes n},\C_{n,\dvecc{x}})$
is upper bounded by the RHS of \eqref{rcbound}.
\end{lemma}

From this, we can obtain the following:

\begin{prop}\label{prop:random coding exponent}
Let $W:\,\X\to\S(\hil)$ be a classical-quantum channel, and let $R>0$. For every $n\in\bN$, 
and every type $P_n\in\P_n(\X)$, there exists a code $\C_n$ of constant composition $P_n$
with rate $\frac{1}{n}\log|\C_n|\ge R$ such that 
\begin{align}
\frac{1}{n}\log P_e(W^{\otimes n},\C_{n})
&\le
-\sup_{0\le\alpha\le 1}(\alpha-1)
\left[R-\sum_{x\in\X}P_n(x)D_{\alpha}(W(x)\|W(P_n))+|\supp P_n|\frac{\log (n+1)}{n}\right].\label{cc bound}
\end{align}
\end{prop}

\begin{proof}
Let $\X_n:=\X^n_{P_n}\subseteq\X^n$ be the set of sequences with type $P_n$. Choosing 
$Q_n:=\frac{1}{|\X_n|}\egy_{\X_n}$
in Lemma \ref{lemma:random coding}, we get the existence of codes $\C_n$ with constant composition $P_n$ such that 
\begin{align*}
P_e(W^{\otimes n},\C_{n})\le
e^{nR(1-\alpha)}
\Tr W^{\otimes n}(\vecc{x})^{\alpha}W^{\otimes n}(Q_n)^{1-\alpha}
\end{align*}
for any $\vecc{x}\in \X_n$, 
where we used that $W^{\otimes n}(Q_n)=\frac{1}{|\X_n|}\sum_{\vecc{y}\in \X_n}W^{\otimes n}(\vecc{y})$ is permutation-invariant.
Now we use the well-known facts that $|\X^n_{P_n}|\ge (n+1)^{-|\supp P_n|}e^{nH(P_n)}$, and that for any 
$\vecc{y}\in \X^n_{P_n}$, $P_n^{\otimes n}(\vecc{y})=e^{-nH(P_n)}$,
(see \eqref{type prob} and \eqref{type card}), 
to obtain that 
\begin{align*}
W^{\otimes n}(Q_n)
=
\frac{1}{|\X_n|}\sum_{\vecc{y}\in \X_n}W^{\otimes n}(\vecc{y})
&\le
(n+1)^{|\supp P_n|}e^{-nH(P_n)}\sum_{\vecc{y}\in \X_n}W^{\otimes n}(\vecc{y})\\
&=
(n+1)^{|\supp P_n|}\sum_{\vecc{y}\in \X_n}P_n^{\otimes n}(\vecc{y})W^{\otimes n}(\vecc{y})\\
&\le
(n+1)^{|\supp P_n|}\sum_{\vecc{y}\in \X^n}P_n^{\otimes n}(\vecc{y})W^{\otimes n}(\vecc{y})\\
&=
(n+1)^{|\supp P_n|}W(P_n)^{\otimes n}.
\end{align*}
Using that $t\mapsto t^{1-\alpha}$ is operator monotone on $\bR_+$ for $\alpha\in[0,1]$, we get that 
\begin{align*}
P_e(W^{\otimes n},\C_{n})\le
(n+1)^{(1-\alpha)|\supp P_n|}e^{nR(1-\alpha)}\Tr W^{\otimes n}(\vecc{x})^{\alpha}\bz W(P_n)^{\otimes n}\jz^{1-\alpha}.
\end{align*}
From this \eqref{cc bound} follows by simple algebra.
\end{proof}

\begin{cor}\label{cor:random coding exponent}
Let $W:\,\X\to\S(\hil)$ be a classical-quantum channel, let $R>0$, and $P\in\P_f(\X)$. 
For any sequence $P_n\in\P_n(\X)$, $n\in\bN$, such that 
$\cup_{n\in\bN}\supp P_n$ is finite, and $\lim_{n\to+\infty}\norm{P_n-P}_1=0$, there exists 
a sequence of codes $(\C_n)_{n\in\bN}$, where 
all $\C_n$ are of constant composition $P_n$,
and every $\C_n$ has rate 
$\frac{1}{n}\log|\C_n|\ge R$, such that 
\begin{align}
\limsup_{n\to+\infty}\frac{1}{n}\log P_e(W^{\otimes n},\C_{n})
&\le
-\sup_{0\le\alpha\le 1}(\alpha-1)
\left[R-\sum_{x\in\X}P(x)D_{\alpha}(W(x)\|W(P))\right].\label{cc exponent}
\end{align}
If, moreover, $R<\sum_{x\in\X}P(x)D(W(x)\|W(P))$ then $P_e(W^{\otimes n},\C_{n})$
goes to zero exponentially fast.
\end{cor}
\begin{proof}
Proposition \ref{prop:random coding exponent} 
yields the existence of a sequence of codes $(\C_n)_{n\in\bN}$ with all the desired properties, except maybe for the upper bound \eqref{cc exponent}, such that 
\begin{align}
\frac{1}{n}\log P_e(W^{\otimes n},\C_{n})
&\le
-\left[R-\sum_{x\in\X}P_n(x)D_{\alpha}(W(x)\|W(P_n))+|\supp P_n|\frac{\log (n+1)}{n}\right]
\end{align}
for every $n\in\bN$ and $\alpha\in[0,1]$. Taking first the limsup in $n$, and then the infimum in $\alpha\in[0,1]$, the assertion follows.
\end{proof}
\medskip

A variant of Corollary \ref{cor:random coding exponent} was given by Csisz\'ar and K\"orner in \cite[Theorem 10.2]{CsiszarKorner2} for classical channels, where the RHS is
\begin{align*}
-\sup_{\new{1/2}<\alpha\le 1}\frac{\alpha-1}{\alpha}\left[R-\chi_{\alpha}(W,P)\right].
\end{align*}
Note that $\frac{\alpha-1}{\alpha}$ gives a strictly better prefactor, while 
$\chi_{\alpha}(W,P)\le\sum_{x\in\X}P(x)D_{\alpha}(W(x)\|W(P))$. 
However, the Csisz\'ar-K\"orner bound is optimal for high enough rates, and hence it would be desirable to obtain an exact analogue of it for classical-quantum channels.
\medskip

Constant composition exponents were obtained also for classical-quantum channels before; for instance,
the following was stated in \cite{universalcq}:
Let $\X$ be a finite set, $\hil$ be a finite-dimensional Hilbert space, 
let $P$ be a probability mass function on $\X$, and $R>0$.
Then there exists a sequence of codes $(\C_n)_{n\in\bN}$ such that 
\begin{align*}
\lim_{n\to+\infty}\frac{1}{n}\log|\C_n|=R, 
\end{align*}
and for any classical-quantum channel $W:\,\X\to\S(\hil)$,
\begin{align}\label{Hayashi universal exponent}
\liminf_{n\to+\infty}-\frac{1}{n}\log  P_e(W^{\otimes n},\C_n)
\ge
\sup_{0\le\alpha\le 1}\frac{\alpha-1}{2-\alpha}\left[R-I_{\alpha}(W,P)\right].
\end{align}
While this is not sufficiently detailed in \cite{universalcq}, the codes above can indeed be chosen to be of 
constant composition.

Note that the bound in \eqref{Hayashi universal exponent} is not as strong as the classical universal random coding exponent given by Csisz\'ar and K\"orner,
as $0<\alpha<2-\alpha$ for all $\alpha<1$, and $\chi_{\alpha}(W,p)\ge I_{\alpha}(W,p)$,
with the inequality being strict in general.

\section{Evaluation of the information spectrum quantity}
\label{sec:infospectrum}

For every $n\in\bN$, let $\hil_n$ be a finite-dimensional Hilbert space, and let 
$\rho_n,\sigma_n\in\S(\hil_n)$. Let 
\begin{align*}
\psi(\alpha):=\begin{cases}
\limsup_{n\to+\infty}\frac{1}{n}\log\Tr\rho_n^{\alpha}\sigma_n^{1-\alpha},&\alpha\in(0,1],\\
\limsup_{n\to+\infty}\frac{1}{n}\log\Tr\bz\rho_n^{1/2}\sigma_n^{\frac{1-\alpha}{\alpha}}\rho_n^{1/2}\jz^{\alpha},&\alpha\in(1,+\infty).
\end{cases}
\end{align*}

The following statement is a central observation in the information spectrum method, and its proof follows by 
standard arguments. We include a complete proof for the readers' convenience.

\begin{lemma}\label{lemma:infospectrum}
In the above setting, 
\begin{align}\label{infospectrum direct}
\lim_{n\to+\infty}\Tr(\rho_n-e^{nr}\sigma_n)_+=1\ds
\begin{cases}
\text{for all }r\in\bR,&\text{ if }\psi(1)<0,\\
\text{for }r<\derleft{\psi}(1),&\text{ if }\psi(1)=0.
\end{cases}
\end{align}
If $\rho_n^0\le\sigma_n^0$ for all large enough $n$, then $\psi(1)=0$, and
\begin{align}\label{infospectrum converse}
\lim_{n\to+\infty}\Tr(\rho_n-e^{nr}\sigma_n)_+=
0 \ds\ds\text{for all }\ds r>\derright{\psi}(1).
\end{align}
\end{lemma}

\begin{proof}
Let $S_{n,r}:=\{\rho_n-e^{nr}\sigma_n>0\}$ be the Neyman-Pearson test with parameter $r$. Then 
\begin{align*}
0\le \Tr(\rho_n-e^{nr}\sigma_n)_+=\Tr\rho_n S_{n,r}-e^{nr}\Tr\sigma_n S_{n,r}
=
1-\alpha_n(S_{n,r})-e^{nr}\beta_n(S_{n,r}),
\end{align*}
with $\alpha_n(S_{n,r}):=\Tr\rho_n(I-S_{n,r})$ and $\beta_n(S_{n,r}):=\Tr\sigma_nS_{n,r}$ being the type I and the type II error probabilities 
corresponding to the test $S_{n,r}$, respectively. In particular,
\begin{align}\label{infospectrum proof3}
0\le e^{nr}\Tr\sigma_n S_{n,r}\le\Tr\rho_n S_{n,r}.
\end{align}

By Audenaert's inequality \cite[Theorem 1]{Hoeffding1},
\begin{align*}
e_n(r):=\alpha_n(S_{n,r})+e^{nr}\beta_n(S_{n,r})=\Tr\rho_n (I-S_{n,r})+e^{nr}\Tr\sigma_n S_{n,r}
\le
e^{nr(1-\alpha)}\Tr\rho_n^{\alpha}\sigma_n^{1-\alpha}
\end{align*}
for all $\alpha\in[0,1]$, and hence
\begin{align}\label{infospectrum proof}
\limsup_{n\to+\infty}\frac{1}{n}\log e_n(r)
&\le
-\sup_{0\le \alpha\le 1}(\alpha-1)\{r-\psi(\alpha)/(\alpha-1)\}.
\end{align}
Since $\lim_{\alpha\nearrow 1}\psi(\alpha)/(\alpha-1)=\derleft{\psi}(1)$ if $\psi(1)=0$, and $+\infty$ otherwise,
we obtain that $e_n(r)\to 0$ as $n\to+\infty$ 
for any $r\in\bR$ when $\psi(1)<0$ and for $r<\derleft{\psi}(1)$ when $\psi(1)=0$. 
Using that $\max\{\alpha_n(S_{n,r}),e^{nr}\beta_n(S_{n,r})\}\le e_n(r)$, we see that 
in these cases, also $e^{nr}\beta_n(S_{n,r})\to 0$ and $1-\alpha_n(S_{n,r})\to 1$, and therefore
$\Tr(\rho_n-e^{nr}\sigma_n)_+\to 1$.
This proves \eqref{infospectrum direct}.

Next, we prove \eqref{infospectrum converse}.
By the monotonicity of the sandwiched R\'enyi divergences, we have, for every $0\le T_n\le I$,
and every $\alpha>1$,
\begin{align*}
\Tr\bz\rho_n^{1/2}\sigma_n^{\frac{1-\alpha}{\alpha}}\rho_n^{1/2}\jz^{\alpha}
&\ge
\bz\Tr\rho_n T_n\jz^{\alpha}
\bz\Tr\sigma_n T_n\jz^{1-\alpha},
\end{align*}
and therefore
\begin{align}\label{infospectrum proof2}
\limsup_{n\to+\infty}\frac{1}{n}\log \Tr\rho_n T_n
&\le
\inf_{\alpha>1}\frac{\alpha-1}{\alpha}
\left\{
\limsup_{n\to+\infty}\frac{1}{n}\log \Tr\sigma_n T_n+\frac{\psi(\alpha)}{\alpha-1}
\right\}.
\end{align}
Now, let $T_n:=S_{n,r}$ with some $r$. Then, by 
\eqref{infospectrum proof}, we have
\begin{align*}
\limsup_{n\to+\infty}\frac{1}{n}\log \Tr\sigma_n T_n&\le
-\sup_{0\le\alpha\le 1}\{\alpha r-\psi(\alpha)\}\le-r.
\end{align*}
Hence, if $r>\derright{\psi}(1)$, then the RHS of \eqref{infospectrum proof2} is negative, and thus
$\Tr\rho_n S_{n,r}\to 0$ as $n\to+\infty$. Using \eqref{infospectrum proof3}, we get that 
also $e^{nr}\Tr\sigma_nS_{n,r}\to 0$, and hence
$\Tr(\rho_n-e^{nr}\sigma_n)_+\to 0$, as required.
\end{proof}

\begin{cor}\label{cor:infospec limit}
For every $n\in\bN$, let $P_n$ be an $n$-type and $\vecc{x}^{(n)}\in\X^n$ be of type $P_n$. Assume that 
$(P_n)_{n\in\bN}$ converges to some $P\in\P_f(\X)$ and 
$\cup_{n\in\bN}\supp P_n$ is finite. If 
$V(x)^0\le W(x)^0$ for all $x\in\supp P_n$ and all $n$ large enough, then 
\begin{align*}
\lim_{n\to+\infty}\Tr\bz V^{\otimes n}(\vecc{x}^{(n)})-e^{nr}W^{\otimes n}(\vecc{x}^{(n)})\jz_+
=
\begin{cases}
1,& r<\sum_{x\in\X}P(x)D(V(x)\|W(x)),\\
0,& r>\sum_{x\in\X}P(x)D(V(x)\|W(x)).
\end{cases}
\end{align*}
\end{cor}
\begin{proof}
We use Lemma \ref{lemma:infospectrum}
with $\rho_n:=V^{\otimes n}(\vecc{x}^{(n)})$, $\sigma_n:=W^{\otimes n}(\vecc{x}^{(n)})$. Then we have
\begin{align*}
\psi_n(\alpha)&:=\frac{1}{n}\log\Tr \rho_n^{\alpha}\sigma_n^{1-\alpha}=
\sum_{x\in\X}P_n(x)\log\Tr V(x)^{\alpha}W(x)^{1-\alpha},
\end{align*}
for $\alpha\in(0,1]$, and
\begin{align*}
\psi_n(\alpha)&:=\frac{1}{n}\log\Tr \bz\rho_n^{1/2}\sigma_n^{\frac{1-\alpha}{\alpha}}\rho_n^{1/2}\jz^{\alpha}=
\sum_{x\in\X}P_n(x)\log\Tr \bz V(x)^{1/2}W(x)^{\frac{1-\alpha}{\alpha}}V(x)^{1/2}\jz^{\alpha},
\end{align*}
for $\alpha>1$. Hence, $\psi(\alpha)=\lim_{n\to+\infty}\psi_n(\alpha)$ exists as a limit, and 
\begin{align*}
\psi(\alpha)
=
\begin{cases}
\sum_{x\in\X}P(x)\log\Tr V(x)^{\alpha}W(x)^{1-\alpha},&\alpha\in(0,1],\\
\sum_{x\in\X}P(x)\log\Tr \bz V(x)^{1/2}W(x)^{\frac{1-\alpha}{\alpha}}V(x)^{1/2}\jz^{\alpha},&\alpha>1.
\end{cases}
\end{align*}
By assumption, $\psi(1)=0$, and 
\begin{align*}
\derleft{\psi}(1)&=
\sum_{x\in\X}P(x)\lim_{\alpha\nearrow 1}D_{\alpha}\old(V(x)\|W(x))
=\sum_{x\in\X}P(x)D(V(x)\|W(x),\\
\derright{\psi}(1)&=
\sum_{x\in\X}P(x)\lim_{\alpha\searrow 1}D_{\alpha}\nw(V(x)\|W(x))
=\sum_{x\in\X}P(x)D(V(x)\|W(x).
\end{align*}
Hence, the assertion follows immediately from Lemma \ref{lemma:infospectrum}.
\end{proof}

%
%

\section*{Acknowledgments}
We are grateful to Fumio Hiai, P\'eter Vrana, Andreas Winter, Hao-Chung Heng, Min-Hsiu Hsieh, Marco Tomamichel, J\'ozsef Pitrik and D\'aniel Virosztek for discussions. 
This work was partially funded by the JSPS grant no.~JP16K00012 (TO), by the
National Research, Development and 
Innovation Office of Hungary via the research grants K124152, KH129601, the 
Quantum Technology National Excellence Program, Project Nr.~2017-1.2.1-NKP-2017-00001, 
and by a Bolyai J\'anos Fellowship of the Hungarian Academy of Sciences
(MM). We are grateful to anonymous referees for their comments that helped to improve the presentation 
of our results, and particularly indebted to one of the referees for pointing out that our main result 
for constant composition codes can also be used to determine the strong converse exponent with cost
constraint.

\bibliography{bibliography-cqconv-cc}

\begin{thebibliography}{10}

\bibitem{AC2011}
Martial Agueh and Guillaume Carlier.
\newblock Barycenters in the {W}asserstein space.
\newblock {\em {SIAM} Journal on Mathematical Analysis}, 43(2):904--924, 2011.

\bibitem{An}
Shun-ichi Amari and Hiroshi Nagaoka.
\newblock {\em Methods of information geometry}, volume 191 of {\em
  Translations of Mathematical Monographs}.
\newblock American Mathematical Society, Providence, RI, 2000.

\bibitem{AH}
Tsuyoshi Ando and Fumio Hiai.
\newblock Operator log-convex functions and operator means.
\newblock {\em Mathematische Annalen}, 350:611--630, 2011.

\bibitem{Araki}
H.~Araki.
\newblock On an inequality of {Lieb} and {Thirring}.
\newblock {\em Letters in Mathematical Physics}, 19:167--170, 1990.

\bibitem{A73}
Suguru Arimoto.
\newblock On the converse to the coding theorem for discrete memoryless
  channels.
\newblock {\em IEEE Transactions on Information Theory}, 19:357--359, May 1973.

\bibitem{Hoeffding1}
K.M.R. Audenaert, J.~Calsamiglia, R.~Mu\ noz Tapia, E.~Bagan, Ll. Masanes,
  A.~Acin, and F.~Verstraete.
\newblock Discriminating states: The quantum {Chernoff} bound.
\newblock {\em Physical Review Letters}, 98:160501, 2007.
\newblock arXiv:quant-ph/0610027.

\bibitem{AD}
Koenraad M.~R. Audenaert and Nilanjana Datta.
\newblock $\alpha$-$z$-relative {R}enyi entropies.
\newblock {\em J.~Math.~Phys.}, 56:022202, 2015.
\newblock arXiv:1310.7178.

\bibitem{AH2019}
Koenraad M.~R. Audenaert and Fumio Hiai.
\newblock Reciprocal {L}ie-{T}rotter formula.
\newblock {\em Linear and Multilinear Algebra}, 64(6):1220--1235, 2016.

\bibitem{Augustin}
U.~Augustin.
\newblock {\em Noisy channels}.
\newblock Habilitation thesis, Univ. of Erlangen-Numberg, 1978.

\bibitem{Beigi}
Salman Beigi.
\newblock Sandwiched {R\'enyi} divergence satisfies data processing inequality.
\newblock {\em Journal of Mathematical Physics}, 54(12):122202, December 2013.
\newblock arXiv:1306.5920.

\bibitem{Bhatia}
Rajendra Bhatia.
\newblock {\em Matrix Analysis}.
\newblock Number 169 in Graduate Texts in Mathematics. Springer, 1997.

\bibitem{BGJ2019}
Rajendra Bhatia, Stephane Gaubert, and Tanvi Jain.
\newblock Matrix versions of the {H}ellinger distance.
\newblock {\em Letters in Mathematical Physics}, 2019.

\bibitem{BJL2018}
Rajendra Bhatia, Tanvi Jain, and Yongdo Lim.
\newblock On the {B}ures-{W}asserstein distance between positive definite
  matrices.
\newblock {\em Expositiones Mathematicae}, 2018.

\bibitem{BH1998}
M.~V. Burnashev and A.~S. Holevo.
\newblock On the reliability function for a quantum communication channel.
\newblock {\em Problems of Information Transmission}, 34(2):97--107, 1998.

\bibitem{CFL}
E.~A. Carlen, R.~L. Frank, and E.~H. Lieb.
\newblock Some operator and trace function convexity theorems.
\newblock {\em Linear Algebra Appl.}, 490:174--185, 2016.

\bibitem{Cheng-Li-Hsieh2018}
Hao-Chung Cheng, Li~Gao, and Min-Hsiu Hsieh.
\newblock Properties of noncommutative {R}enyi and {A}ugustin information.
\newblock arXiv:1811.04218, 2018.

\bibitem{ChengHansonDattaHsieh2018}
Hao-Chung Cheng, Eric~P. Hanson, Nilanjana Datta, and Min-Hsiu Hsieh.
\newblock Duality between source coding with quantum side information and c-q
  channel coding.
\newblock arXiv:1809.11143, 2018.

\bibitem{ChengHsiehTomamichel2019}
Hao-Chung Cheng, Min-Hsiu Hsieh, and Marco Tomamichel.
\newblock Quantum sphere-packing bounds with polynomial prefactors.
\newblock {\em IEEE Transactions on Information Theory}, 65(5):2872--2898,
  2019.
\newblock arXiv:1704.05703v2.

\bibitem{Csiszar}
Imre Csisz\'ar.
\newblock Generalized cutoff rates and {R\'enyi's} information measures.
\newblock {\em IEEE Transactions on Information Theory}, 41(1):26--34, January
  1995.

\bibitem{CsiszarKorner}
Imre Csisz\'ar and J\'anos K\"orner.
\newblock {\em Information theory: coding theorems for discrete memoryless
  channels}.
\newblock Akad\'emiai Kiad\'o, Budapest, 1981.

\bibitem{CsiszarKorner2}
Imre Csisz\'ar and J\'anos K\"orner.
\newblock {\em Information theory: coding theorems for discrete memoryless
  channels, 2nd edition}.
\newblock Cambridge University Press, 2011.

\bibitem{DW2015}
Marco Dalai and Andreas Winter.
\newblock Constant compositions in the sphere packing bound for
  classical-quantum channels.
\newblock {\em IEEE Transactions on Information Theory}, 63(9):5603 -- 5617,
  2017.
\newblock arXiv:1509.00715.

\bibitem{Datta}
Nilanjana Datta.
\newblock Min- and max-relative entropies and a new entanglement monotone.
\newblock {\em IEEE Transactions on Information Theory}, 55(6):2816--2826,
  2009.

\bibitem{DK}
G.~Dueck and J.~K\"orner.
\newblock Reliability function of a discrete memoryless channel at rates above
  capacity.
\newblock {\em IEEE Transactions on Information Theory}, 25:82--85, 1979.

\bibitem{Fan}
Ky~Fan.
\newblock Minimax theorems.
\newblock {\em Proc Natl Acad Sci USA.}, 39(1):42--47, 1953.

\bibitem{FarkasRevesz2006}
B\'alint Farkas and Szil\'ard R\'ev\'esz.
\newblock Potential theoretic approach to rendezvous numbers.
\newblock {\em Monatshefte f\"ur Mathematik}, 148(4):309--331, 2006.

\bibitem{FL}
Rupert~L. Frank and Elliott~H. Lieb.
\newblock Monotonicity of a relative {R\'enyi} entropy.
\newblock {\em Journal of Mathematical Physics}, 54(12):122201, December 2013.
\newblock arXiv:1306.5358.

\bibitem{Har1968}
E.~A. Haroutunian.
\newblock Bounds to the error probability exponent for semicontinuous
  memoryless channels.
\newblock {\em Prob. Peredach. Inform.}, 4:37--48, 1968.
\newblock \textit{Problems Inform. Transmission}, 4:4 (1968), 29?39.

\bibitem{H:pinching}
Masahito Hayashi.
\newblock Optimal sequence of {POVM's} in the sense of {Stein's} lemma in
  quantum hypothesis testing.
\newblock {\em J.~Phys.~A: Math.~Gen.}, 35:10759--10773, 2002.

\bibitem{Hayashicq}
Masahito Hayashi.
\newblock Error exponent in asymmetric quantum hypothesis testing and its
  application to classical-quantum channel coding.
\newblock {\em Physical Review A}, 76(6):062301, December 2007.
\newblock arXiv:quant-ph/0611013.

\bibitem{universalcq}
Masahito Hayashi.
\newblock Universal coding for classical-quantum channel.
\newblock {\em Commun.~Math.~Phys.}, 289(3):1087--1098, May 2009.

\bibitem{Hayashibook2}
Masahito Hayashi.
\newblock {\em Quantum Information Theory: Mathematical Foundation, 2nd ed.}
\newblock Graduate Texts in Physics. Springer, 2017.

\bibitem{HN}
Masahito Hayashi and Hiroshi Nagaoka.
\newblock General formulas for capacity of classical-quantum channels.
\newblock {\em IEEE Transactions on Information Theory}, 49(7):1753--1768, July
  2003.
\newblock arXiv:quant-ph/0206186.

\bibitem{HT14}
Masahito Hayashi and Marco Tomamichel.
\newblock Correlation detection and an operational interpretation of the
  {R}\'enyi mutual information.
\newblock {\em Journal of Mathematical Physics}, 57:102201, 2016.

\bibitem{Hiai_book}
F.~Hiai.
\newblock Matrix analysis: Matrix monotone functions, matrix means, and
  majorization.
\newblock {\em Interdisciplinary Information Sciences}, 16:139--248, 2010.

\bibitem{Hi3}
F.~Hiai.
\newblock Concavity of certain matrix trace and norm functions.
\newblock {\em Linear Algebra Appl.}, 439:1568--1589, 2013.

\bibitem{HiaiMosonyi2017}
F.~Hiai and M.~Mosonyi.
\newblock Different quantum $f$-divergences and the reversibility of quantum
  operations.
\newblock {\em Rev.~Math.~Phys.}, 29, 2017.

\bibitem{HP_GT}
Fumio Hiai and D\'enes Petz.
\newblock The {Golden-Thompson} trace inequality is complemented.
\newblock {\em Linear Algebra Appl.}, 181:153--185, 1993.

\bibitem{H}
Alexander~S. Holevo.
\newblock The capacity of the quantum channel with general signal states.
\newblock {\em IEEE Transactions on Information Theory}, 44(1):269--273,
  January 1998.

\bibitem{JOPP}
V.~Jaksic, Y.~Ogata, Y.~Pautrat, and C.-A. Pillet.
\newblock Entropic fluctuations in quantum statistical mechanics. an
  introduction.
\newblock In {\em Quantum Theory from Small to Large Scales, August 2010},
  volume~95 of {\em Lecture Notes of the Les Houches Summer School}. Oxford
  University Press, 2012.

\bibitem{Kneser}
Hellmuth Kneser.
\newblock Sur un t\'eor\`eme fondamental de la th\'eorie des jeux.
\newblock {\em C. R. Acad. Sci. Paris}, 234:2418--2420, 1952.

\bibitem{KW}
Robert Koenig and Stephanie Wehner.
\newblock A strong converse for classical channel coding using entangled
  inputs.
\newblock {\em Physical Review Letters}, 103(7):070504, August 2009.
\newblock arXiv:0903.2838.

\bibitem{KRS}
Robert K\"onig, Renato Renner, and Christian Schaffner.
\newblock The operational meaning of min- and max-entropy.
\newblock {\em IEEE Transactions on Information Theory}, 55(9):4337--4347,
  2009.

\bibitem{KA}
F.~Kubo and T.~Ando.
\newblock Means of positive linear operators.
\newblock {\em Math. Ann.}, 246:205--224, 1980.

\bibitem{LT}
E.H. Lieb and W.~Thirring.
\newblock {\em Studies in mathematical physics}.
\newblock University Press, Princeton, 1976.

\bibitem{LimPalfia2012}
Yongdo Lim and Mikl\'os P\'alfia.
\newblock Matrix power means and the {K}archer mean.
\newblock {\em Journal of Functional Analysis}, 262(4):1498--1514, 2012.

\bibitem{LinTomamichel15}
Mingyan~Simon Lin and Marco Tomamichel.
\newblock Investigating properties of a family of quantum renyi divergences.
\newblock {\em Quantum Information Processing}, 14(4):1501--1512, 2015.

\bibitem{Ma}
K.~Matsumoto.
\newblock A new quantum version of $f$-divergence.
\newblock In {\em Nagoya Winter Workshop 2015: Reality and Measurement in
  Algebraic Quantum Theory}, pages 229--273, 2018.
\newblock arXiv:1311.4722.

\bibitem{MH}
Mil\'an Mosonyi and Fumio Hiai.
\newblock On the quantum {R\'enyi} relative entropies and related capacity
  formulas.
\newblock {\em IEEE Transactions on Information Theory}, 57(4):2474--2487,
  April 2011.

\bibitem{MO-cqconv}
Mil\'an Mosonyi and Tomohiro Ogawa.
\newblock Strong converse exponent for classical-quantum channel coding.
\newblock {\em Communications in Mathematical Physics}, 355(1):373--426, 2017.

\bibitem{Renyi_new}
Martin {M\"uller}-Lennert, Fr\'ed\'eric Dupuis, Oleg Szehr, Serge Fehr, and
  Marco Tomamichel.
\newblock On quantum {R\'enyi} entropies: A new generalization and some
  properties.
\newblock {\em Journal of Mathematical Physics}, 54(12):122203, December 2013.
\newblock arXiv:1306.3142.

\bibitem{N}
Hiroshi Nagaoka.
\newblock Strong converse theorems in quantum information theory.
\newblock {\em Proceedings of ERATO Workshop on Quantum Information Science},
  page~33, 2001.
\newblock Also appeared in Asymptotic Theory of Quantum Statistical Inference,
  ed. M. Hayashi, World Scientific, 2005.

\bibitem{NakibogluRenyi}
Bari{\c s} Nakibo{\u g}lu.
\newblock The {R}enyi capacity and center.
\newblock arXiv:1608.02424, 2016.

\bibitem{NakibogluAugustin}
Bari{\c s} Nakibo{\u g}lu.
\newblock The {A}ugustin capacity and center.
\newblock {\em Problems of Information Transmission}, 55:299--342, 2019.
\newblock arXiv:1803.07937.

\bibitem{Ogawa_randomcoding}
Tomohiro Ogawa.
\newblock An information-spectrum approach to hash properties for quantum
  states and relation to channel coding.
\newblock In {\em The 28th Symposium on Information Theory and its Applications
  (SITA2015)}, November 2015.

\bibitem{ON99}
Tomohiro Ogawa and Hiroshi Nagaoka.
\newblock Strong converse to the quantum channel coding theorem.
\newblock {\em IEEE Transactions on Information Theory}, 45(7):2486--2489,
  November 1999.
\newblock arXiv:quant-ph/9808063.

\bibitem{OPW}
Masanori Ohya, D\'enes Petz, and Noboru Watanabe.
\newblock On capacities of quantum channels.
\newblock {\em Probability and Mathematical Statistics}, 17:179--196, 1997.

\bibitem{P86}
D\'enes Petz.
\newblock Quasi-entropies for finite quantum systems.
\newblock {\em Reports in Mathematical Physics}, 23:57--65, 1986.

\bibitem{PV2019}
J\'ozsef Pitrik and D\'aniel Virosztek.
\newblock Quantum {H}ellinger distances revisited.
\newblock arXiv:1903.10455, 2019.

\bibitem{RennerPhD}
Renato Renner.
\newblock {\em Security of Quantum Key Distribution}.
\newblock PhD thesis, Swiss Federal Institute of Technology Zurich, 2005.
\newblock Diss.~ETH No.~16242.

\bibitem{Renyi}
Alfr\'ed R\'enyi.
\newblock On measures of entropy and information.
\newblock In {\em Proc.~4th Berkeley Sympos.~Math.~Statist.~and Prob.},
  volume~I, pages 547--561. Univ. California Press, Berkeley, California, 1961.

\bibitem{SW}
Benjamin Schumacher and Michael Westmoreland.
\newblock Sending classical information via noisy quantum channels.
\newblock {\em Physical Review A}, 56(1):131--138, July 1997.

\bibitem{SW2}
Benjamin Schumacher and Michael~D. Westmoreland.
\newblock Optimal signal ensembles.
\newblock {\em Physical Review A}, 63:022308, 2001.

\bibitem{Shannon}
C.E. Shannon.
\newblock A mathematical theory of communication.
\newblock {\em The Bell System Technical Journal}, 27:379--423, 623/656, 1948.

\bibitem{Sheverdyaev1982}
A.~Yu. Sheverdyaev.
\newblock Lower bound for error probability in a discrete memoryless channel
  with feedback.
\newblock {\em Probl. Peredachi Inf.}, 18(4), 1982.

\bibitem{Sibson}
R.~Sibson.
\newblock Information radius.
\newblock {\em Z.~Wahrscheinlichkeitsth.~Verw.~Gebiete}, 14:149--161, 1969.

\bibitem{TomamichelBook}
M.~Tomamichel.
\newblock {\em Quantum Information Processing with Finite Resources}, volume~5
  of {\em Mathematical Foundations, SpringerBriefs in Math. Phys.}
\newblock Springer, 2016.

\bibitem{TomamichelTan2015}
Marco Tomamichel and Vincent~Y.F. Tan.
\newblock Second-order asymptotics for the classical capacity of image-additive
  quantum channels.
\newblock {\em Communications in Mathematical Physics}, 338:103--137, 2015.

\bibitem{Uhlmann1973}
A.~Uhlmann.
\newblock Endlich dimensionale dichtmatrizen, ii.
\newblock {\em Wiss. Z. Karl-Marx-University Leipzig}, 22:139--177, 1973.

\bibitem{Umegaki}
H.~Umegaki.
\newblock Conditional expectation in an operator algebra.
\newblock {\em Kodai Math.~Sem.~Rep.}, 14:59--85, 1962.

\bibitem{WWY}
Mark~M. Wilde, Andreas Winter, and Dong Yang.
\newblock Strong converse for the classical capacity of entanglement-breaking
  and {Hadamard} channels via a sandwiched {R\'enyi} relative entropy.
\newblock {\em Communications in Mathematical Physics}, 331(2):593--622,
  October 2014.
\newblock arXiv:1306.1586.

\bibitem{W}
Andreas Winter.
\newblock Coding theorem and strong converse for quantum channels.
\newblock {\em IEEE Transactions on Information Theory}, 45(7):2481Đ2485, 1999.

\bibitem{YKL}
H.~P. Yuen, R.~S. Kennedy, and M.~Lax.
\newblock Optimum testing of multiple hypothesis in quantum detection theory.
\newblock {\em IEEE Trans. Inf. Theory}, 21(2):125--134, 1975.

\bibitem{Zhang2018}
Haonan Zhang.
\newblock {C}arlen-{F}rank-{L}ieb conjecture and monotonicity of $\alpha$-$z$
  {R}\'enyi relative entropy.
\newblock arXiv:1811.01205, 2018.

\end{thebibliography}
\end{document}